\documentclass[11pt]{article}
\pdfoutput=1

\usepackage{natbib}
\usepackage{mathtools}
\mathtoolsset{showonlyrefs=true}
\usepackage{amssymb,amsmath,amsthm,latexsym,amsfonts,amscd,dsfont,enumerate}

\usepackage[]{hyperref}
\hypersetup{
	colorlinks=true, 
	breaklinks=true, 
	urlcolor= blue, 
	linkcolor= blue, 
	bookmarksopen=true, 
	pdfauthor={Bichuch, M., Capponi, A., Sturm, S.}
}
\usepackage{graphicx, xcolor}
\usepackage{bm}
\usepackage[T1]{fontenc}
\usepackage{lmodern}
\usepackage{fixltx2e} 
\usepackage[top=1.2in, left=0.9in, bottom=1.2in, right=0.9in]{geometry}
\newcommand{\sell}{+}
\newcommand{\buy}{-}
\newcommand{\buysell}{\pm}

\usepackage{tikz}
\usetikzlibrary{arrows,positioning}
\tikzset{
    >=stealth',
    punkt/.style={
           rectangle,
           rounded corners,
           draw=black, very thick,
           text width=6.5em,
           minimum height=2em,
           text centered},
    pil/.style={
           ->,
           thick,
           shorten <=2pt,
        shorten >=2pt,
   }
}

\newcommand{\G}{\mathcal{G}}
\newcommand{\gt}{\mathcal{G}_t}

\newcommand{\R}{\mathbb{R}}
\newcommand{\Rplus}{\mathbb R_{>0}}
\newcommand{\Px}{\mathbb{P}}
\newcommand{\Exx}{\mathbb{E}}

\newcommand{\Qxx}{\mathbb{Q}}
\newcommand{\Q}{\Qxx}
\newcommand{\Gx}{\mathbb{G}}

\newcommand{\uu}{{\bar{{u}}}}

\newcommand{\vv}{{\bar{{v}}}}

\def\ind{{\mathchoice{1\mskip-4mu\mathrm l}{1\mskip-4mu\mathrm l}
{1\mskip-4.5mu\mathrm l}{1\mskip-5mu\mathrm l}}}
\newcommand{\abs}[1]{ \left \vert #1 \right \vert}

\newcommand{\psir}{\psi^r}
\newcommand{\rrp}{r_r^+}
\newcommand{\rrm}{r_r^-}
\newcommand{\rcp}{r_c^+}
\newcommand{\rcm}{r_c^-}
\newcommand{\rfp}{r_f^+}
\newcommand{\rfm}{r_f^-}

\newcommand{\XVA}{\mbox{XVA}}
\newcommand{\xva}{xva}
\newcommand{\CVA}{\mbox{CVA}}
\newcommand{\DVA}{\mbox{DVA}}

\newtheorem{theorem}{Theorem}[section]
\newtheorem{definition}[theorem]{Definition}
\newtheorem{corollary}[theorem]{Corollary}
\newtheorem{proposition}[theorem]{Proposition}
\newtheorem{remark}[theorem]{Remark}

\newtheorem{assumption}[theorem]{Assumption}
\usepackage[normalem]{ulem}
\newcommand{\w}{w}

\title{Arbitrage-Free XVA}
\author{
Maxim Bichuch \thanks{Email: mbichuch@jhu.edu, Department of Applied Mathematics and Statistics, Johns Hopkins University}
\and
Agostino Capponi \thanks{Email: ac3827@columbia.edu, Industrial Engineering and Operations Research Department, Columbia University}
\and
Stephan Sturm \thanks{Email: ssturm@wpi.edu, Department of Mathematical Sciences, Worcester Polytechnic Institute}
}

\begin{document}

\maketitle

\begin{abstract}
We develop a framework for computing the total valuation adjustment (XVA) of a European claim accounting for funding costs, counterparty credit risk, and collateralization. Based on no-arbitrage arguments, we derive backward stochastic differential equations (BSDEs) associated with the replicating portfolios of long and short positions in the claim. This leads to the definition of buyer{'}s and seller{'}s XVA, which in turn identify a no-arbitrage interval.
In the case that borrowing and lending rates coincide, we provide a fully explicit expression for the unique XVA, expressed as a percentage of the price of the traded claim, and for the corresponding replication strategies. In the general case of asymmetric funding, repo and collateral rates, we study the semilinear partial differential equations (PDE) characterizing buyer{'}s and seller{'}s $\XVA$ and show the existence of a unique classical solution to it. To illustrate our results, we conduct a numerical study demonstrating how funding costs, repo rates, and counterparty risk contribute to determine the total valuation adjustment.\footnote{This article subsumes the two permanent working papers by the same authors: ``Arbitrage-Free Pricing of XVA - Part I: Framework and Explicit Examples'', and ``Arbitrage-Free Pricing of XVA - Part II: PDE Representation and Numerical Analysis''. These papers are accessible at \url{http://arxiv.org/abs/1501.05893} and \url{http://arxiv.org/abs/1502.06106}, respectively.}
\end{abstract}

\vspace{5mm}

\begin{flushleft}
\textbf{Keywords:} XVA, counterparty credit risk, funding spreads, backward stochastic differential equations, arbitrage-free valuation. 
\end{flushleft}

\section{Introduction}

When managing a portfolio, a trader needs to raise cash in order to finance a number of operations. Those include maintaining the hedge of the position, posting collateral resources, and paying interest on collateral received. Moreover, the trader needs to account for the possibility that the position may be liquidated prematurely due to his own or his counterparty{'}s default, hence entailing additional costs due to the closeout procedure. Cash resources are provided to the trader by his treasury desk, and must be remunerated. If he is borrowing, he will be charged an interest rate depending on current market conditions as well as on his own credit quality. Such a rate is usually higher than the rate at which the trader lends excess cash proceeds from his investment strategy to his treasury. The difference between borrowing and lending rate is also referred to as \textit{funding spread}.

Even though pricing by replication can still be put to work under this rate asymmetry, the classical Black-Scholes formula no longer yields the price of the claim. In the absence
of default risk, few studies have been devoted to pricing and hedging claims in markets with differential rates. \cite{Korn95} considers option pricing in a market with a higher borrowing than lending rate, and derives an interval of acceptable prices for both the buyer and the seller. \cite{Cvi93} consider the problem of hedging contingent claims under portfolio constraints allowing for a higher borrowing than lending rate. \cite{ElKaroui} study the super-hedging price of a contingent claim under rate asymmetry via nonlinear backward stochastic differential equations (BSDEs).

The above studies do not consider the impact of counterparty credit risk on valuation and hedging of the derivative security. The new set of rules mandated by the Basel Committee (\cite{Basel3}) to govern bilateral trading in OTC markets requires to take into account default and funding costs when marking to market derivatives positions.
This has originated a growing stream of literature, some of which is surveyed next. \cite{Crepeya} and \cite{Crepeyb} introduce a BSDE approach for the valuation of counterparty credit risk taking funding constraints into account. He decomposes the value of the transaction into three separate components, the contract (portfolio of over-the-counter derivatives), the hedging assets used to hedge market risk of the portfolio as well as
counterparty credit risk, and the funding assets needed to finance the hedging strategy. \cite{BrigoPalCCP} and \cite{BrigoPerPal} derive a risk-neutral pricing formula by taking into account counterparty credit risk, funding, and collateral servicing costs, and provide the corresponding BSDE representation. \cite{Piterbarg} derives a closed form solution for the price of a derivative contract, which distinguishes between funding, repo and collateral rates, but ignores the possibility of counterparty{'}s default. Moreover, he assumes that borrowing and lending rates are equal, an assumption which has been later relaxed by \cite{Mercurio}. \cite{Burgard}  and \cite{BurgardCR} generalize \cite{Piterbarg}{'}s model to include default risk of the trader and of his counterparty. They derive PDE representations for the price of the derivative via a replication approach, assuming the absence of arbitrage and sufficient smoothness of the derivative price. \cite{br} develop a general semimartingale market framework and derive the BSDE representation of the wealth process associated with a self-financing trading strategy that replicates a default-free claim. As in \cite{Piterbarg}, they do not take counterparty credit risk into account. \cite{NiRut}, study the existence of fair bilateral prices. A good overview of the current literature is given in \cite{CrepeyBieleckiBrigo}.

In the present article we introduce a valuation framework which allows us to quantify the total valuation adjustment, abbreviated as XVA, of a European type claim. We consider an underlying portfolio consisting of a default-free stock and two risky bonds underwritten by the trader{'s} firm and his counterparty. Stock purchases and sales are financed through the security lending market. We allow for asymmetry between treasury borrowing and lending rates, repo lending and borrowing rates, as well as between interest rates paid by the collateral taker and received by the collateral provider. We derive the nonlinear BSDEs associated with the portfolios replicating long and short positions in the traded claim, taking into account counterparty credit risk and closeout payoffs exchanged at default. Due to rate asymmetries, the BSDE which represents the valuation process of the portfolio replicating a long position in the claim cannot be directly obtained (via a sign change) from the one replicating a short position. More specifically, there is a no-arbitrage interval which can be defined in terms of the buyer{'}s and the seller{'}s XVA.

We show that our framework recovers the model proposed by \cite{Piterbarg}, as well its extension to the case in  which the hedger and his counterparty can default. In both cases, we can express the total valuation adjustment in closed form, as a percentage of the publicly available price of the claim. This gives an interpretation of the $\XVA$ in terms of funding costs of a trade and counterparty risk, and has risk management implications because it pushes banks to compress trades so as to reduce their borrowing costs and counterparty credit exposures (see \cite{ISDAcompression}). One of the crucial assumptions of the Piterbarg's setup is that rates are symmetric. In the case of asymmetric rates, closed form expressions are unavailable, but we can still exploit the connection between the BSDEs and the corresponding nonlinear PDEs to study numerically how funding spreads, collateral and counterparty risk affect the total valuation adjustment. In this regard, our study extends the previous literature in two directions. First, we develop a rigorous study of the semilinear PDEs associated with the nonlinear BSDEs yielding the XVA. Related studies of the PDE representations of $\XVA$ include \cite{Burgard} and \cite{BurgardCR}, who consider an extended Black-Scholes framework in which two corporate bonds are introduced in order to hedge the default risk of the trader and of his counterparty. They generalize their framework in \cite{BurgardRrisk} to include collateral mitigation and evaluate the impact of different funding strategies. Second, we provide a comprehensive numerical analysis which exploits the previously established existence and uniqueness result. We find strong sensitivity of $\XVA$ to funding costs, repo rates, and counterparty risk. Viewing both buyer{'}s and seller{'}s $\XVA$ as functions of collateralization levels defines a no-arbitrage band whose width increases with the funding spread and the difference between borrowing and lending repo rates.
As the position becomes more collateralized, the trader needs to finance a larger position and the XVA increases. Both buyer's and seller's XVA may decrease if the rates of return of trader and counterparty bonds are
higher than the funding costs incurred for replicating the closeout position.

The paper is organized as follows. We develop the model in Section \ref{sec:model} and introduce the replicated claim and collateral process in Section \ref{sec:claim}. We analyze arbitrage-free valuation and XVA in the presence of funding costs and counterparty risk, referred to as XVA, in Section \ref{sec:BSDEform}. Section \ref{sec:expexp} provides an explicit expression for the $\XVA$ under equal borrowing and lending rates. Section \ref{sec:numanalysis} develops a numerical analysis when borrowing and lending rates are asymmetric. Section \ref{sec:conclusions} concludes the paper. Some proofs of technical results are delegated to an Appendix.

\section{The model}\label{sec:model}

We consider a probability space $(\Omega,\G,{\Px})$ rich enough to support all subsequent constructions. Here, $\Px$ denotes the physical probability measure. Throughout the paper, we refer to ``$I$'' as the investor, trader or hedger interested in computing the total valuation adjustment, and to ``$C$'' as the counterparty to the investor in the transaction. The background or reference filtration that includes all market information except for default events and augmented by all $(\mathcal{G},\Px)$-nullsets, is denoted by $\mathbb{F} := (\mathcal{F}_t)_{t \geq 0}$. The filtration containing default event information is denoted by $\mathbb{H} := (\mathcal{H}_t)_{t \geq 0}$. Both filtrations will be specified in the sequel of the paper. We denote by ${\Gx}:=(\gt)_{t \geq 0}$ the enlarged filtration given by $\gt := \mathcal{F}_t \vee \mathcal{H}_t$, augmented by $(\mathcal{G},\Px)$-nullsets. Note that because of the augmentation of $\mathbb{F}$  by nullsets, the filtration $\mathbb{G}$ satisfies the usual conditions of $(\mathcal{G},\Px)$-completeness and right continuity; see Section 2.4 of \cite{Belanger}.

We distinguish between \textit{universal} instruments, and \textit{investor specific} instruments, depending on whether their valuation is \textit{public} or \textit{private}. Private valuations are based on discount rates, which depend on investor specific characteristics, while public valuations depend on publicly available discount factors. Throughout the paper, we will use the superscript $\wedge$ when referring specifically to public valuations. Section \ref{sec:univ} introduces the universal securities. Investor specific securities are introduced in Section \ref{sec:hedgerspe}.

\subsection{Universal instruments} \label{sec:univ}

This class includes the default-free stock on which the financial claim is written, and the security account used to support purchases or sales of the stock security. Moreover, it includes the risky bond issued by the trader as well as the one issued by his counterparty.

\paragraph{The stock security.}
We let $\mathbb{F} := (\mathcal{F}_t)_{t \geq 0}$ be the $(\mathcal{G},\Px)$-augmentation of the filtration generated by a standard Brownian motion $W^{\Px}$ under the measure $\Px$. Under the physical measure, the dynamics of the stock price is given by
\[
    dS_t = \mu S_t \,dt + \sigma S_t \,dW_t^{\Px},
\]
where $\mu$ and $\sigma$ are constants denoting, respectively, the appreciation rate and the volatility of the stock.

\paragraph{The security account.}
Borrowing and lending activities related to the stock security happen through the security lending or repo market. We do not distinguish between security lending and repo, but refer to all of them as repo transactions. We consider two types of repo transactions: security driven and cash driven, see also \cite{Adrian}. The security driven transaction is used to overcome the prohibition on ``naked'' short sales of stocks, that is the prohibition to the trader of selling a stock which he does not hold and hence cannot deliver. The repo market helps to overcome this by allowing the trader to lend cash to the participants in the repo market who would post the stock as a collateral to the trader. The trader would later return the stock collateral in exchange of a pre-specified amount, usually slightly higher than the original loan amount. Hence, effectively this collateralized loan has a rate, referred to as the repo rate. The cash lender can sell the stock on an exchange, and later, at the maturity of the repo contract, buy it back and return it to the cash borrower. We illustrate the mechanics of the security driven transaction in Figure \ref{fig:secdriven}.

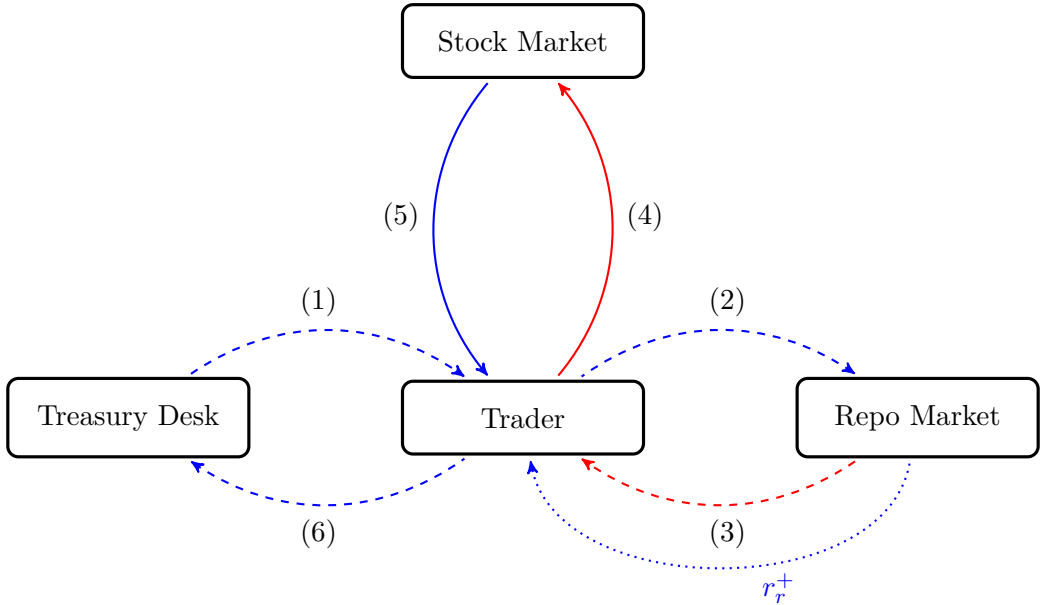
\begin{figure}[ht]
    \centering
    \begin{tikzpicture}[thick,scale=1, every node/.style={transform shape}]
        \node[punkt, inner sep=10pt] (trader) {Trader};
        \node[punkt, inner sep=10pt,  left=2cm of trader] (funder) {Treasury Desk}
            edge[pil, bend left=35, blue, dashed] (trader)
            edge[pil, <-, bend right=35, blue, dashed] (trader);
        \node[above left =1cm of trader] (one) {(1)};
        \node[below left =1cm of trader] (six) {(6)};
        \node[punkt, inner sep=10pt,  above=4cm of trader] (stock) {Stock Market}
            edge[pil, bend right=40, blue] (trader)
            edge[pil, <-, bend left=40, red] (trader);
        \node[below = 2cm of stock.west] (five) {(5)};
        \node[below =2cm of stock.east] (four) {(4)};
        \node[punkt, inner sep=10pt,  right=2cm of trader] (repo) {Repo Market}
            edge[pil, bend left=35, red, dashed] (trader)
            edge[pil, <-, bend right=35, blue, dashed] (trader)
            edge[pil, bend left=80, blue, dotted] (trader);
        \node[above right =1cm of trader] (two) {(2)};
        \node[below right=1cm of trader] (three) {(3)};
        \node[below right=2cm of trader, blue] (rrp) {$\rrp$};
    \end{tikzpicture}
     \caption{Security driven repo activity: Solid lines are purchases/sales, dashed lines borrowing/lending, dotted lines interest due; blue lines are cash, red lines are stock. The treasury desk lends money to the trader (1) who uses it to lend to the repo market (2) receiving in turn collateral (3). He sells the stock on the market to get effectively into a short position (4) earning cash from the deal (5) which he uses to repay his debt to the funding desk (6). As a cash lender, he receives interest at the rate $\rrp$ from the repo market. There are no interest payments between trader and treasury desk as the payments (1) and (6) cancel each other out.}
\label{fig:secdriven}
\end{figure}

The other type of transaction is cash driven. This is essentially the other side of the trade, and is implemented when the trader wants a long position in the stock security. In this case, he borrows cash from the repo market, uses it to purchase the stock security posted as collateral to the loan, and agrees to repurchase the collateral later at a slightly higher price. The difference between the original price of the collateral and the repurchase price defines the repo rate. As the loan is collateralized, the repo rate will be lower than the rate of an uncollateralized loan. At maturity of the repo contract, when the trader has repurchased the stock collateral from the repo market, he can sell it on the exchange. The details of the cash driven transaction are summarized in Figure \ref{fig:cashdriven}.

\begin{figure}[ht]
    \centering
    \begin{tikzpicture}[thick,scale=1, every node/.style={transform shape}]
        \node[punkt, inner sep=10pt] (trader) {Trader};
        \node[punkt, inner sep=10pt,  left=2cm of trader] (funder) {Treasury Desk}
            edge[pil, bend left=35, blue, dashed] (trader)
            edge[pil, <-, bend right=35, blue, dashed] (trader);
        \node[above left =1cm of trader] (one) {(1)};
        \node[below left =1cm of trader] (six) {(6)};
        \node[punkt, inner sep=10pt,  above=4cm of trader] (stock) {Stock Market}
            edge[pil, <-, bend right=40, blue] (trader)
            edge[pil, bend left=40, red] (trader);
        \node[below = 2cm of stock.west] (two) {(2)};
        \node[below =2cm of stock.east] (three) {(3)};
        \node[punkt, inner sep=10pt,  right=2cm of trader] (repo) {Repo Market}
            edge[pil, bend left=35, blue, dashed] (trader)
            edge[pil, <-, bend right=35, red, dashed] (trader)
            edge[pil, <-, bend left=80, blue, dotted] (trader);
        \node[above right =1cm of trader] (four) {(4)};
        \node[below right=1cm of trader] (five) {(5)};
        \node[below right=2cm of trader, blue] (rrm) {$\rrm$};
    \end{tikzpicture}
       \caption{Cash driven repo activity: Solid lines are purchases/sales, dashed lines borrowing/lending, dotted lines interest due; blue lines are cash, red lines are stock. The treasury desk lends money to the trader (1) who uses it to purchase stock (2) from the stock market (3). He uses the stock as collateral (4) to borrow money from the repo market (5) and uses it to repay his debt to the funding desk (6). The trader has thus to pay interest at the rate $\rrm$ to the repo market. There are no interest payments between trader and treasury desk as the payments (1) and (6) cancel each other out.}
    \label{fig:cashdriven}
\end{figure}
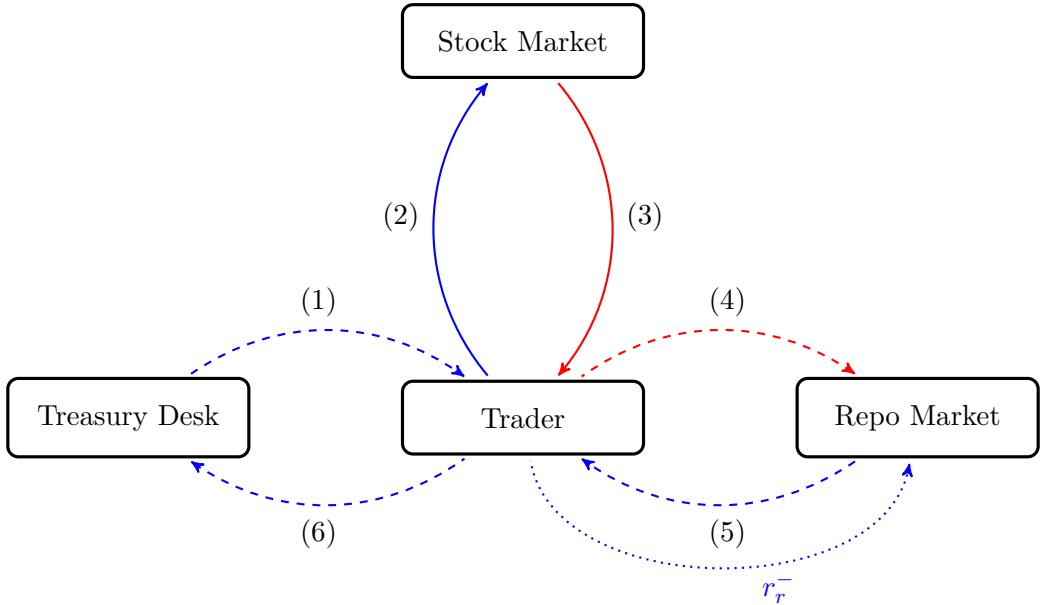

We use $r_r^{+}$ to denote the rate charged by the hedger when he lends money to the repo market and implements his short-selling position. We use $\rrm$ to denote the rate that he is charged when he borrows money from the repo market and implements a long position. We denote by $B^{\rrp}$ and $B^{\rrm}$ the repo accounts whose drifts are given, respectively, by $\rrp$ and $\rrm$. Their dynamics are given by
  \[
    dB_t^{r_r^{\pm}} = r_r^{\pm} B_t^{r_r^{\pm}} dt.
  \]
For future purposes, define
  \begin{equation}\label{eq:Brt}
    B_t^{r_r} := B_t^{r_r}\bigl(\psir \bigr) = e^{\int_0^t r_r (\psir_s) ds},
  \end{equation}
where
  \begin{equation}\label{eq:rr}
    r_r(x) = \rrm \ind_{\{x<0\}}+ \rrp \ind_{\{x>0\}}.
  \end{equation}
Here, $\psir_t$ denotes the number of shares of the repo account held at time $t$. Equations \eqref{eq:Brt}-\eqref{eq:rr} indicate that the trader earns the rate $\rrp$ when lending $\psir>0$ shares of the repo account to implement the short-selling of $-\xi$ shares of the stock security, i.e., $\xi < 0$. Similarly, he has to pay interest rate $\rrm$ on the $-\psir$ ($\psir < 0$) shares of the repo account that he has borrowed by posting $\xi>0$ shares of the stock security as collateral. Because borrowing and lending transactions are fully collateralized, it always holds that
  \begin{equation}\label{eq:selff}
    \psir_t B_t^{r_r} = - \xi_t S_t.
  \end{equation}

\paragraph{The risky bond securities.}
Let $\tau_i$, $i \in \{I, C\}$, be the default times of trader and counterparty. These default times are exponentially distributed random variables with constant
intensities $h_i^{\Px}$, $i \in \{I,C\}$, and are independent of the filtration $\mathbb{F}$. We use $H_i(t)= \ind_{\{\tau_i\leq t\}}$, $t\geq0$, to denote the default indicator process of $i$. The default event filtration is given by $\mathbb{H} = (\mathcal{H}_t)_{t \geq 0}$, $\mathcal{H}_t = \sigma(H^I_u, H^C_u \; ; u \leq t)$. Such a default
model is a special case of the bivariate Cox process framework, for which the $(H)$-hypothesis (see \cite{Elliott}) is well known to hold. In particular, this implies that the $\mathbb{F}$-Brownian motion $W^{\Px}$ is also a $\mathbb{G}$-Brownian motion.

We introduce two risky bond securities with zero recovery underwritten by the trader $I$ and by his counterparty $C$, and maturing at the same time $T$. We denote their price processes by $P^I$ and $P^C$, respectively. For $0 \leq t \leq T$, $i \in \{I, C\}$, the dynamics of their price processes are given by
  \begin{equation}\label{eq:priceproc}
    dP^i_t = \mu_i P^i_t \, dt - P^i_{t-} \,dH_t^i, \qquad P^i_0 = e^{-\mu_i T},
  \end{equation}
with return rates $\mu_i$. We do not allow bonds to be traded in the repo market. Our assumption is driven by the consideration that the repurchase agreement market
for risky bonds is often illiquid. We also refer to the introductory discussion in \cite{Brennan} stating that even if a bond can be shorted on the repo market, the tenor of the agreement
is usually very short.

Throughout the paper, we use $\tau := \tau_I \wedge \tau_C \wedge T$ to denote the earliest of the transaction maturity $T$, trader and counterparty default time.

\subsection{Hedger specific instruments} \label{sec:hedgerspe}

This class includes the funding account and the collateral account of the hedger.

\paragraph{Funding account.}
We assume that the trader lends and borrows moneys from his treasury at possibly different rates. Denote by $\rfp$ the rate at which the hedger lends to the treasury, and by $\rfm$ the rate at which he borrows from it. We denote by $B^{r_f^{\pm}}$ the cash accounts corresponding to these funding rates, whose dynamics are given by
  \[
    dB_t^{r_f^{\pm}} = r_f^{\pm} B_t^{r_f^{\pm}} dt.
  \]
Let $\xi^f_t$ the number of shares of the funding account at time $t$. Define
  \begin{equation}\label{eq:Brf}
    B_t^{r_f}  := B_t^{r_f}\bigl(\xi^f) = e^{\int_0^t r_f(\xi^f_s) ds},
  \end{equation}
where
  \begin{equation}\label{eq:rrf}
    r_f  := r_f(y)= \rfm \ind_{\{{y < 0}\}}+\rfp \ind_{\{{y > 0}\}}.
  \end{equation}
Equations~\eqref{eq:Brf}~-~\eqref{eq:rrf} indicate that if the hedger{'}s position at time $t$, $\xi^f_t$, is negative, then he needs to finance his position. He will do so by borrowing from the treasury at the rate $\rfm$. Similarly, if the hedger{'}s position is positive, he will lend the cash amount to the treasury at the rate $\rfp$.

\paragraph{Collateral process and collateral account.}
The role of the collateral is to mitigate counterparty exposure of the two parties, i.e the potential loss on the transacted claim incurred by one party if the other defaults. The collateral process $C:={(C_t; \; t\geq 0)}$ is an $\mathbb{F}$ adapted process. We use the following sign conventions. If $C_t > 0$, the hedger is said to be the \textit{collateral provider}. In this case the counterparty measures a positive exposure to the hedger, hence asking him to post collateral so as to absorb potential losses arising if the hedger defaults. Vice versa, if $C_t < 0$, the hedger is said to be the \textit{collateral taker}, i.e., he measures a positive exposure to the counterparty and hence asks her to post collateral.

Collateral is posted and received in the form of cash in line with data reported by \cite{ISDA14}, according to which
cash collateral is the most popular form of collateral.\footnote{According to \cite{ISDA14} (see Table 3 therein), cash represents slightly more than $78\%$ of the total collateral delivered and these figures are broadly consistent across years. Government securities instead only constitute $18\%$ of total collateral delivered and other forms of collateral consisting of riskier assets, such as municipal bonds, corporate bonds, equity or commodities only represent a fraction slightly higher than $3\%$.}

We denote by $\rcp$ the rate on the collateral amount received by the hedger if he has posted the collateral, i.e., if he is the collateral provider, while $\rcm$ is the rate paid by the hedger if he has received the collateral, i.e., if he is the collateral taker. The rates $r_c^\pm$ typically correspond to Fed Funds or EONIA rates, i.e., to the contractual rates earned by cash collateral in the US and EURO markets, respectively. We denote by $B^{r_c^{\pm}}$ the cash accounts corresponding to these collateral rates, whose dynamics are given by
  \[
    dB_t^{r_c^{\pm}} = r_c^{\pm} B_t^{r_c^{\pm}} dt.
  \]
Moreover, let us define
  \[
    B_t^{r_c} := B_t^{r_c}(C) = e^{\int_0^t r_c(C_s) ds},
  \]
where
  \[
    r_c(x) = \rcp \ind_{\{x>0\}} + \rcm \ind_{\{x<0\}}.
  \]

Let $\psi_t^c$ be the number of shares of the collateral account $B^{r_c}_t$ held by the trader at time $t$. Then it must hold that
  \begin{equation}\label{eq:collrel}
    \psi_t^{c}B_t^{r_c}  = - C_t.
  \end{equation}
The latter relation means that if the trader is the collateral taker at $t$, i.e., $C_t < 0$, then he has purchased shares of the collateral account i.e., $\psi_t^{c} > 0$. Vice versa, if the trader is the collateral provider at time $t$, i.e., $C_t > 0$, then he has sold
shares of the collateral account to her counterparty.

Before proceeding further, we visualize in Figure \ref{fig:transflow} the mechanics governing the entire flow of transactions taking place.

\begin{figure}[ht]
    \centering
    \begin{tikzpicture}[thick,scale=0.9, every node/.style={transform shape}]
        \node[punkt, inner sep=10pt] (trader) {Trader};
        \node[punkt, inner sep=10pt,  left=2cm of trader] (funder) {Treasury Desk}
            edge[pil, bend left=35, blue, dotted] (trader)
            edge[pil, <-, bend right=35, blue, dotted] (trader)
            edge[pil, <-, bend left=10, blue, dashed] (trader)
            edge[pil, bend right=10, blue, dashed] (trader);
        \node[above left =1cm of trader] (rfp) {$\rfp$};
        \node[below left =1cm of trader] (rfm) {$\rfm$};
        \node[left =0.5cm of trader] (fundingcash) {Cash};
        \node[punkt, inner sep=10pt,  above=4cm of trader] (stock) {Stock \& Repo Market}
            edge[pil, <-, bend right=40, blue, dotted] (trader)
            edge[pil, bend left=40, blue, dotted] (trader)
            edge[pil, <-, bend left=20, red, solid] (trader)
            edge[pil, bend right=20, red, solid] (trader);
        \node[above =2cm of trader] (stocklending) {Stock};
        \node[below = 2cm of stock.west] (rrm) {$\rrm$};
        \node[below =2cm of stock.east] (rrp) {$\rrp$};
        \node[punkt, inner sep=10pt,  right=3.2cm of trader] (bond) {Bond Market}
            edge[pil, bend left=30, black, solid] (trader)
            edge[pil, <-, bend right=30, black, solid] (trader);
        \node[right =0.3cm of trader] (fundingcash) {Bonds $P^I$, $P^C$};
        \node[punkt, inner sep=10pt,  below=4cm of trader] (counter) {Counterparty}
            edge[pil, <-, bend right=50, blue, dotted] (trader)
            edge[pil, bend left=50, blue, dotted] (trader)
            edge[pil, <-, bend left=35, blue, dashed] (trader)
            edge[pil, bend right=35, blue, dashed] (trader);
        \node[below =1.5cm of trader] (collateralization) {Collateral};
        \node[below = 3.5cm of trader.west] (rcp) {$\rcp$};
        \node[below =3.5cm of trader.east] (rcm) {$\rcm$};
    \end{tikzpicture}
    \caption{Trading: Solid lines are purchases/sales, dashed lines borrowing/lending, dotted lines interest due; blue lines are cash, red lines are stock purchases for cash and black lines are bond purchases for cash.}
    \label{fig:transflow}
\end{figure}
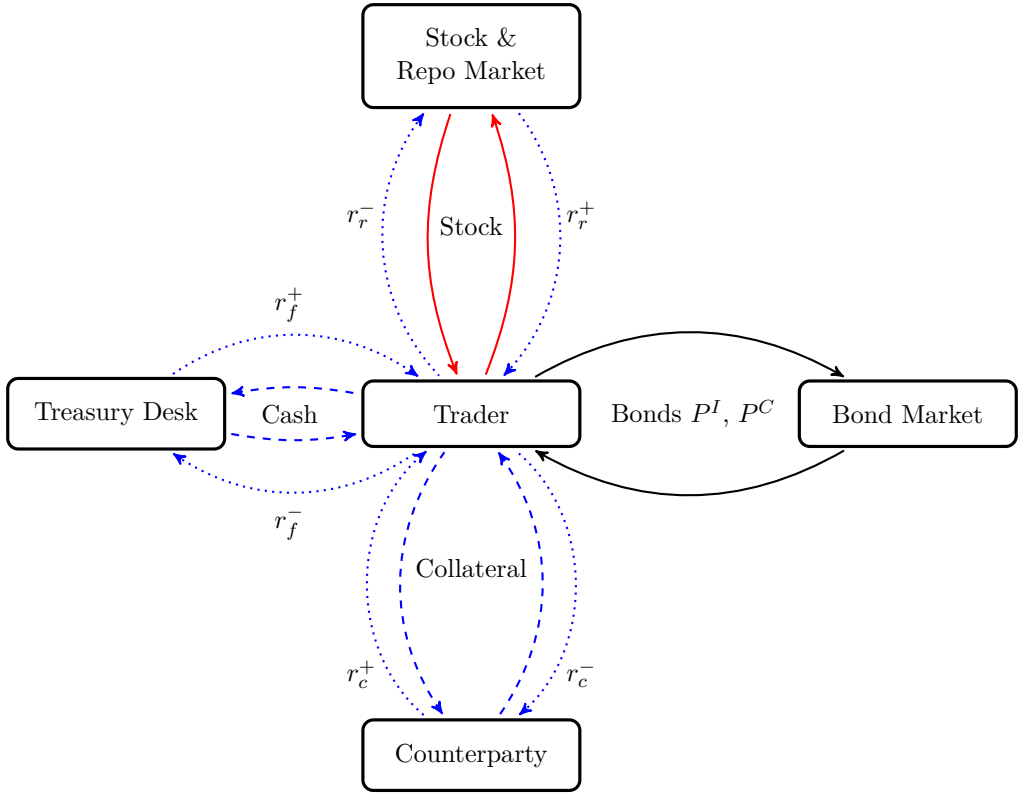

\section{Replicated claim, close-out value and wealth process} \label{sec:claim}

We take the viewpoint of a trader who wants to replicate a European type claim on the stock security. Such a claim is purchased or sold by the trader from/to his counterparty over-the-counter and hence subject to counterparty credit risk. The closeout value of the claim is decided by a valuation agent who might either be one of the parties or a third
party, in accordance with market practices as reviewed by the International Swaps and Derivatives Association {(ISDA)}. {The valuation agent determines the closeout value} of the transaction by calculating the Black Scholes price of the derivative using the discount rate $r_D$. Such a (publicly available) discount rate enables the hedger to introduce a valuation measure $\Qxx$ defined by the property that all securities have instantaneous growth rate $r_D$ under this measure. The rest of the section is organized as follows. We give the details of the valuation measure in Section \ref{sec:valuation}. We introduce the valuation process of the claim to be replicated and of the collateral process in Section \ref{sec:repclaim}. We define the class of admissible strategies in Section \ref{sec:repl} and specify the closeout procedure in Section \ref{sec:closeout}.

\subsection{The valuation measure} \label{sec:valuation}
We first introduce the default intensity model. Under the physical measure $\Px$, default times of trader and counterparty are assumed to be independent exponentially distributed random variables with constant intensities $h_i^{\Px}$, $i \in \{I,C\}$. It then holds that for each $i \in \{I, C\}$
  \[
    \varpi_t^{i,\Px} := H_t^i - \int_0^t\bigl(1-H_u^i\bigr)h_i^{\Px} \, du
  \]
is a $(\mathbb{G},\Px)$-martingale. The valuation measure $\Qxx$ associated with the publicly available discount rate $r_D$ chosen by the valuation agent is equivalent to $\Px$ and given by the Radon-Nikod\'{y}m density
  \begin{equation}\label{eq:q-girsanov}
    \frac{d\Qxx}{d\Px} \bigg|_{\mathcal{G}_{\tau}} = e^{\frac{r_D-\mu}{\sigma}W_{\tau}^{\Px} - \frac{(r_D-\mu)^2}{2\sigma^2}\tau} \Bigl(\frac{\mu_I - r_D}{h_I^{\Px}}\Bigr)^{H^I_\tau} e^{(r_D-\mu_I+h_I^{\Px})\tau}\Bigl(\frac{\mu_C - r_D}{h_C^{\Px}}\Bigr)^{H^C_\tau} e^{(r_D-\mu_C+h_C^{\Px})\tau}.
  \end{equation}
We also recall that $\mu_I$ and $\mu_C$ denote the rate of returns of the bonds underwritten by the trader
and counterparty respectively. Under $\Qxx$, the dynamics of the risky assets are given by
  \begin{align}
    dS_t &= r_D S_t \,dt + \sigma S_t \,dW_t^{\Qxx} \label{eq:S-Q}, \\
    dP_t^I & = r_D P_t^I \, dt - P_{t-}^I d\varpi_t^{I,\Qxx} \nonumber ,\\
    dP_t^C & = r_D P_t^C \, dt - P_{t-}^C d\varpi_t^{C,\Qxx} \nonumber,
  \end{align}
where $W^{\Qxx}:=(W^{\Qxx}_t; \; 0 \leq t \leq \tau)$ is a $(\mathbb{G},\Qxx)$-Brownian motion and $\varpi^{I,\Qxx} := (\varpi_t^{I,\Qxx}; \; 0 \leq t \leq \tau)$ as well as $\varpi^{C,\Qxx}:=(\varpi_t^{C,\Qxx}; \;0 \leq t \leq \tau)$ are  $(\mathbb{G},\Qxx)$-martingales. The above dynamics of $P_t^I$ and $P_t^C$ under the valuation measure $\Qxx$ can be deduced from their respective price processes given in \eqref{eq:priceproc} via a straightforward application of the It\^{o}{'}s formula.

By application of Girsanov{'}s theorem, we have the following relations: $W^{\Qxx}_t = W^{\Px}_t + \frac{\mu-r_D}{\sigma} t$,
$\varpi_t^{i,\Qxx} = \varpi_t^{i,\Px} + \int_0^t \bigl(1-H_u^i\bigr) (h_i^{\Px} - h_i^{\Qxx}) du$. The quantity, $h_i^{\Qxx} = \mu_i-r_D$, $i \in \{I,C\}$, is the default intensity of name $i$ under the valuation measure and is assumed to be positive.

\subsection{Replicated claim and collateral specification} \label{sec:repclaim}
The price process of the contingent claim ${\vartheta} \in L^2(\Omega, \mathcal{F}_T, \Qxx)$ to be replicated is, according to the valuation agent, given by
\[
\hat{V}_t^{{\vartheta}} := e^{-r_D(T-t)} \mathbb{E}^{\Qxx}\bigl[{\vartheta} \, \bigr\vert \, \mathcal{F}_t \bigr].
\]
We will drop the superscript ${\vartheta}$ and just write $\hat{V}_t$ when it is understood from the context. In the case of a European option we have that ${\vartheta}= {\Phi}(S_T)$, where $\Phi:\mathds{R}_{>0} \rightarrow \mathds{R}$ is a real valued function representing the terminal payoff of the claim and thus $\hat{V}_t = \hat{V}(t,S_t)$.

Additionally, the hedger has to post collateral for the claim. As opposed to the collateral used in the repo agreement, which is always the stock, the collateral mitigating counterparty credit risk of the claim is always cash. The collateral is chosen to be a fraction of the current exposure process of one party to the other. If the hedger sells a European call or put option on the security to his counterparty (he would then need to replicate the payoff $\Phi(S_T)${, $\Phi \geq 0$,} and deliver it to the counterparty at $T$), the counterparty always measures a positive exposure to the hedger, while the hedger has zero exposure to the counterparty. As a result, the trader will always be the collateral provider, while the counterparty the collateral taker. By a symmetric reasoning,
if the hedger buys a European call or put option from his counterparty (he would then replicate the payoff $-\Phi(S_T)${, $\Phi \geq 0$,} to hedge his position), then he will always be the collateral taker. On the event that neither the trader nor the counterparty have defaulted by time $t$, the collateral process is defined by
  \begin{equation}\label{eq:rulecoll}
    C_t : = \alpha \hat{V}_t \ind_{\{\tau > t\}}= \alpha \hat{V}(t,S_t)\ind_{\{\tau > t\}},
  \end{equation}
where $0 \leq \alpha \leq 1$ is the collateralization level. The case when $\alpha = 0$ corresponds to zero collateralization, $\alpha = 1$ gives full collateralization. Collateralization levels are industry specific and are reported on a quarterly basis by ISDA, see for instance \cite{ISDA11} , Table 3.3. therein.\footnote{The average collateralization level in 2010 across all OTC derivatives was $73.1\%$. Positions with banks and broker dealers are the most highly collateralized among the different counterparty types with levels around $88.6\%$. Exposures to non-financial corporations and sovereign governments and supra-national institutions tend to have the lowest collateralization levels, amounting to $13.9\%$.}

Our collateral rule differs from \cite{Piterbarg}, where the collateral is assumed to match the value of the contract inclusive of funding costs, repo spreads and collateralization. Using such an approach, both the hedger and his counterparty generally disagree on the level of posted collateral if their funding, repo or collateral rates differ or if they use different models to measure credit risk. Hence, the two counterparties would need to enter into negotiations to agree on a collateral level. In our model, this is avoided because the valuation agent determines the collateral requirements based on the Black-Scholes price $\hat{V}_t$ of the claim, exclusive of funding and counterparty risk related costs.

\subsection{The wealth process} \label{sec:repl}

We allow for collateral to be fully rehypothecated. This means that the collateral taker is granted an unrestricted right to use the collateral amount, i.e., he can use it to purchase
investment securities. This is in agreement with most ISDA annexes, including the New York Annex, English Annex, and Japanese Annex. We notice that for the case of cash collateral,
the percentage of rehypothecated collateral amounts to about $90\%$ (see Table 8 in \cite{ISDA14}) hence largely supporting our assumption of full collateral re-hypothecation. As in
\cite{br}, the collateral received can be seen as an ordinary component of a hedger{'}s trading strategy, although this applies only prior to the counterparty{'}s default. {We denote by $V_t({\bm\varphi})$ the wealth process of the hedger and emphasize that the collateral can actually be used by him (via rehypothecation) if he is the collateral taker.}

Let ${\bm\varphi} := \bigl(\xi_t,\xi_t^f,\xi_t^I, \xi_t^C; t \geq 0\bigr)$, where we recall that $\xi_t$ denotes the number of shares of the security, $\xi^{f}_t$ the number of shares in the funding account, and we use $\xi^I_t$ and $\xi^C_t$ to denote the number of shares of trader and counterparty bonds, respectively, at time $t$. Recalling Eq.~\eqref{eq:collrel}, and expressing all positions in terms of number of shares multiplied by the price of the corresponding security, the wealth process $V({\bm \varphi})$ is given by the following expression
  \begin{equation}\label{eq:wealth}
    V_t({\bm\varphi}) := \xi_t S_t + \xi_t^I P_t^I + \xi_t^C P_t^C + \xi_t^f B_t^{r_f} + \psir_t B_t^{r_r} - \psi_t^{c} B_t^{r_c},
  \end{equation}
where we notice that the number of shares $\psir$ of the repo account and the number of shares $\psi^{c}$ held in the collateral account are uniquely determined by equations~\eqref{eq:selff} and~\eqref{eq:collrel}, respectively.

\begin{definition}
A collateralized trading strategy ${\bm\varphi}$ is \textit{self-financing} if, for $t \in [0,T]$, it holds that
  \[
    V_t({\bm\varphi}) := V_0({\bm\varphi}) + \int_0^t \xi_u \, dS_u + \int_0^t \xi_u^I \, dP_u^I + \int_0^t \xi_u^C \, dP_u^C + \int_0^t \xi_u^f \, dB_u^{r_f} +  \int_0^t \psir_u \, dB_u^{r_r} -\int_0^t \psi_u^{c} \, dB_u^{r_c},
  \]
where $V_0({\bm\varphi}) =V_0$ is the initial endowment.
\end{definition}

Moreover, we define the class of admissible strategies as follows:

\begin{definition}\label{def:control-U}
The admissible set of trading strategies is given as class of $\mathbb{F}$-predictable processes such that the portfolio process is bounded from below, see also \cite{Delbaen}.
\end{definition}

\subsection{Close-out value of transaction}\label{sec:closeout}

The ISDA market review of OTC derivative collateralization practices, see \cite{ISDA2010}, section 2.1.5, states that the surviving party should evaluate the transactions that have been terminated due to the default event, and claim for a reimbursement only after mitigating losses with the available collateral. In our study, we follow the risk-free closeout convention meaning that the trader liquidates his position at the market value when his counterparty defaults. Next, we describe how this is modeled in our framework. Denote by $\theta$ the value of the replicating portfolio at $\tau$, where we recall that $\tau$ has been defined in Section \ref{sec:model}. This value represents the amount of wealth that the trader must hold in order to replicate the closeout payoffs when the transaction terminates. It is given by
  \begin{align}\label{eq:theta}
    \nonumber  \theta \phantom{:}=  {\theta(\tau, \hat{V})} & := {\hat{V}_{\tau} + \ind_{\{\tau_C<\tau_I \}} L_C Y^- - \ind_{\{\tau_I<\tau_C \}} L_I Y^+} \\
     & \phantom{:}= { \ind_{\{\tau_I<\tau_C \}} \theta_{I}(\hat V_{\tau}) +  \ind_{\{\tau_C<\tau_I \}} \theta_{C}(\hat V_{\tau})},
  \end{align}
where $Y:= {\hat{V}_\tau} - C_{\tau} {=(1-\alpha)\hat{V}_\tau}$ is the value of the claim at default, netted of the posted collateral and we define
\begin{equation}\label{eq:thetaIC}
   \theta_{I}(v)  := v -   L_I  ((1-\alpha)v )^{+},\qquad
   \theta_{C}(v)  := v +   L_C ((1-\alpha) v )^{-},
\end{equation}
where for a real number $x$ we are using the notations $x^+ := \max(x,0)$, and $x^- := \max(0,-x)$. The term $\ind_{\{\tau_C<\tau_I \}} L_C Y^-$ originates the residual $\CVA$ term after collateral mitigation, while
$\ind_{\{\tau_I<\tau_C\}}  L_I Y^+$ originates the $\DVA$ term, see also \cite{BrigoCapPal} and \cite{capmig} for additional details. The quantities $0 \leq L_I \leq 1$ and $0 \leq L_C \leq 1$ are the loss rates against
the trader and counterparty claims, respectively.

\begin{remark}
We elaborate on why {$\theta$} is the amount which needs to be replicated by the trader when the transaction terminates. Suppose that the trader has sold a call option to the counterparty (hence $\hat{V}(t,S_{t})>0$ for all $t$). This means that the trader is always the collateral provider, $C_t = \alpha \hat{V}(t,S_{t}) > 0$ for all $t$, given that the counterparty always measures a positive exposure to the trader. If the counterparty defaults first and before the maturity of the claim, the trader will net the amount $\hat{V}(\tau,S_{\tau})$ owed to the counterparty with his collateral posted to the counterparty, and only return to her the residual amount $Y$. As a result, his trading strategy must yield an actual wealth equal to this amount. The above expression of closeout given by ${\theta}= \hat{V}(\tau,S_{\tau})$ indicates that this is indeed the case. Because the counterparty already holds the collateral, the trader only needs to return to her the amount $Y$, which is precisely the wealth process of the trader at $\tau$.
\end{remark}

\section{Arbitrage-free valuation and XVA} \label{sec:BSDEform}

The goal of this section is to find a valuation for the derivative security with payoff $\Phi(S_T)$ that is free from arbitrage in a certain sense. Before discussing arbitrage-free valuations, we have to make sure that the underlying market does not admit arbitrage from the hedger's perspective (as discussed in \cite[Section 3]{br}). In the underlying market, the trader is only allowed to borrow/lend stock, buy/sell risky bonds and borrow/lend from/to the funding desk. In particular, neither the derivative security, nor the collateral process is involved.

\begin{definition}\label{def:no-arb}
 We say that the market $(S,P^I,P^C)$ admits \textit{hedger's arbitrage} if we can find a trading strategy ${\bm\varphi} = \bigl(\xi_t,\xi_t^f,\xi_t^I, \xi_t^C; t \geq 0\bigr)$ such that, given a non-negative initial capital $x \geq 0$ of the hedger and denoting the wealth process corresponding to it $\bigl(V_t({\bm\varphi}, x)\bigr)_{t\geq 0}$, we have that $\Px \bigl[V_\tau({\bm\varphi}, x) \geq e^{\rfp \tau}x \bigr] =1$ and $\Px\bigl[V_\tau({\bm\varphi}, x) > e^{\rfp \tau}x\bigr] >0$. If the market does not admit hedger's arbitrage for the hedger's initial capital $x \geq 0$, we say that the market is arbitrage free from the hedger's perspective.
\end{definition}

We will omit the arguments $x$, ${\bm\varphi}$ or both in the wealth process $V({\bm\varphi},x)$, whenever understood from the context. In the sequel, we make the following assumption:

\begin{assumption}\label{ass:necessary}
The following relations hold between the different rates: $\rrp \le \rfm$, $\rfp \leq \rfm$, and {$\rfp \vee r_D < \mu_I \wedge \mu_C$}.
\end{assumption}

\begin{remark}\label{rem:measure}
The above assumption is necessary to preclude arbitrage. {The condition $r_D < \mu_I \wedge \mu_C$ is needed for the existence of the valuation measure as discussed at the end of section \ref{sec:valuation} ($h_i^{\Qxx} = \mu_i-r_D$ and risk-neutral default intensities must be positive).} If, by contradiction, $\rrp > \rfm$, the trader can borrow cash from the funding desk at the rate $\rfm$ and lend it to the repo market at the rate $\rrp$, while holding the stock as a collateral. This results in a sure win for the trader. Similarly, if the trader could fund his strategy from the treasury at a rate $\rfm < \rfp$, it would clearly result in an arbitrage. The condition $\rfp < \mu_I$ (and mutatis mutandis $\rfp < \mu_C$) has a more practical interpretation: it precludes the arbitrage opportunity of short selling the bond underwritten by the trader{'s} firm and investing the proceeds in the funding account.
\end{remark}

We next provide a sufficient condition guaranteeing that the underlying market is free of arbitrage.

\begin{proposition}\label{thm:arb-market}
Suppose that in addition to Assumption \ref{ass:necessary}, $\rrp \leq \rfp \leq \rrm$. Then the model does not admit arbitrage opportunities for the hedger for any $x \geq 0$.
\end{proposition}
We remark that in a market model without defaultable securities, similar inequalities between borrowing and lending rates have been derived by \cite{br} (Proposition 3.3), and by \cite{NiRut} (Proposition 3.1). We impose additional relations between lending rates and return rates of the risky bonds given that our model also allows for counterparty risk.

\begin{proof}
First, observe that under the conditions given above we have
  \begin{align*}
    \nonumber   r_r \psir_t& = \rrp \psir_t \ind_{\{\psir_t>0\}} + \rrm \psir_t \ind_{\{\psir_t< 0\}} \leq \rfp \psir_t \ind_{\{\psir_t>0\}} + \rfp \psir_t \ind_{\{\psir_t < 0\}} = \rfp \psir_t \\
    r_f \xi_t^f & = \rfp \xi_t^f \ind_{\{\xi_t^f >0\}} + \rfm \xi_t^f \ind_{\{{\xi_t^f} < 0\}} \leq \rfp \xi_t^f \ind_{\{\xi_t^f >0\}} + \rfp \xi_t^f\ind_{\{\xi_t^f < 0\}} = \rfp \xi_t^f
  \end{align*}
Next, it is convenient to write the wealth process under a suitable measure $\tilde{\Px}$ specified via the stochastic exponential
  \[
    \frac{d\tilde{\Px}}{d\Px} \bigg|_{\mathcal{G}_{\tau}} = e^{\frac{\rfp-\mu}{\sigma}W_{\tau}^{\Px} - \frac{(\rfp-\mu)^2}{2\sigma^2}\tau} \Bigl(\frac{\mu_I - \rfp}{h_I^{\Px}}\Bigr)^{H^I_\tau} e^{(\rfp -\mu_I + h_I^{\Px})\tau}\Bigl(\frac{\mu_C - \rfp}{h_C^{\Px}}\Bigr)^{H^C_\tau} e^{(\rfp-\mu_C + h_C^{\Px})\tau}
  \]
By Girsanov's theorem, $\tilde{\Px}$ is an equivalent measure to $\Px$ such that the dynamics of the risky assets are given by
  \begin{align*}
    dS_t & = \rfp S_t \, dt + \sigma S_t dW_t^{\tilde{\Px}},\\
    dP_t^I & = \rfp P_t^I \, dt - P_{t-}^I d\varpi_t^{I,\tilde{\Px}},\\
    dP_t^C & = \rfp P_t^C \, dt - P_{t-}^C d\varpi_t^{C,\tilde{\Px}}
  \end{align*}
where $W^{\tilde{\Px}}:=(W^{\tilde{\Px}}_t; \; 0 \leq t \leq \tau)$ is a  $(\mathbb{G},\tilde{\Px})$-Brownian motion {and} $\varpi^{I,\tilde{\Px}} := (\varpi_t^{I,\tilde{\Px}}; \; 0 \leq t \leq \tau)$ {as well as} $\varpi^{C,\tilde{\Px}}:=(\varpi_t^{C,\tilde{\Px}}; \;0 \leq t \leq \tau)$ are $(\mathbb{G},\tilde{\Px})$-martingales. The $\rfp$ discounted assets $\tilde{S}_t := e^{-\rfp t} S_t$, $\tilde{P}_t^I := e^{-\rfp t} P_t^I$ and $\tilde{P}_t^C := e^{-\rfp t} P_t^C$ are thus {$(\mathbb{G},\tilde{\Px})$-}martingales. In particular, $W^{\tilde{\Px}} = W^{\Px} + \frac{\mu-\rfp}{\sigma}$ and the default intensity of the hedger and of his counterparty under $\tilde{\Px}$ are given by $h_i^{\tilde{\Px}} = \mu_i-\rfp$, {$i \in \{I,C\}$}, {which is positive} in light of the assumptions of the proposition.

Denote the wealth process associated with $(S_t,P_t^I,P_t^C)_{t \geq 0}$ in the underlying market by $\check{V}_t$. Using the self-financing condition, its dynamics are given by
  \begin{align*}
    \nonumber d\check{V}_t &= \bigl(r_f \xi_t^f B_t^{r_f} + \rfp \xi_t S_t + r_r \psir_tB_t^{r_r} + \rfp \xi_t^I P_t^I + \rfp \xi_t^C P_t^C \bigr) \, dt  \\
    \nonumber & \phantom{=}+ \xi_t \sigma S_t \, dW_t^{\tilde{\Px}}  - {\xi_t^I} P_{t-}^I  \, d\varpi_t^{I,\tilde{\Px}}  - {\xi_t^C} P_{t-}^C \, d\varpi_t^{C,\tilde{\Px}}\\
    & = \bigl(r_f \xi_t^f B_t^{r_f} + r_r \psir_t B_t^{r_r}  \bigr) \, dt + \xi_t \,d S_t + {\xi_t^I} \, dP_t^I + {\xi_t^C} \, dP_t^C.
  \end{align*}
Then we have that
  \begin{align*}
    \check{V}_\tau({\bm\varphi},x) - \check{V}_0({\bm \varphi},x) & = \int_0^\tau \bigl( \rfp \xi_t S_t + r_f \xi_t^f B_t^{r_f} + r_r \psir_tB_t^{r_r} + \rfp\xi_t^I P^I_t + \rfp\xi_t^C P^C_t  \bigr) \, dt  \\
    \nonumber & \phantom{=}+ \int_0^\tau \xi_t \sigma S_t \, dW_t^{\tilde{\Px}} - \int_0^\tau {\xi_t^I} P_{t-}^I \, d\varpi_t^{I, \tilde{\Px}} - \int_0^\tau {\xi_t ^C} P_{t-}^C\, d \varpi_t^{C,\tilde{\Px}} \\
    & \leq \int_0^\tau \bigl( \rfp \xi_t S_t + \rfp \xi_t^f B_t^{r_f} + \rfp \psir_tB_t^{r_r} + \rfp\xi_t^I P^I_t + \rfp\xi_t^C P^C_t \bigr) \, dt  \\
    \nonumber & \phantom{=}+ \int_0^\tau \xi_t \sigma S_t \, dW_t^{\tilde{\Px}} - \int_0^\tau {\xi_t^I} P_{t-}^I \, d\varpi_t^{I, \tilde{\Px}} - \int_0^\tau {\xi_t^C} P_{t-}^C\, d \varpi_t^{C,\tilde{\Px}} \\
    & = \int_0^\tau \rfp \check{V}(x) \, dt + \int_0^\tau \xi_t \sigma S_t \, dW_t^{\tilde{\Px}} - \int_0^\tau {\xi_t^I} P_{t-}^I \, d\varpi_t^{I, \tilde{\Px}} - \int_0^\tau {\xi_t^C} P_{t-}^C\, d \varpi_t^{C,\tilde{\Px}}.
  \end{align*}
Therefore, it follows that
 \[
    e^{-\rfp\tau} \check{V}_{\tau}({\bm \varphi},x) -\check{V}_0({\bm \varphi},x) \leq \int_0^\tau \xi_t \, d\tilde{S}_t - \int_0^{\tau} {\xi_t^I} \, d\tilde{P}_{t-}^I - \int_0^{\tau} {\xi_t^C} \, d\tilde{P}_{t-}^C.
  \]
Note that the right hand side of the above inequality is a local martingale bounded from below (as the value process is bounded from below by the admissibility condition), and therefore is a supermartingale. Taking expectations, we conclude that
  \[
    \Exx^{\tilde{\Px}} \bigl[  e^{-\rfp\tau} \check{V}_\tau({\bm \varphi},x) - \check{V}_0({\bm \varphi},x) \bigr] \leq 0.
  \]
Thus either $\tilde{\Px} \bigl[\check{V}_\tau({\bm \varphi},x) = e^{\rfp\tau} x \bigr] = 1$ or $\tilde{\Px} \bigl[ \check{V}_\tau({\bm \varphi},x) < e^{\rfp\tau} x \bigr] >0 $. As $\tilde{\Px}$ is equivalent to $\Px$, this shows that arbitrage opportunities for the hedger are precluded in this model (he would receive {$e^{\rfp \tau}x$} by lending the positive cash amount $x$ to the treasury desk at the rate $\rfp$).
\end{proof}

Next we want to define the notion of an arbitrage free price of a derivative security from the hedger{'}s perspective. We will assume that the hedger has zero initial capital, or equivalently, he does not have liquid initial capital which can be used for hedging the claim until maturity. The hedging portfolio will thus be entirely financed by purchase/sale of the stock via the repo market and purchase/sale of bonds via the funding account. While our entire analysis might be extended to the case of nonzero initial capital, such an assumption will simplify notation and allow us to highlight the key aspects of the study.

\begin{definition}\label{def:arb-price}
The valuation $P \in \mathbb{R}$ of a derivative security with terminal payoff $\vartheta \in \mathcal{F}_T$ is called \textit{hedger's arbitrage-free} if for all $\gamma \in \mathbb{R}$, buying $\gamma$ securities for $\gamma P$ and hedging in the market with an admissible strategy and zero initial capital, does not create hedger's arbitrage.
\end{definition}

Before giving a characterization of hedger's arbitrage-free valuations, we analyze the dynamics of the wealth process. We will rewrite it under the valuation measure $\Qxx$ for notational simplicity. Using the condition~\eqref{eq:selff}, we obtain from Eq.~\eqref{eq:wealth} that
  \begin{align}\label{eq:vtlast}
    \nonumber dV_t &= \bigl(r_f \xi_t^f B_t^{r_f} + (r_D - r_r) \xi_t S_t + r_D \xi_t^I P_t^I + r_D \xi_t^C P_t^C - r_c \psi_t^c B_t^{r_c} \bigr) \, dt  \\
    \nonumber & \phantom{=}+ \xi_t \sigma S_t \, dW_t^{\Qxx}  - {\xi_t^I} P_{t-}^I \, d\varpi_t^{I,\Qxx}  - {\xi_t^C} P_{t-}^C  \, d\varpi_t^{C,\Qxx}\\
    \nonumber &= \Bigl( \rfp \bigl(\xi_t^f B_t^{r_f}\bigr)^+ -\rfm \bigl(\xi_t^f B_t^{r_f}\bigr)^- +  (r_D - \rrm) \bigl(\xi_t S_t\bigr)^+ - (r_D - \rrp) \bigl(\xi_t S_t\bigr)^- + r_D \xi_t^I P_t^I + r_D \xi_t^C P_t^C\Bigr) \, dt \\
    \nonumber & \phantom{=}-\Bigl(\rcm \bigl( \psi_t^c B_t^{r_c}\bigr)^+ -  \rcp \bigl(\psi_t^c B_t^{r_c}\bigr)^-  \Bigr) \, dt  + \xi_t \sigma S_t \, dW_t^{\Qxx} - \xi_t^I P_{t-}^I \, d\varpi_t^{I,\Qxx}  - \xi_t^C P_{t-}^C  \, d\varpi_t^{C,\Qxx} \\
    \nonumber &= \Bigl( \rfp \bigl(\xi_t^f B_t^{r_f}\bigr)^+ -\rfm \bigl(\xi_t^f B_t^{r_f}\bigr)^- +  (r_D - \rrm) \bigl(\xi_t S_t\bigr)^+ - (r_D - \rrp) \bigl(\xi_t S_t\bigr)^- + r_D \xi_t^I P_t^I + r_D \xi_t^C P_t^C\Bigr) \, dt \\
     & \phantom{=} + \Bigl(\rcp \bigl(C_t\bigr)^+ - \rcm \bigl(C_t\bigr)^-  \Bigr) \, dt  + \xi_t \sigma S_t \, dW_t^{\Qxx}  - \xi_t^I P_{t-}^I   \, d\varpi_t^{I,\Qxx}  - \xi_t^C P_{t-}^C d\varpi_t^{C,\Qxx}.
  \end{align}
Setting
  \begin{equation}\label{eq:Zetas}
    Z_t = \xi_t \sigma S_t,\qquad Z^I_t = -\xi_t^I P_{t-}^I,\qquad Z_t^C = -\xi_t^C P_{t-}^C,
  \end{equation}
and using again the condition~\eqref{eq:selff} and Eq.~\eqref{eq:wealth}, we obtain
  \begin{align}
    \xi_t^f B_t^{r_f} &= V_t - \xi_t^I P^I_t - \xi_t^C P^C_t + \psi_t^{c} B_t^{r_c} \nonumber \\
    &= V_t- \xi_t^I P^I_t - \xi_t^C P^C_t - C_t.
  \label{eq:funding}
  \end{align}
Then the dynamics in Eq.~\eqref{eq:vtlast} reads as
  \begin{align}\label{eq:vtlast2}
    \nonumber dV_t &= \Bigl(\rfp \bigl(V_t + Z_t^I + Z_t^C - C_t\bigr)^+ -\rfm \bigl(V_t + Z_t^I + Z_t^C - C_t\bigr)^-\\
    \nonumber & \phantom{=} +  (r_D - \rrm) \frac{1}{\sigma}\bigl(Z_t\bigr)^+ -  (r_D - \rrp) \frac{1}{\sigma}\bigl(Z_t\bigr)^- - r_D Z_t^I - r_D Z_t^C   +  \rcp \bigl(C_t\bigr)^+ - \rcm \bigl(C_t\bigr)^- \Bigr)  dt \\
    & \phantom{=} + Z_t\, dW_t^{\Qxx} + Z_t^I \, d\varpi_t^{I,{ \Qxx}}  + Z_t^C \, d\varpi_t^{C,\Qxx}.
  \end{align}
We next define {the drivers}
  \begin{align}
    f^+\bigl(t,v,z,z^I,z^C; \hat{V}\bigr) &:= -\Bigl(\rfp \bigl(v+ z^I + z^C- \alpha \hat{V}_t\bigr)^+ -\rfm \bigl(v + z^I + z^C - \alpha \hat{V}_t\bigr)^- \nonumber\\ & \phantom{=:}+  (r_D - \rrm) \frac{1}{\sigma}z^+  -  (r_D - \rrp) \frac{1}{\sigma}z^- - r_D z^I - r_D z^C \nonumber\\
    & \phantom{=:} + \rcp \bigl(\alpha \hat{V}_t\bigr)^+ - \rcm\bigl(\alpha \hat{V}_t\bigr)^-  \Bigr)\label{eq:f+}\\
    f^-\bigl(t,v,z,z^I,z^C; \hat{V}\bigr) &:= - f^+\bigl(t,-v,-z,-z^I,-z^C; -\hat{V}\bigr)\label{eq:f-},
  \end{align}
which depend on the market valuation process $\bigl(\hat{V}_t\bigr)_{{t \geq 0}}$ (via the collateral $C$) and we omit indicators as we are interested in hedging only up to default. In particular $f^\pm \, : \, \Omega \times [0,T] \times \R^4$, $(\omega, t,v,z,z^I,z^C) \mapsto f^\pm\bigl(t,v,z,z^I,z^C; \hat{V}\bigr)$ are drivers of BSDEs admitting unique solutions {as implied by Corollary \ref{cor:BSDE-V-red}.} Moreover, define $V^{+,\gamma}$ and $V^{-,\gamma}$ as solutions of the BSDEs
  \begin{align}\label{eq:BSDE-sell}
 \nonumber -dV_t^{+,\gamma} &= f^+\bigl(t,V_t^{+,\gamma},Z_t^{+,\gamma},Z_t^{I,+,\gamma},Z_t^{C,+,\gamma}; \hat{V}\bigr) \, dt - Z^{+,\gamma}_t\, dW_t^{\Qxx} - Z_t^{I,+,\gamma} \, d\varpi_t^{I,\Qxx}  - Z_t^{C,+,\gamma} \, d\varpi_t^{C,\Qxx},\\
    V_\tau^{+,\gamma} & = \gamma \Bigl({\theta_I(\hat{V}_\tau)\ind_{\{\tau_I < \tau_C \wedge T\}} + \theta_C(\hat{V}_\tau)\ind_{\{\tau_C < \tau_I \wedge T\}}} + \vartheta \ind_{\{\tau = T\}}\Bigr).
  \end{align}
and
  \begin{align}\label{eq:BSDE-buy}
\nonumber -dV_t^{-,\gamma} &= f^-\bigl(t,V_t^{-\gamma},Z_t^{-,\gamma},Z_t^{I,-,\gamma},Z_t^{C,-,\gamma}; \hat{V}\bigr) \, dt - Z^{-,\gamma}_t\, dW_t^{\Qxx} - Z_t^{I,-,\gamma} \, d\varpi_t^{I,\Qxx}  - Z_t^{C,-,\gamma} \, d\varpi_t^{C,\Qxx},\\
    V_\tau^{-,\gamma} & = \gamma \Bigl({\theta_I(\hat{V}_\tau)\ind_{\{\tau_I < \tau_C \wedge T\}} + \theta_C(\hat{V}_\tau)\ind_{\{\tau_C < \tau_I \wedge T\}}} + \vartheta \ind_{\{\tau = T\}}\Bigr).
  \end{align}
We note that {$V^{+,\gamma}$} describes the wealth process when replicating the claim $\gamma \vartheta$ for $\gamma >0$ (hence hedging the position after selling $\gamma$ securities with terminal payoff $\vartheta$) with zero initial capital.
On the other hand, $\bigl(-V_t^{-,\gamma}\bigr)$ describes the wealth process when replicating the claim $-\gamma \vartheta$, $\gamma>0$ (hence hedging the position after buying $\gamma$ securities with terminal payoff $\vartheta$) with zero initial capital. {Notice that by positive homogeneity of the drivers $f^+$ and $f^-$ of the BSDEs \eqref{eq:BSDE-sell} and \eqref{eq:BSDE-buy}, it is enough to consider the cases when $\gamma =1$.} To ease the notation, we set $V_t^{+,1} = V_t^+$ and $V_t^{-,1} = V_t^-$. We also note that the two BSDEs are intrinsically related: {$(V^-,Z^-,Z^{I,-},Z^{C,-})$} is a solution to the data $\bigl(f^-,{\theta_I(\hat{V}_\tau)\ind_{\{\tau_I < \tau_C \wedge T\}} + \theta_C(\hat{V}_\tau)\ind_{\{\tau_C < \tau_I \wedge T\}}}, \vartheta\bigr)$ if and only if {$(-V^-,-Z^-,-Z^{I,-},-Z^{C,-})$} is a solution to the data $\bigl(f^+, {\theta_I(-\hat{V}_\tau)\ind_{\{\tau_I < \tau_C \wedge T\}} + \theta_C(-\hat{V}_\tau)\ind_{\{\tau_C < \tau_I \wedge T\}}}, -\vartheta\bigr)$.
	
Our goal is to compute the \textit{total valuation adjustment} $\XVA$ that needs to be added to the Black-Scholes price of the claim to get the actual valuation. As we have seen, the situation is asymmetric for sell- and buy-valuations, so we will have to define it both from the seller's and buyer's viewpoint.
	\begin{definition}
		The seller's $\XVA$ is the $\mathbb{G}$-adapted stochastic process $(\XVA)_{t\geq 0}$ defined as
		\begin{align}
		\XVA_t^{\sell} := V^+_t - \hat{V}_t
		\label{eq:XVA-sell}
		\end{align}
		while the buyer's $\XVA$ is defined as
		\begin{align}
		\XVA_t^{\buy} := V^-_t - \hat{V}_t.
		\label{eq:XVA-buy}
		\end{align}
	\end{definition}
	
$\XVA^{\sell}$ quantifies the total costs (including collateral, funding, and counterparty risk related costs) incurred by the trader to hedge a long position in the option, whereas $\XVA^{\buy}$ quantifies the total costs incurred when hedging a short position. As we will see in Section \ref{sec:pitdef}, these two $\XVA$s agree only if the drivers of the BSDEs are linear.
Noting that the agent's valuation process $\hat{V}$ satisfies the Black-Scholes BSDE
\begin{align}
 -d\hat{V}_t & = -r_D\hat{V}_t \, dt - \hat{Z}_t \, dW_t^\Q,\nonumber \\
 \hat{V}_T & = \vartheta,
\end{align}
which is well known to admit a unique solution, we can immediately obtain BSDEs for the $XVA^{\pm}$:
	\begin{align}\label{eq:XVABSDE}
	-d\XVA_t^{\buysell} & = \tilde{f}^{\pm}\bigl(t,\XVA_t^{\buysell},\tilde{Z}_t^{\pm},\tilde{Z}_t^{I,\pm},\tilde{Z}_t^{C,\pm}; \hat{V}\bigr) \, dt \nonumber\\& \phantom{=} - \tilde{Z}_t^{\pm}\, dW_t^{\Qxx} - \tilde{Z}_t^{I,\pm} \, d\varpi_t^{I,\Qxx}  - \tilde{Z}_t^{C,\pm} \, d\varpi_t^{C,\Qxx},\nonumber
	\\
	\XVA_\tau^{\buysell} &= \tilde{\theta} _{C}(\hat V_\tau) \ind_{\{\tau_C<\tau_I \wedge T\}} + \tilde{\theta} _{I}(\hat V_\tau)\ind_{\{\tau_I<\tau_C \wedge T \}},
	\end{align}
	with
	\begin{align}
	\tilde{Z}_t^{\pm} & := Z_t^{\pm} - \hat{Z}_t, \qquad \tilde{Z}_t^{I, \pm} = Z_t^{I, \pm}, \qquad \tilde{Z}_t^{C, \pm} = Z_t^{C, \pm}, \nonumber \\  \tilde{\theta} _{C}(\hat v)  & := L_C ((1-\alpha) \hat{v} )^{-}, \qquad \tilde{\theta} _{I}(\hat v)  := - L_I  ((1-\alpha)\hat{v} )^{+},
	\label{eq:hats}
	\end{align}
	and	
	\begin{align}
	\tilde{f}^+\bigl(t,xva,\tilde{z},\tilde{z}^I,\tilde{z}^C; \hat{V}\bigr) &{:}= -\Bigl(\rfp \bigl(xva+ \tilde{z}^I + \tilde{z}^C +(1- \alpha) \hat{V}_t\bigr)^+ -\rfm \bigl(xva + \tilde{z}^I + \tilde{z}^C +(1- \alpha)\hat{V}_t\bigr)^- \nonumber\\ & \phantom{=:}+  (r_D - \rrm) \frac{1}{\sigma}\tilde{z}^+  -  (r_D - \rrp) \frac{1}{\sigma}\tilde{z}^- - r_D \tilde{z}^I - r_D \tilde{z}^C \nonumber\\
	& \phantom{=:} + \rcp \bigl(\alpha \hat{V}_t\bigr)^+ - \rcm\bigl(\alpha \hat{V}_t\bigr)^-  \Bigr) + r_D\hat{V}_t,\label{eq:tilde-f+}\\
	\tilde{f}^-\bigl(t,xva,\tilde{z},\tilde{z}^I, \tilde{z}^C; \hat{V}\bigr) &{:}= - {\tilde f} ^+\bigl(t,-xva,-\tilde{z},-\tilde{z}^I,-\tilde{z}^C; -\hat{V}\bigr).\label{eq:tilde-f-}
	\end{align}
	Note that comparing \eqref{eq:f+} with \eqref{eq:tilde-f+} and  \eqref{eq:f-} with \eqref{eq:tilde-f-} we see that
	\begin{align}
	\tilde{f}^{\pm} (t,v,\tilde{z},\tilde{z}^I,\tilde{z}^C; \hat{v}\bigr)= f^{\pm} (t,v+\hat{v},\tilde{z},\tilde{z}^I,\tilde{z}^C; \hat{v}\bigr) + r_D \hat{v}.
	\label{eq:f-compare}
	\end{align}
Next, we can apply the reduction technique developed by \cite{CrepeyRed} to find a continuous BSDE describing the $\XVA$ prior to default.
		
		\begin{theorem}\label{thm:reduction}
			The BSDEs
			\begin{align}\label{eq:reduced}
			-d\check{U}_t^{\buysell} & = \check{g}^{\pm}\bigl(t,\check{U}_t^{\buysell},\check{Z}_t^{\pm}; \hat{V}\bigr) \, dt - \check{Z}^{\pm}_t\, dW_t^{\Qxx}\nonumber  \\
			\check{U}_T^{\buysell} &= 0
			\end{align}
			in the filtration $\mathbb{F}$ with
			\begin{align}
			\check{g}^+\bigl(t,\check{u},\check{z}; \hat{V}\bigr) & := h_I^\Qxx\bigl(\tilde{\theta} _{I}(\hat{V}_t)-\check{u}\bigr) + h_C^\Qxx\bigl(\tilde{\theta} _{C}(\hat{V}_t)-\check{u}\bigr) + \tilde{f}^+\bigl(t, \check{u}, \check{z}, \tilde{\theta} _{I}( \hat{V}_t)-\check{u}, \tilde{\theta} _{C}( \hat{V}_t) -\check{u};\hat{V}\bigr)\label{eq:g+}\\
			\check{g}^-\bigl(t,\check{u},\check{z}; \hat{V}\bigr) &:= - \check{g}^+\bigl(t,-\check{u},-\check{z}; -\hat{V}\bigr),
			\label{eq:g-}
			\end{align}
			admit unique solutions {$\bigl(\check{U}^{\buysell}, \check{Z}^{\pm}\bigr)$}, which are related to the unique solutions $\bigl(\XVA^{\buysell}, \tilde{Z}^{\pm}, \tilde{Z}^{I,\pm}, \tilde{Z}^{C,\pm}\bigr)$ of the BSDEs~\eqref{eq:XVABSDE} via the following relations. On the one hand
			\begin{eqnarray}\label{eq:reduced_identity1}
			\check{U}_t^{\buysell} & := & \XVA_{t\wedge \tau-}^{\buysell},\\
			\check{Z}_t^{\pm} & := & \tilde{Z}_t^{\pm} \ind_{\{t<\tau\}},
			\end{eqnarray}
			are solutions to the reduced BSDE \eqref{eq:reduced}, and on the other hand the solutions to the full $\XVA$ BSDEs given by Eq.~\eqref{eq:XVABSDE} are given by
			\begin{eqnarray}\label{eq:reduced_identity2}
			\nonumber \XVA_t^{\buysell}  &:=& \check{U}_t^{\buysell} \ind_{\{t<\tau\}} + \Bigl( \tilde{\theta} _{C}(\hat{V}_{\tau_C}) \ind_{\{\tau_C < \tau_I\wedge T\} }  + {\tilde{\theta} _{I}(\hat{V}_{\tau_I})} \ind_{\{\tau_I < \tau_C\wedge T\}} \Bigr) \ind_{\{t \geq \tau\}}, \\
			\nonumber \tilde{Z}_t^{\pm}  & := & \check{Z}_t^{\pm}\ind_{\{t < \tau\}}, \nonumber \\
			\nonumber \tilde{Z}_t^{I,\pm} & := & \Bigl(\tilde{\theta} _I(\hat{V}_t) - \check{U}_t^{\buysell}\Bigr)\ind_{\{t \leq \tau\}}, \nonumber\\
			\tilde{Z}_t^{C,\pm} & := & \Bigl(\tilde{\theta} _C(\hat{V}_t) - \check{U}_t^{\buysell}\Bigr)\ind_{\{t \leq \tau\}}.
			\end{eqnarray}\end{theorem}
		
		\begin{proof}
			As the equations~\eqref{eq:reduced} are continuous BSDEs with Lipschitz driver, existence and uniqueness are classical (cf., e.g., \cite[Theorem 2.1.]{ElKaroui}). The equivalence of the full $\mathbb{G}$-BSDEs and the reduced $\mathbb{F}$-BSDEs follows from
			\cite[Theorem 4.3]{CrepeyRed}: In our case we do not change the probability measure and thus condition (A) is satisfied as the $(H)$-hypothesis holds (note that condition (B) also holds in our context). Condition (J) is trivially satisfied as the terminal condition does not depend on the auxiliary processes $\tilde{Z}$, $\tilde{Z}^C$ and $\tilde{Z}^I$. Finally, by the martingale representation theorems with respect to $\mathbb{F}$ and $\mathbb{G}$ (see \cite[Section 5.2]{bielecki01}), the martingale problems considered in \cite{CrepeyRed} and the actual BSDEs considered in the present article have the same unique solutions.
		\end{proof}

The uniqueness of the solutions to the original BSDEs for $V^\pm$ as well as to their projected versions in the $\mathbb{F}$-filtration follows from the definition of $\XVA$:
		
		\begin{corollary}\label{cor:BSDE-V-red}
			Both the BSDE \eqref{eq:BSDE-sell} and \eqref{eq:BSDE-buy} admit a unique solution. These solutions are related to the unique solutions {$\bigl(\bar{U}^\pm, \bar{Z}^{\pm}\bigr)$} of the BSDEs
						\begin{align}\label{eq:reduced-V}
						-d\bar{U}_t^\pm & = g^{\pm}\bigl(t,\bar{U}_t^\pm,\bar{Z}_t^{\pm}; \hat{V}\bigr) \, dt - \bar{Z}^{\pm}_t\, dW_t^{\Qxx}\nonumber  \\
						\bar{U}_T^\pm &= \vartheta
						\end{align}
						in the filtration $\mathbb{F}$ with
						\begin{align*}
						g^+\bigl(t,\bar{u},\bar{z}; \hat{V}\bigr) & := h_I^\Qxx\bigl(\theta_{I}(\hat{V}_t)-\bar{u}\bigr) + h_C^\Qxx\bigl(\theta_{C}(\hat{V}_t)-\bar{u}\bigr) + f^+\bigl(t, \bar{u}, \bar{z}, \theta_{I}( \hat{V}_t)-\bar{u}, \theta_{C}( \hat{V}_t) -\bar{u};\hat{V}\bigr)\\
						g^-\bigl(t,\bar{u},\bar{z}; \hat{V}\bigr) &:= - g^+\bigl(t,-\bar{u},-\bar{z}; -\hat{V}\bigr),
						\end{align*}
via the following relations. On the one hand
					\begin{eqnarray*}
						\bar{U}_t^\pm & := & V_{t\wedge \tau-}^\pm,\\
						\bar{Z}_t^{\pm} & := & Z_t^{\pm} \ind_{\{t<\tau\}},
					\end{eqnarray*}
					are solutions to the reduced BSDEs \eqref{eq:reduced-V}, while on the other hand the solutions to the full BSDEs  given by equations~\eqref{eq:BSDE-sell} and \eqref{eq:BSDE-buy} are given by
					\begin{eqnarray*}
						\nonumber V_t^\pm  &:=& \bar{U}_t^{\buysell} \ind_{\{t<\tau\}} + \Bigl( \theta_{C}(\hat{V}_{\tau_C}) \ind_{\{\tau_C < \tau_I\wedge T\} }  + {\theta_{I}(\hat{V}_{\tau_I})} \ind_{\{\tau_I < \tau_C\wedge T\}} \Bigr) \ind_{\{t \geq \tau\}}, \\
						\nonumber Z_t^{\pm}  & := & \bar{Z}_t^{\pm}\ind_{\{t<\tau\}}, \nonumber \\
						\nonumber Z_t^{I,\pm} & := & \Bigl(\theta_I(\hat{V}_t) - \bar{U}_t^\pm\Bigr)\ind_{\{t\leq \tau\}}, \nonumber\\
						Z_t^{C,\pm} & := & \Bigl(\theta_C(\hat{V}_t) - \bar{U}_t^\pm\Bigr)\ind_{\{t\leq\tau\}}.
					\end{eqnarray*}
		\end{corollary}

\begin{remark}\label{remark:hedge}
We discuss how the replicating strategies of the XVA process can be obtained from the above given results. We use the {tilde symbol $(~~\widetilde{}~~)$} to denote the strategies replicating the XVA process (e.g., $\tilde{\xi}, \tilde{\xi^{I}}, \tilde{\xi^{C}}$ denote, respectively, the number of stock, trader and counterparty bond shares used to replicate $\XVA$). This enables us to distinguish them from the strategies used to replicate the price process $V_t$ of the claim. From the martingale representation theorem and the dynamics of the stock price (see also the discussion at the end of the proof of Proposition 5.3), it follows that
  \begin{align*}
\tilde{\xi}^{\pm}_t &= \frac{\tilde{Z}_t^{\pm}}{\sigma S_t} \ind_{\{t<\tau\}},\quad {\tilde \psi_t^{r,\pm}}  = - \frac{\tilde{\xi}_t^{\pm} S_t}{B_t^{r_r}} \ind_{\{t<\tau\}}.
\end{align*}
From Theorem \ref{thm:reduction} along with equations \eqref{eq:Zetas} and \eqref{eq:hats}, it follows that
\begin{align*}
\tilde{\xi}_{t}^{I,\pm} &= - \frac{\tilde Z_t^{I,\pm}}{P_{t-}^I}  = -\frac{ - L_I  ((1-\alpha)\hat{V}_t )^{+} - \check{U}_t^{\buysell}}{P_{t-}^I} \ind_{\{t {\leq} \tau\}},\\
\tilde{\xi}_{t}^{C,\pm} &= - \frac{\tilde Z_t^{C,\pm}}{P_{t-}^C}  = -\frac{ L_C ((1-\alpha) \hat{V}_t)^{-}- \check{U}_t^{\buysell}}{P_{t-}^C}\ind_{\{t {\leq }\tau\}}.
\end{align*}
From equations~\eqref{eq:rulecoll} and \eqref{eq:collrel} it follows that
\begin{align}
{\tilde\psi}_t^{c} =  - \frac{\alpha \hat{V}_t}{B_t^{r_c}}  \ind_{\{\tau > t\}}.
\end{align}
Finally, from Eq.~\eqref{eq:funding}, {(replacing $V_t^\pm$ with $XVA_t^{\buysell} + \hat{V}_t$)} we obtain
\begin{align}
\tilde{\xi}_t^{f,\pm} &=  \frac{V_t^{\pm} - \hat V_t- \xi_t^I P^I_{t-} - \xi_t^C P^C_{t-}  {- \alpha} \hat{V_t}}{B_t^{r_f}}  \ind_{\{\tau > t\}} \nonumber\\
&=\frac{ L_C ((1-\alpha) \hat{V}_t)^{-}- L_I  ((1-\alpha)\hat{V}_t )^{+} + \check{U}_t^{\buysell} - \alpha \hat V_t}{B_t^{r_f}}  \ind_{\{\tau > t\}},
\label{eq:funding1}
\end{align}
where in the last equation we have used the definitions of $\XVA_t^{\buysell}$ given in equations~\eqref{eq:XVA-buy} and~\eqref{eq:XVA-sell} along with the fact that $\XVA_t^{\buysell} = {\check{U}_t^{\buysell}}$ on the set $\{\tau > t\}$.
\end{remark}
Next, we analyze under which conditions the valuations obtained by solving the BSDEs are arbitrage free.

	\begin{theorem}\label{thm:arb-price}
		Let $\Phi$ be a function of polynomial growth. Assume that
		\begin{equation}\label{eq:A14}
		\rrp \leq \rfp \leq \rrm, \qquad \rfp \leq \rfm, \qquad \rfp \vee r_D < \mu_I \wedge \mu_C,
		\end{equation}
		and
		\begin{equation}\label{eq:comp}
		\rcp \vee \rcm \leq  \rfm \leq \mu_I \wedge \mu_C.
		\end{equation}
		If $V_0^{-} \leq V_0^{+}$, where {$V^{+}$ and $V^{-}$} are the first components of the solutions to the BSDEs \eqref{eq:BSDE-sell} and \eqref{eq:BSDE-buy}, then there exist valuations $\pi^{sup}$ and $\pi^{inf}$, $\pi^{inf} \leq \pi^{sup}$, for the option $\vartheta = \Phi(S_T)$ such that all values in the closed interval $[\pi^{inf}, \pi^{sup}]$ are free of hedger's arbitrage. All valuations strictly bigger than $\pi^{sup}$ and strictly smaller than $\pi^{inf}$ provide then arbitrage opportunities for the hedger. In particular, we have that $\pi^{sup} = V_0^{+}$ and $\pi^{inf} = V_0^{-}$.
	\end{theorem}
	
	\begin{proof}
	First, notice that by virtue of the conditions in \eqref{eq:A14}, the underlying market model is free of hedger's arbitrage. Next, we note that the trader can perfectly hedge a short call position with terminal payoff $\vartheta = \Phi(S_T)$ using the initial capital $V_0^{+}$ {as the polynomial growth of $\Phi$ implies $\vartheta \in L^2(\Omega, \mathcal{F}_T, \Px)$}. Thus it is clear that any value $P > V_0^{+}$ is not arbitrage free, as we could just sell the option for that value, use $V_0^{+}$ to hedge the claim and deposit $P-V_0^{+}$ in the funding account. Using the same argument, we can conclude that buying an option for any value $P<V_0^{-}$ will lead to arbitrage.
	
	Second, assume by contradiction that a valuation $P  \leq V_0^{+}$ would lead to an arbitrage when selling the option. This means that starting with initial capital $P$, the trader can perfectly hedge a claim with terminal payoff $\vartheta' \in \mathcal{F}_T$, where $\vartheta' \geq \vartheta =\Phi(S_T)$ a.s. and $\Px[\vartheta'>\vartheta]>0$. As $h_j^\Qxx = \mu_j - r_D,~j\in\{I,C\}$, we have
	\begin{align*}
		& \phantom{=::} g^+(t,\bar{u}, \bar{z}; \hat{V}_t^{\vartheta'}) - g^+(t,\bar{u}, \bar{z}; \hat{V}_t^{\vartheta}) \\
		& = h_I^\Qxx \bigl( \theta_I(\hat{V}_t^{\vartheta'}) - \theta_I(\hat{V}_t^{\vartheta})\bigr) + h_C^\Qxx \bigl( \theta_C(\hat{V}_t^{\vartheta'}) - \theta_C(\hat{V}_t^{\vartheta})\bigr) \\
		& \phantom{=::}+ \bigl(\rfm-\rfp\bigr)\Bigr(\bigl(\theta_I(\hat{V}_t^{\vartheta'}) + \theta_C(\hat{V}_t^{\vartheta'}) - \bar{u} -\alpha \hat{V}_t^{\vartheta'}\bigr)^+ - \bigl(\theta_I(\hat{V}_t^{\vartheta}) + \theta_C(\hat{V}_t^{\vartheta}) - \bar{u} -\alpha \hat{V}_t^{\vartheta}\bigr)^+\Bigr)\\
		& \phantom{=::}- \rfm\Bigl(\theta_I(\hat{V}_t^{\vartheta'}) - \theta_I(\hat{V}_t^{\vartheta}) + \theta_C(\hat{V}_t^{\vartheta'}) - \theta_C(\hat{V}_t^{\vartheta})  - \alpha \bigl( \hat{V}_t^{\vartheta'} - \hat{V}_t^{\vartheta}\bigr) \Bigr)\\
		& \phantom{=::}+ r_D\Bigl(\theta_I(\hat{V}_t^{\vartheta'}) - \theta_I(\hat{V}_t^{\vartheta}) + \theta_C(\hat{V}_t^{\vartheta'}) - \theta_C(\hat{V}_t^{\vartheta})\Bigr)\\
		& \phantom{=::}+ \alpha \rcm \Bigl((\hat{V}_t^{\vartheta'})^- - (\hat{V}_t^{\vartheta})^-\Bigr) - \alpha \rcp \Bigl(\bigl(\hat{V}_t^{\vartheta'}\bigr)^+ - \bigl(\hat{V}_t^{\vartheta}\bigr)^+\Bigr)\\
		& \geq \bigl(\mu_I - \rfm\bigr) \bigl({\theta}_I(\hat{V}_t^{\vartheta'}) - {\theta}_I(\hat{V}_t^{\vartheta})\bigr) + \bigl(\mu_C - \rfm\bigr) \bigl( { \theta}_I(\hat{V}_t^{\vartheta'}) - {\theta}_I(\hat{V}_t^{\vartheta})\bigr) \\
		&\phantom{=::}  + \alpha \bigl({\rfm}- (\rcp  \vee \rcm)  \bigr)\bigl(\hat{V}_t^{\vartheta'} - \hat{V}_t^{\vartheta}\bigr) \geq 0.
	\end{align*}

To deduce the first inequality, we have used that the term multiplying $r_f^- - r_f^+$ is positive. This can be directly seen from the direct computation below
	\begin{align*}
		& \theta_I(\hat{V}_t^{\vartheta'}) + \theta_C(\hat{V}_t^{\vartheta'}) -\alpha \hat{V}_t^{\vartheta'} = (1-\alpha)(1-L_I)\bigl(\hat{V}_t^{\vartheta'}\bigr)^+ - (1-\alpha)(1+L_C)\bigl(\hat{V}_t^{\vartheta'}\bigr)^-\\
		\geq  & (1-\alpha)(1-L_I)\bigl(\hat{V}_t^{\vartheta}\bigr)^+ - (1-\alpha)(1+L_C)\bigl(\hat{V}_t^{\vartheta}\bigr)^-
		= \theta_I(\hat{V}_t^{\vartheta}) + \theta_C(\hat{V}_t^{\vartheta}) -\alpha \hat{V}_t^{\vartheta}.
	\end{align*}

    To deduce the last inequality, we have used~\eqref{eq:comp}. Thus, we can apply the comparison principle for $\mathbb{F}$-BSDEs \cite[Theorem 2.2]{ElKaroui} to $\bar{U}$ and then notice that $\vartheta' \geq \vartheta$ implies that $\hat{V}_t^{\vartheta'} \geq \hat{V}_t^{\vartheta}$, which in turn leads to the following inequality between the closeout terms: $\theta_C(\hat{V}_t^{\vartheta'}) \geq \theta_C(\hat{V}_t^{\vartheta})$ and $\theta_I(\hat{V}_t^{\vartheta'}) \geq \theta_I(\hat{V}_t^{\vartheta})$ (see also Eq.~\eqref{eq:thetaIC} for their definitions). It thus follows that $V_t^{\vartheta'} \geq V_t^{\vartheta}$ and in particular $P > V_0^+$ (using strict comparison, i.e., $P = V_0$ implies $\vartheta' =\vartheta$ a.s.), contradicting our assumption. Using a symmetric argument, it follows that $P  \geq V_0^{-}$.	Thus, if $V_0^{-} \leq V_0^{+}$, we can conclude that all valuations in the interval $[\pi^{inf} = V_0^{-} , V_0^{+} = \pi^{sup}]$ are
    arbitrage-free, whereas no arbitrage free valuation exists if $V_0^{-} > V_0^{+}$.
	\end{proof}

We note that the width of the no-arbitrage interval can be described both in terms of the valuations of the claim and of the $\XVA$ as
	\[
	\XVA_0^{\sell} - \XVA_0^{\buy} = V^+_0 - V^-_0 = \pi^{sup} - \pi^{inf}.
	\]

The reduced BSDE also enables us to provide a PDE representation {for the case ${\vartheta} = \Phi(S_T)$} using just classical arguments based on the nonlinear Feynman-Kac formula. We provide such a representation in the following proposition whose proof is reported in the appendix.
\begin{proposition}\label{thm:PDE-sol}
 The two-dimensional semilinear Cauchy problem
   \begin{alignat}{3}
    -u^\pm_t + \mathcal{L} u^\pm &= g^{\pm}\bigl(t, u^\pm, \sigma (u^\pm_x + \hat u_x);\hat u\bigr), \qquad \qquad && u^\pm(T, x)  && =0, \nonumber\\
    -  {\hat w}_t + \mathcal{L}  {\hat w} & =  - r_D  {\hat w},  &&  {\hat w}(T, x) &&= \Phi(e^x),\label{eq:PDE_representation}
  \end{alignat}
where the differential operator $\mathcal{L}$ is defined by
  \[
    \mathcal{L} := -\Bigl(r_D - \frac{\sigma^2}{2}\Bigr) \partial_x-\frac{\sigma^2}{2} \partial_{xx},
  \]
 admits the unique viscosity solutions $u^{\pm}$ and it holds that $u^\pm(t, \log{(S_t)}) = \XVA_t^{\buysell}\ind_{\{t<\tau\}}$. Moreover, if $\Phi$ is piecewise continuously differentiable with $\Phi$ and $\Phi'$ (where defined) having at most polynomial growth, then the Cauchy problems \eqref{eq:PDE_representation} has classical solutions.
\end{proposition}

\begin{remark}
The transformation from the BSDE \eqref{eq:BSDE-sell} (or \eqref{eq:BSDE-buy}) to the PDE \eqref{eq:PDE_representation} was derived by projecting a $\mathbb{G}$-BSDE to a $\mathbb{F}$-BSDE and then applying the Feynman-Kac formula to it. One can follow an alternative route, deriving first a PDE in four variables related to the BSDE with jumps via a nonlinear Feynman-Kac formula,
and then reducing the dimension. This is shown in detail in \cite{BCS2}, and we provide here the gist of the argument.

On $\{t<\tau\}$, define the pre-default solutions $v^{\pm}(t, S_t, \varpi_t^{I,\Qxx},\varpi_t^{C,\Qxx})=V^{\pm}_t \ind_{\{\tau > t\}}$ to  the  BSDEs \eqref{eq:BSDE-sell} and \eqref{eq:BSDE-buy}. The existence of these measurable functions $v^{\pm}$, i.e., the fact that $V^{\pm}$ are Markovian, {follows from} Proposition 4.1.1 in \cite{Delong}. Specifically, Theorem 3.2 in \cite{BCS2} shows that $v^{\pm}$  satisfy
\begin{align}
\nonumber &-v_t^{\pm}- \sum_{j\in\{I,C\}} h_j^{\Qxx}\Bigl( \theta_{j}\bigl(\hat v(t,s)\bigr) -v^{\pm}(t,s,w^I ,w^C)  - v_{j}^{\pm}\Bigr)- r_D sv_s^{\pm} - \frac12\sigma^2 s^2v_{ss}^{\pm}\\
&-f^{\pm}\Bigl(t, v^{\pm}, \sigma sv_s^{\pm}(t,s,w^I ,w^C), \theta_{I}\bigl( \hat v(t,s)\bigr)  - v^{\pm}(t,s,w^I ,w^C), \theta_{C}\bigl( \hat v(t,s)\bigr) -v^{\pm}(t,s,w^I ,w^C);\hat v(t,s)\Bigr)=0,\nonumber\\
&v^{\pm}(T, s, \cdot, \cdot) = \Phi(s)\label{eq:PDE}
\end{align}
in the viscosity sense. In the above expression, we have used the notation $v_{i}^{\pm} = \frac{\partial v^{\pm}}{\partial w^i},~i\in\{I,C\}$. Additionally, Theorem 3.2 in \cite{BCS2} shows that $v^{\pm}$ are the unique viscosity solutions of the PDEs \eqref{eq:PDE} satisfying
\begin{equation*}
\lim\limits_{\abs{x}\to\infty} \abs{v^{\pm}(\cdot,e^x, \cdot, \cdot)} e^{-c\log^2\abs{x}}=0,~c>0.
\end{equation*}
The above step transfers the BSDEs \eqref{eq:BSDE-sell} and \eqref{eq:BSDE-buy} to the PDEs \eqref{eq:PDE}. The key step in the PDE domain that parallels the transformation from the original BSDEs to the reduced BSDEs is  given in Remark 3.3 of \cite{BCS2}. Essentially, we are only concerned with $V^{\pm}_t$ before any default occurs, hence we do no need to keep track of the martingale terms $\varpi_t^{j,\Qxx}$'s. These are only needed to track the occurrence of a default. If no default has occurred, $\varpi^{j,\Qxx}_t = -h_j^\Qxx t$. It then follows that $v^\pm$ becomes a function of only two variables, i.e. $\vv^{\pm}(t,s) := v^\pm(t,s, -h_I^\Qxx t, -h_C^\Qxx t)$. The PDEs \eqref{eq:PDE} become
\begin{align}
\nonumber&-\vv_t ^{\pm}+ (h_I^\Qxx+h_C^\Qxx)\vv^{\pm}(t,s)  - r_D s\vv_s^{\pm} - \frac12\sigma^2 s^2\vv_{ss}^{\pm}\\
&\quad-f^{\pm}(t, \vv^{\pm}, \sigma s\vv_s^{\pm}(t,s), \theta_{I}( \hat v(t,s))  - \vv^{\pm}(t,s) \theta_{C}( \hat v(t,s)) -\vv^{\pm}(t,s);\hat v(t,s))=\sum_{j\in\{I,C\}} h_j^\Qxx \theta_{j}( \hat v(t,s)),\nonumber\\
&\vv^{\pm}(T, s) = \Phi(s)\label{eq:PDE1}.
\end{align}
From here, applying the standard change of variables $x=\log s, $  $\bar \w^{\pm}(t,x) = \bar{v}^{\pm}(t,e^x)$ and $\hat \w(t,x) = \hat v(t, e^x)$, we obtain the PDEs
\begin{align}
\nonumber&-\bar \w_t^{\pm} - \Bigl(r_D - \frac{\sigma^2}{2}\Bigr) \bar \w_x^{\pm} - \frac12\sigma^2 \bar \w_{xx}^{\pm} + \bigl(h_I^\Qxx+h_C^\Qxx\bigr) \bar \w^{\pm} \\
& -f^{\pm}\Bigl(t, \bar \w^{\pm}(t,x), \sigma \bar \w_x^{\pm}(t,x), \theta_{I}\bigl( \hat \w(t, x)\bigr)  - \bar \w^{\pm}(t,x), \theta_{C}\bigl( \hat \w(t, x)\bigr) -\bar \w^{\pm}(t,x);{\hat \w(t, x)}\Bigr) \nonumber \\
& = \sum_{j\in\{I,C\}} h_j^\Qxx \theta_{j}\bigl( {\hat \w(t, x)} \bigr),\nonumber\\
&\bar \w^{\pm}(T, x) = \Phi(e^x)\label{eq:PDE2}.
\end{align}
The PDEs \eqref{eq:PDE_representation} follow by deriving the PDEs for $\w^{\pm} = \bar \w^{\pm} - \hat \w$, which are associated with the BSDEs representation of the $\XVA$ given by equations~\eqref{eq:XVA-sell} and~\eqref{eq:XVA-buy}. Hence, we can express $\XVA$ as a solution to a Cauchy problem for a two-dimensional system of semilinear PDEs,
\begin{align}
&-\w^{\pm}_t + \mathcal{L} \w^{\pm} = f^{\pm}\Bigl(t, \w^{\pm} + \hat \w, \sigma (\w^{\pm}_x + \hat \w^{\pm}_x), \hat\theta_{I}( \hat \w)  - \w^{\pm}, \hat\theta_{C}( \hat \w) -\w^{\pm};\hat \w\Bigr) \nonumber \\
& + \sum_{j\in\{I,C\}} h_j^\Qxx \bigl( \hat \theta_{j}( \hat \w) - \w^{\pm} \bigr) + r_D \hat \w ,~~\w^{\pm}(T,x)=0,\nonumber\\
&- \hat{\w}_t + \mathcal{L} \hat{\w}  +r_D \hat \w= 0,~~\hat{\w}(T, x) = \Phi(e^x).\label{eq:PDE_representation1}
\end{align}
Recalling the definitions of $f^{\pm}, \tilde f^{\pm},$ and $g^{\pm}$ and their relations given in
equations~\eqref{eq:f+}, \eqref{eq:f-}, \eqref{eq:tilde-f+}, \eqref{eq:tilde-f-}, \eqref{eq:f-compare} and \eqref{eq:g+}, \eqref{eq:g-} we obtain that
equations~\eqref{eq:PDE_representation1} and \eqref{eq:PDE_representation} coincide. At this point, we can conclude that the unique solutions $\w^{\pm}$ to the PDEs \eqref{eq:PDE_representation1} are only in the viscosity sense. If $\Phi$ is piecewise continuously differentiable and also $\Phi'$ (where defined) {has} at most polynomial growth, then it follows from Theorem \ref{thm:PDE-sol} that the uniqueness of the solutions also holds in the classical sense. Moreover, we can apply Theorem A.1 of \cite{BCS2} and obtain that, on the set $\{t<\tau\}$, we have that
  \begin{align}
    {\tilde{Z}_t^{\pm}} &= \sigma S_t \w^{\pm}_x\bigl(t,\log(S_t)\bigr) \nonumber,\\
    {\tilde Z_t^{i,\pm}} &= \theta_i\bigl(\hat V(t, S_t)\bigr) - {\w^{\pm}(t,\log(S_{t}))},
\qquad \qquad i \in \{I,C\}.
  \label{eq:stratexpl}
  \end{align}
\end{remark}

\section{Explicit examples}\label{sec:expexp}
We specialize our framework to deal with a concrete example for which we can provide fully explicit expressions for the total valuation adjustment. More specifically, we consider \cite{Piterbarg}'s model {as well as} an extension of it, accounting for counterparty credit risk and closeout costs. This means that defaultable bonds of trader and counterparty become an integral part of the hedging strategy. Throughout the section, we make the following assumptions on the interest rates, as in Piterbarg{'}s setup:
\begin{equation}\label{eq:Piterrates}
\rfp = \rfm = r_f, \nonumber \qquad \rcp = \rcm = r_c, \nonumber \qquad r_D = \rrp = \rrm = r_r.
\end{equation}
We also assume that $r_f>r_r>r_c$, the case to be expected in practice according to \cite{Piterbarg}. Under the above assumptions, the security, funding and collateral accounts {do not depend } on whether the position in the security is long or short, whether the amount is borrowed from or lent to the treasury, and whether the collateral is posted or received. Due to the symmetry between rates, the buyer{'}s and seller{'}s
$\XVA$ coincide, and hence we can drop both the plus and minus superscripts in the BSDEs. The difference between the discount rate $r_D$ chosen by the valuation agent and the repo rate may also be interpreted as a proxy for illiquidity of the repo market. Under this interpretation, $r_D=r_r$ corresponds to a regime of full liquidity. The BSDEs become linear and lead to explicit additive decompositions of $\XVA$ in terms of different adjustments as detailed in the sequel. We also remark that similar decompositions have been obtained by \cite{BrigoPerPal}, see Theorem 4.3 and subsequent remarks therein.

\subsection{Piterbarg's model} \label{sec:pitdef}
This section provides an explicit representation of $\XVA$ and the associated hedging strategy in the framework proposed by \cite{Piterbarg}. Besides symmetry between rates, \cite{Piterbarg} precludes the possibility of defaults in the model, but maintains the presence of collateral. Before presenting the main result of this section, we introduce the following quantities
\begin{equation*}
P_t^{r_f}:= \frac{B_t^{r_f}}{B_T^{r_f}}, \qquad \qquad P_t^{r_r}:= \frac{B_t^{r_r}}{B_T^{r_r}},
\end{equation*}
which can be understood as the time $t$ prices of two (fictitious) risk-free bonds with discount factors $r_f$ and $r_r$ respectively. The next proposition, proven in the appendix, provides an explicit representation of XVA and its replicating strategy.
\begin{proposition} \label{prop:Pitnodef}
The total valuation adjustment is given by
\begin{equation}\label{eq:FVAPiterbarg}
\XVA_t = \Biggl ({\frac{P_t^{r_f}}{P_t^{r_r}}} - 1 \Biggr) \Biggl(1-\alpha \frac{r_f - r_c}{r_f - r_r} \Biggr) \hat{V}_t := \beta_t \hat{V}_t.
\end{equation}
Furthermore, the replication strategy in stock is given by
\begin{align}
\tilde{\xi}_t & = \beta_t \hat{\Delta}_t,
\label{eq:unitedDeltaPiterbarg}
\end{align}
where
\begin{equation}
    \hat{\Delta}_t = \hat{\Delta}(t,S_t) := \frac{\partial}{\partial S} \hat{V}(t, S_t) = \frac{\partial}{\partial S} \mathbb{E}^{\Qxx}\biggl[\frac{B_t^{r_r}}{B_T^{r_r}}  \Phi(S_T) \bigg\vert \mathcal{F}_t\biggr].
\label{eq:delta}
\end{equation}
\end{proposition}

The representation~\eqref{eq:FVAPiterbarg} expresses $\XVA$ as a percentage of the publicly available price $\hat{V}_t$ of the claim. Moreover, Proposition \ref{prop:Pitnodef} shows that funding costs enter in two ways in the $\XVA$ and the corresponding replicating strategy in the stock: (1) funding occurs at the rate $r_f$, but $\XVA$ and the hedging strategy are based on public valuation using $r_r$ as the discount rate and (2) a funding adjustment proportional to the size of posted collateral, $\alpha \hat{V}(t,S_t)$, originating from the difference between funding and collateral rates.

\begin{remark}\label{remark:P-comparison}
The model considered in this section is a special case of the model proposed by \cite{Piterbarg}. Specifically, the model in \cite{Piterbarg} is more general because the rates $r_r, r_f$ and $r_c$ can be stochastic processes, and he also allows for a general collateral specification. Differently from us, \cite{Piterbarg} does not define and study XVA, but rather focuses on determining the price of the claim under the above mentioned assumptions. Under the symmetry assumption on rates and using the collateralization rule given in Eq.~\eqref{eq:rulecoll}, it can seen that equation (3) of \cite{Piterbarg} is exactly the solution of our BSDE~\eqref{eq:BSDE-sell}, which admits the explicit representation
\begin{align}
    \nonumber \bigl(B_t^{r_f} \bigr)^{-1} V_t &= \mathbb{E}^{\Qxx}\Bigl[\bigl.\bigl(B_T^{r_f}\bigr)^{-1} \Phi(S_T) \bigr\vert \mathcal{F}_t \Bigr] + \alpha (r_f - r_c) \int_t^T
    \bigl(B_s^{r_f} \bigr)^{-1} B_s^{r_r} \mathbb{E}^{\Qxx}\Bigl[\bigl. \bigl(B_s^{r_r}\bigr)^{-1} \hat{V}_s \bigr\vert \mathcal{F}_t \Bigr] d{ s} \\
    &= \bigl(B_T^{r_f}\bigr)^{-1} B_T^{r_r} \frac{1}{B_t^{r_r}} \hat{V}_t + \alpha (r_f - r_c) \int_t^T
    \bigl(B_s^{r_f} \bigr)^{-1} B_s^{r_r} \bigl(B_t^{r_r}\bigr)^{-1} \hat{V}_t ds.
\label{eq:BSDE1}
\end{align}
The case of zero collateral in \cite{Piterbarg} corresponds to setting $\alpha=0$ in our case. It is also possible to derive equation (5) of \cite{Piterbarg}, in which it is assumed that the rate used for discounting the claim is the collateral rate $r_c$.
To see this, consider Eq.~\eqref{eq:BSDE-sell}, which under this setting (and using Eq.~\eqref{eq:rulecoll}) becomes
\begin{align*}
dV_t &= r_f(V_t - \alpha \hat V_t) + r_c \alpha \hat V_t - Z_t dW_t^{\Qxx}=r_c V_t +  (r_f -r_c) (V_t - C_t) - Z_t dW_t^{\Qxx}.
\end{align*}
We recall that $C_t$, $t \geq 0$, denotes the collateral process. {It then} follows that
\begin{align*}
\bigl(B_t^{r_c}\bigr)^{-1} V_t &= \bigl(B_T^{r_c}\bigr)^{-1} \Phi(S_T)  - \int_t^T \bigl(B_s^{r_c}\bigr)^{-1} {Z}_s dW_s^{\Qxx} + \bigl({r_f - r_c}\bigr)  \int_t^T \bigl(B_s^{r_c}\bigr)^{-1} (V_s - C_s) ds.
\end{align*}
Taking {the conditional} expectation, and noticing that $Z$ (computable in a similar way as in the proof of Proposition \ref{prop:Pitnodef})
is square integrable and therefore the stochastic integral above is a true martingale, we get
\begin{align}
 \bigl(B_t^{r_c}\bigr)^{-1} V_t &= \mathbb{E}^{\Qxx}\Bigl[\bigl.\bigl(B_T^{r_c}\bigr)^{-1} \Phi(S_T) \bigr\vert \mathcal{F}_t \Bigr] +  (r_f - r_c) \int_t^T
\bigl(B_s^{r_c} \bigr)^{-1} \mathbb{E}^{\Qxx}\bigl[\bigl. V_s -C_s \bigr\vert \mathcal{F}_t \bigr] d{ s}. \label{eq:BSDE2}
\end{align}
Deviating from our framework and assuming that the collateral rule is based on the hedger's valuation, we can use similar arguments as above and recover equation (7) of \cite{Piterbarg} in the case of full collateralization ($C_t = V_t$). The above analysis serves to illustrate the generality of our framework, in which special tractable cases can be recovered by suitable specifications of the model parameters.
\end{remark}

We next analyze the dependence of $\XVA$ on funding rates and collateralization levels. Figure \ref{fig:nodef} shows that in the case of a European call option the $\XVA$ is negative when the collateralization level is small. This is consistent with the expression~\eqref{eq:FVAPiterbarg} and can be understood as follows.
Suppose that $\alpha = 0$. In Piterbarg's model, the hedger of the option is long in stock and finances purchases of stock shares at the repo rate $r_r$. He is also long in the funding account and accrues interests at the higher rate $r_f$. In the Black-Scholes world, the seller buys stock shares and invests cash, both at the rate $r_r = r_D$. Hence, if $r_f > r_r$, the existence of the funding account is advantageous for the hedger. As a result, the hedger's price will be lower than the Black-Scholes price so that the $\XVA$ is negative. As $\alpha$ gets larger, the trader also needs to finance purchases of the collateral to post to the counterparty. In order to do so, he borrows from the treasury at the rate $r_f$. However, he only receives interest at rate $r_c$ on the posted collateral. This yields a loss to the trader given that $r_c < r_r < r_f$. Figure \ref{fig:nodef} confirms our intuition. It also shows that the stock position of the trader decreases as the funding rate $r_f$ increases, and increases if the collateralization level $\alpha$ increases.
When $\XVA$ is negative, i.e., the hedger{'s} price {is} lower than the Black-Scholes price, then the
strategy of the trader is to go short in the stock.

\begin{figure}
\includegraphics[width=0.49\textwidth]{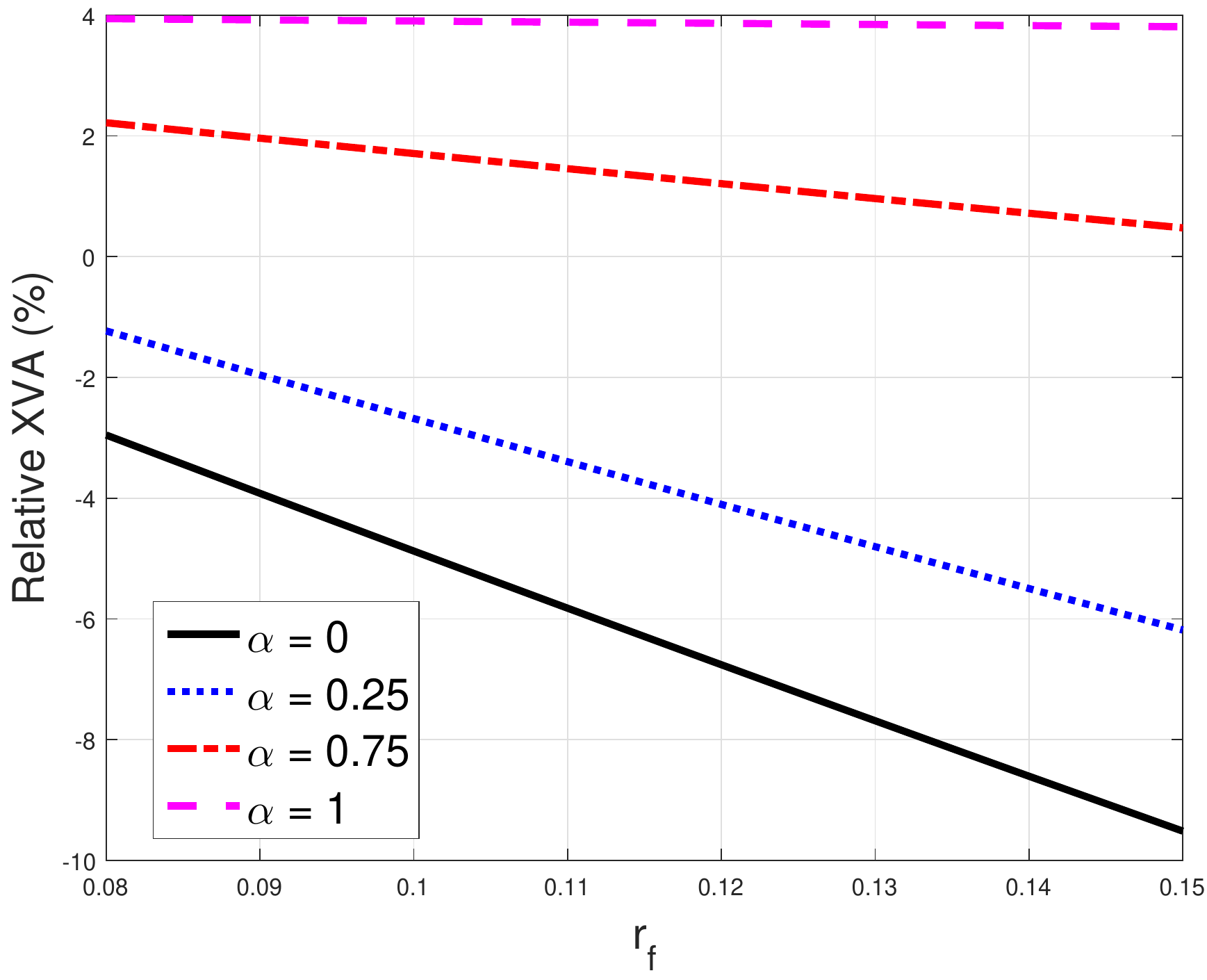}
\includegraphics[width=0.49\textwidth]{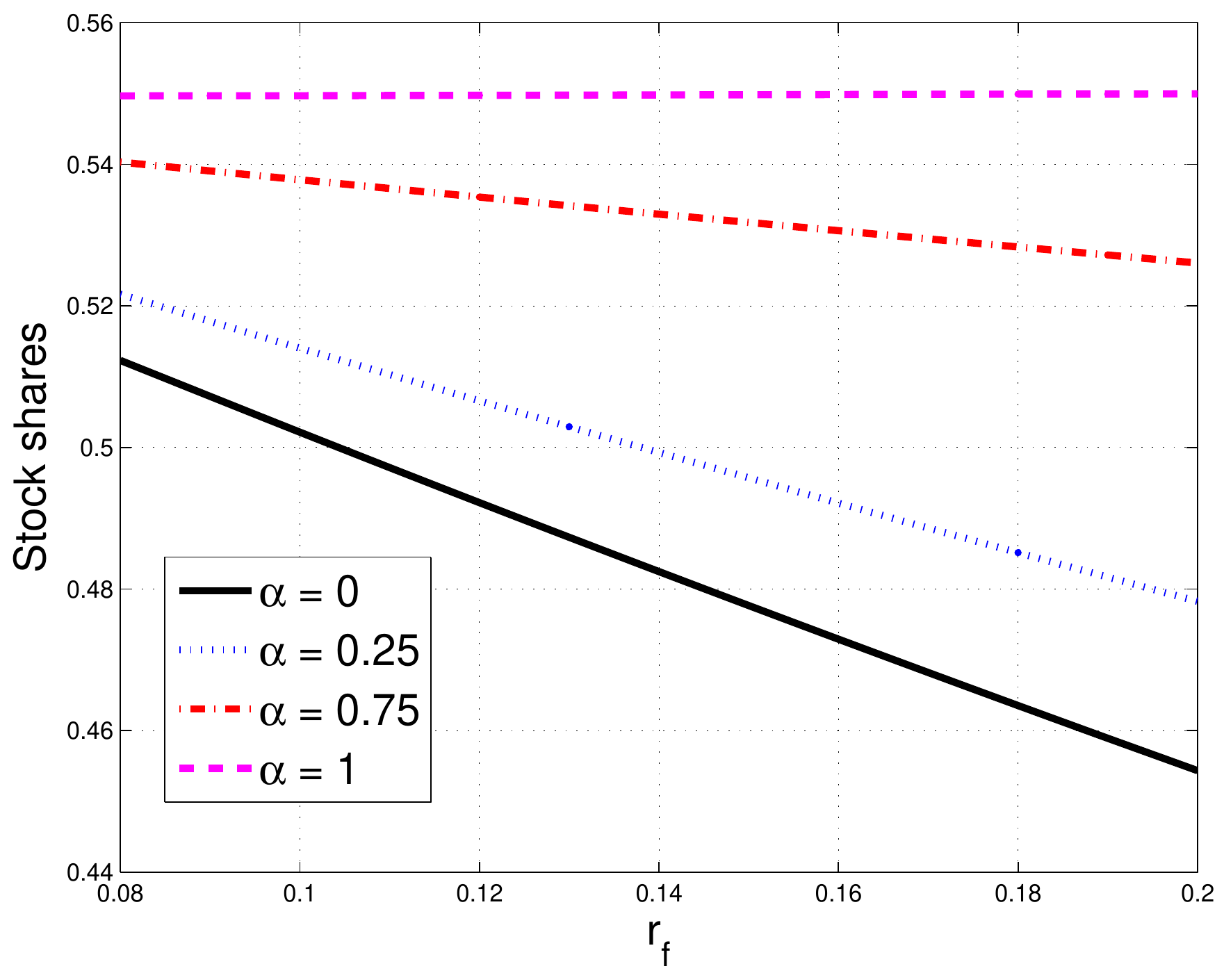}
\caption{Left: $\XVA$  as a function of $r_f$ for different collateralization levels $\alpha$. Right: Number of stock shares in the $\XVA$ replication strategy. We set $r_D = 0.05$, $r_c = 0.01$, $\sigma = 0.2$, and $\alpha = 0$. The claim is an at-the-money European call option with maturity $T=1$.}
\label{fig:nodef}
\end{figure}

\subsection{Piterbarg's model with defaults}
In this section, we generalize Piterbarg's model by including the possibility that the trader or his counterparty default. Proposition \ref{prop:Pitdef} gives the explicit expression for $\XVA$ of {a} European claim, as well as closed-form expressions for the replicating strategies in stocks and bonds. We specialize this result to the case of an option in Remark \ref{rem:options}.

\begin{proposition} \label{prop:Pitdef}
The total valuation adjustment is given by
\begin{align}\label{eq:FVAPiterbargdef}
\XVA_t \ind_{\{\tau > t\}} &= {\Bigl((r_r -r_f) + \alpha (r_f-r_c)\Bigr)}\frac{1- e^{-(\eta-r_r) (T-t)} }{ \eta-r_r} \hat{V}_t \ind_{\{\tau > t\}} \nonumber \\
& \phantom{=:} + ({\mu_C - r_f} )   L_C  \frac{1- e^{-(\eta-r_r) (T-t)} }{ \eta-r_r} \bigl((1-\alpha) \hat{V}_t \bigr)^{-} \ind_{\{\tau > t\}}\nonumber\\
& \phantom{=:} - ({\mu_I - r_f}  )   L_I \frac{1- e^{-(\eta-r_r) (T-t)} }{ \eta-r_r}  \bigl((1-\alpha) \hat{V}_t \bigr)^{+} \ind_{\{\tau > t\}},
\end{align}
where $\eta := {\mu_I + \mu_C -r_f}$. Furthermore, the $\XVA$ replication strategies in stock, counterparty and trader's bonds are given by
\begin{align}\label{strats}
\tilde{\xi}_t=& \biggl((r_r -r_f) + \alpha (r_f-r_c)\frac{1- e^{-(\eta-r_r) (T-t)}}{ \eta-r_r}  \nonumber \\
 & \phantom{=}+ ( {\mu_C - r_f} )   L_C  \frac{1- e^{-(\eta-r_r) (T-t)} }{ \eta-r_r} (1-\alpha) {\ind_{\{\hat{V}_t < 0\}}} \nonumber\\
& \phantom{=} - ( {\mu_I - r_f} )   L_I \frac{1- e^{-(\eta-r_r) (T-t)} }{ \eta-r_r}  (1-\alpha) {\ind_{\{\hat{V}_t > 0\}}}
 \biggr )  \hat{\Delta}_t \ind_{\{\tau > t\}}, \nonumber \\
\tilde{\xi}_t^I &= \frac{{{\XVA_t} + L_I (1-\alpha) (\hat{V}_t)^+}}{{P_t^I}} \ind_{\{\tau > t\}},\nonumber \\
\tilde{\xi}_t^C &= \frac{{\XVA_t - L_C  (1-\alpha) (\hat{V}_t)^- }}{{P_t^C}} \ind_{\{\tau > t\}}.
 \end{align}
Moreover, on the event $\{t=\tau_I<\tau_C\wedge T\}$, the XVA process is given by
\begin{equation}\label{eq:defaul-I}
\XVA_t =  \tilde{\xi}_t^f B_t^{r_f} + \tilde{\xi}_t^C P^C_t - {\tilde{\psi}_t^c} B_t^{r_c} - \hat{V}_t = - L_I  ((1-\alpha)\hat{V}_t )^{+},
\end{equation}
while on the event $\{t=\tau_C<\tau_I\wedge T\}$ {by}
\[
\XVA_t = \tilde{\xi}_t^f B_t^{r_f} + \tilde{\xi}_t^I P^I_t - \tilde{\psi}_t^c B_t^{r_c} - \hat{V}_t =  L_C  ((1-\alpha)\hat{V}_t )^{-}.
\]

\end{proposition}
Recall that under the no-arbitrage conditions given in {Assumption}~\ref{ass:necessary} we have that $\mu_I > r_f$, $\mu_C > r_f$ and hence $\eta > r_f > r_r$. As a consequence, the expression for the XVA given in Eq.~\eqref{eq:FVAPiterbargdef} is well defined. {Moreover,} the number of shares held in the repo, collateral and funding account ($\tilde{\psi}_t^r,\tilde{\psi}_t^{c}, \tilde{\xi}_t^{f}$) are uniquely determined by the holdings in stock, investor and counterparty bonds (see Remark \ref{remark:hedge}).

As in the classical Piterbarg setup, the representation~\eqref{eq:FVAPiterbargdef} shows that $\XVA$ admits a decomposition, but now into three separate contributing terms. The first term corresponds to the replication in the absence of defaults. {It captures the costs of replicating the publicly available claim $\hat{V}(t,S_t)$ and of funding the collateral posted until the earliest of investor or counterparty default.} The second term corresponds to the (funding-adjusted) replicating cost of the $\CVA$ component, and the third term to the (funding-adjusted)
replication cost of the $\DVA$ component.

\begin{remark} \label{rem:options}
Eq.~\eqref{eq:FVAPiterbargdef} reduces to
\begin{align}\label{eq:reduced-sol4}
\XVA_{t} \ind_{\{\tau > t\}} &= \Bigl( \underbrace{(r_r -r_f) + \alpha (r_f-r_c)}_{{\text{funding}}} \underbrace{- L_I (1-\alpha) ({\mu_I - r_f} )}_{{\DVA}} \Bigr)
\frac{1- e^{-(\eta-r_r) (T-t)} }{ \eta-r_r} \hat{V}_t \ind_{\{\tau > t\}}  \nonumber \\
&:= A \hat{V}_t \ind_{\{\tau > t\}}.
\end{align}
Indeed, $\XVA$ is now expressed as a percentage of the publicly available price $\hat{V}_t$ of the claim. As the trader is short in the call option and hence needs to replicate a long position to hedge, he always faces zero exposure to the counterparty (the trader needs to post collateral to the counterparty, but the latter does not need to do so with the trader), and hence he only needs to replicate
the $\DVA$ component of the trade which is not mitigated by the posted collateral. In this case the number of stock and bond shares in~\eqref{strats} needed for his replicating strategy simplify to
\begin{align*}
    \tilde{\xi}_{{t}} = A \times \frac{\partial}{\partial S} \hat{V}(t,S_t)\ind_{\{\tau > t\}} ,\qquad
    \tilde{\xi}_t^I =  \frac{{A \times \hat{V}_t} + L_I (1-\alpha) \hat{V}_t} {{P_t^I} }\ind_{\{\tau > t\}}, \qquad
    \tilde{\xi}_t^C =  \frac{{A \times \hat{V}_t} }{{P_t^C}}\ind_{\{\tau > t\}}.
\end{align*}
Moreover, on the event $\{t=\tau_I<\tau_C\wedge T\}$ the value of the portfolio is $\XVA_t  = \bigl(1- L_I  ((1-\alpha)\bigr)\hat{V}_t$, and similarly on $\{t=\tau_C<\tau_I\wedge T\}$ {we have} $\XVA_t = \hat V_t$.
Despite the $\CVA$ component of the trader is absent, the hedger still trades in the counterparty bond. This is because he needs to hedge the default risk of his counterparty given that the claim is being replicated until the earliest of hedger and counterparty default time.
\end{remark}

We conclude with a numerical assessment of the above derived results. We consider one at-the-money option with $S_0=K=1$, so that at maturity $T=1$ we have the payoff $\Phi(x) = (x-K)^+$. Figure \ref{fig:pricedec1} reports the value of the funding and $\DVA$ component contributing to the decomposition given in Eq.~\eqref{eq:reduced-sol4}. In the safer scenario (left panel), the funding component becomes predominant as $r_f$ increases, while the contribution coming from the replication costs of the closeout position at the trader's default time is small. As the default risks of trader and
counterparty become higher (right panel) and for not too high funding rates $r_f$, the $\DVA$ component dominates given that the closeout procedure is triggered earlier and the closeout payoff is larger.

Comparing the bottom panels of Figures \ref{fig:defXVA} and \ref{fig:defXVArisk}, we see that in both cases a similar number of trader{'}s bond shares is used to replicate the jumps to the closeout values. However, the return on the trader{'}s bond under the valuation measure is higher under the risky scenario, and hence offers a larger contribution to XVA.
  \begin{figure}
    \includegraphics[width=0.49\textwidth]{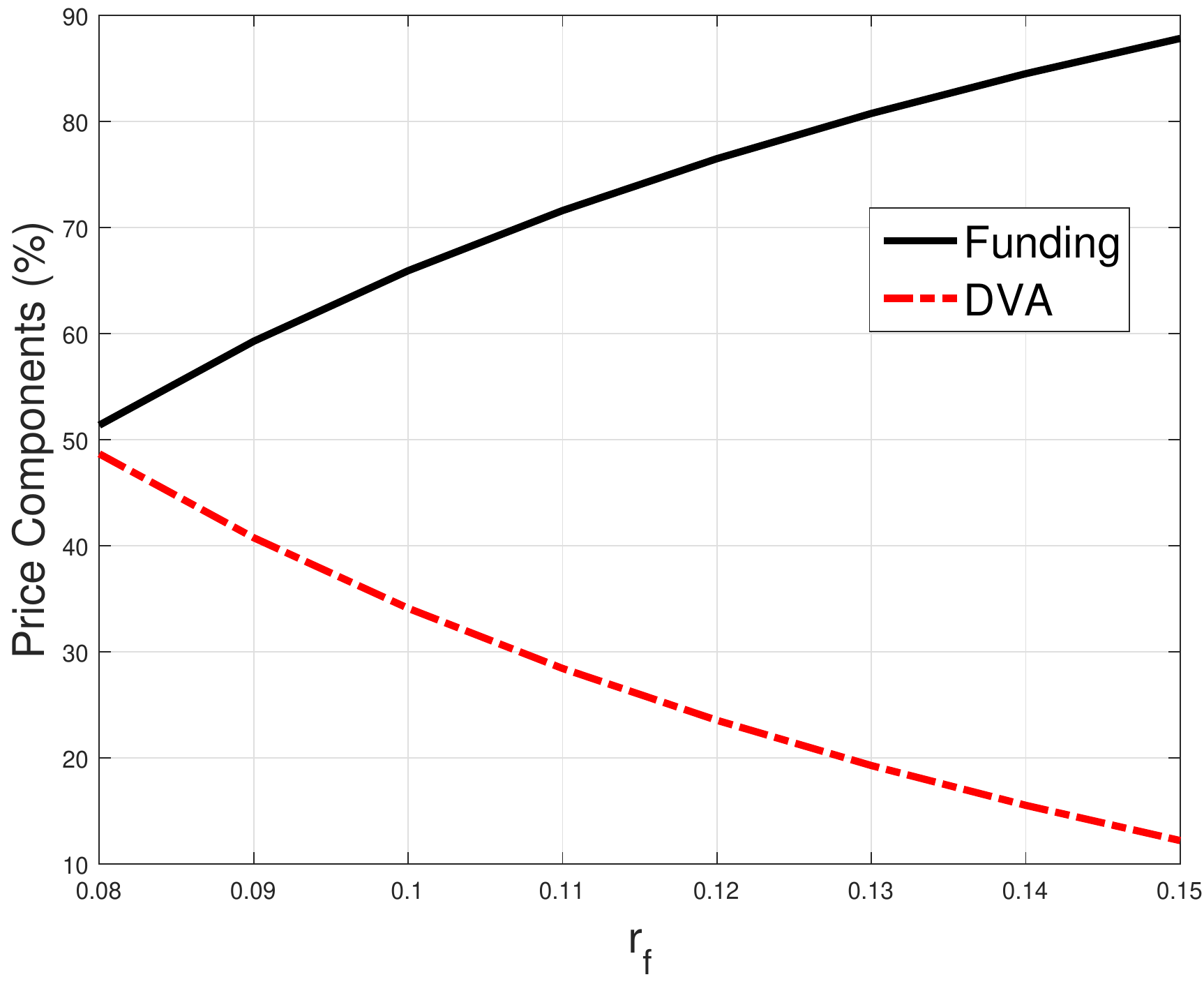}
    \includegraphics[width=0.49\textwidth]{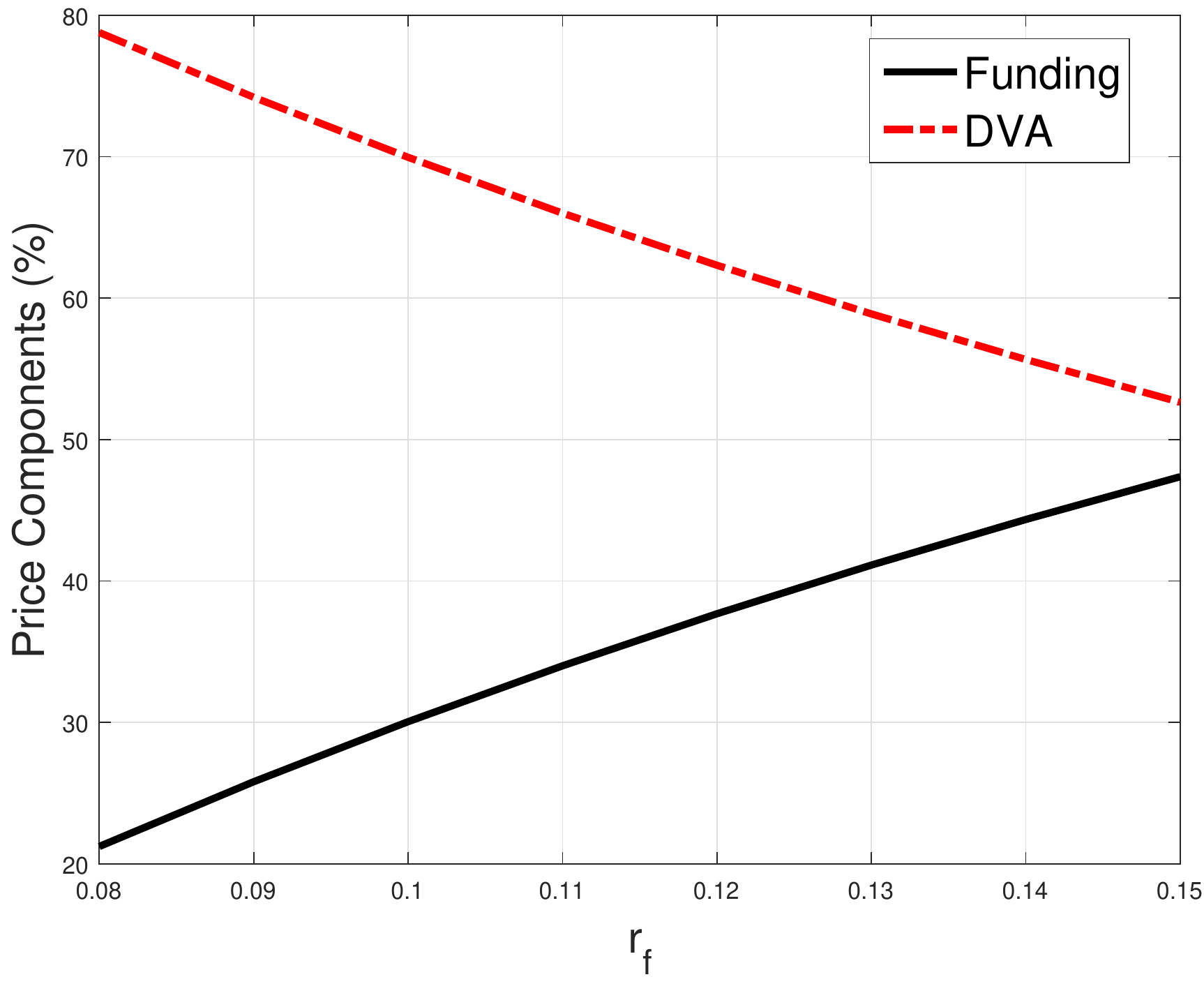}
    \caption{XVA decompositions (in \% of {the market value of the replicated claim}) in the Piterbarg model with defaults. We set $r_r = 0.05$, $r_c = 0.01$, $\sigma = 0.2$, $\alpha=0.25$, $L_C = 0.5$, $L_I = 0.5$.
    The claim is an at-the-money European call option with maturity $T=1$. Left: {$\mu_I = 0.2$, $\mu_C = 0.25$}. Right: {$\mu_I = 0.55$, $\mu_C = 0.55$}.}
  \label{fig:pricedec1}
  \end{figure}
When $\alpha$ is high, the $\XVA$ is positive. The position in the trader's bond is higher than in the counterparty bond given that a residual $\DVA$ component (after collateral mitigation) needs to be replicated, while the $\CVA$ component is zero because the trader faces zero exposure to his counterparty. As $\alpha$ increases, a smaller residual $\DVA$ component needs to be replicated given that the position is more collateralized. As a consequence, the trader reduces the position in his own bonds.

A direct comparison of Figures \ref{fig:defXVA} and \ref{fig:defXVArisk} suggests that for moderately low levels of collateralization, the trader purchases a higher number of his own bonds under the risky scenario, and partly finances this position using the proceeds coming from the short sale of counterparty bonds.

\begin{remark}
The formulas for the $\XVA$ given in Eq.~\eqref{eq:FVAPiterbargdef} and its specialization to a non-negative payoff~\eqref{eq:reduced-sol4} can also be derived using representation results for linear BSDEs with jumps and direct martingale arguments. Such an approach is used in \cite[Section 5.2]{BCS1}.
\end{remark}

  \begin{figure}[ht!]
    \centering
      \includegraphics[width=6.6cm]{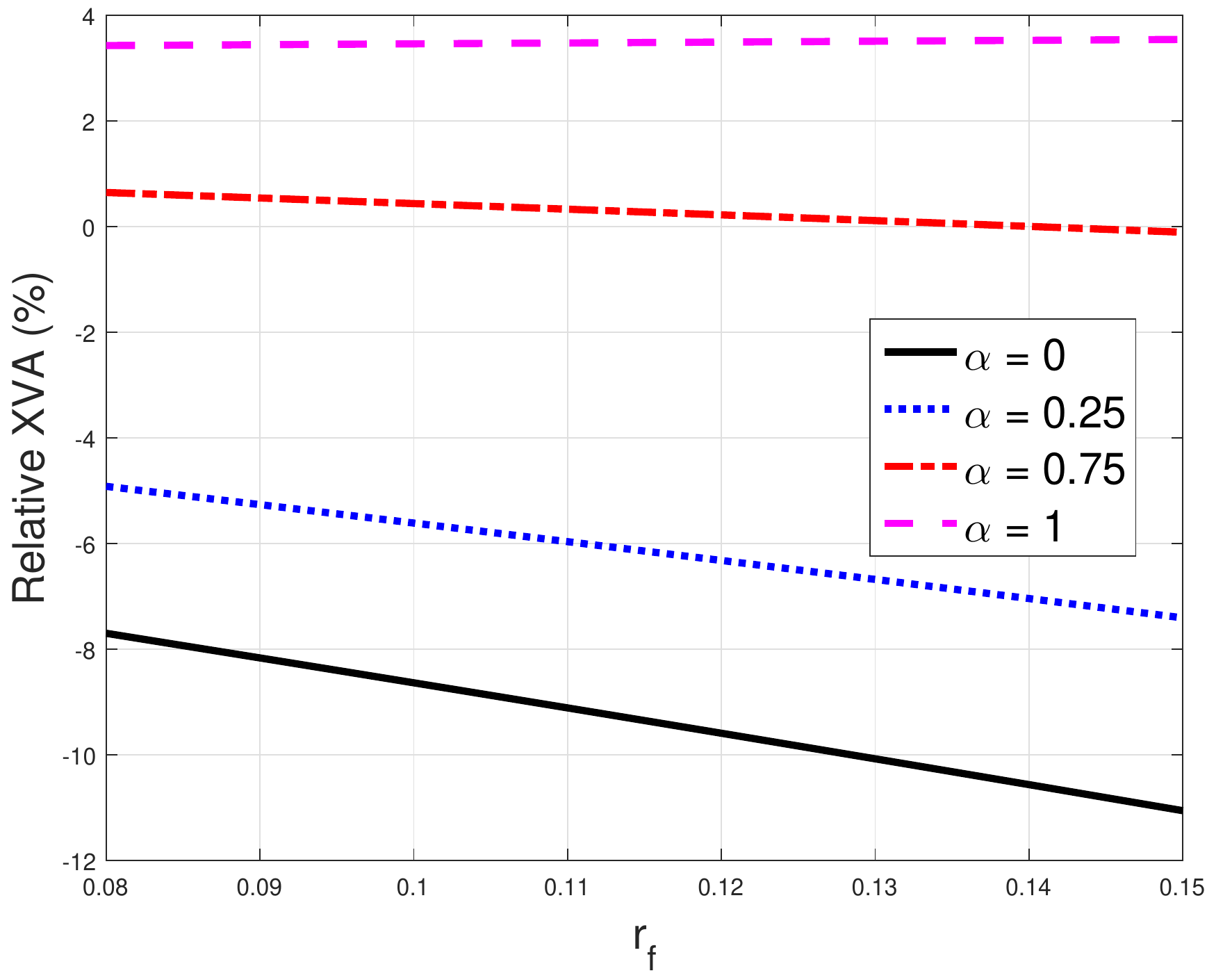}
      \includegraphics[width=6.6cm]{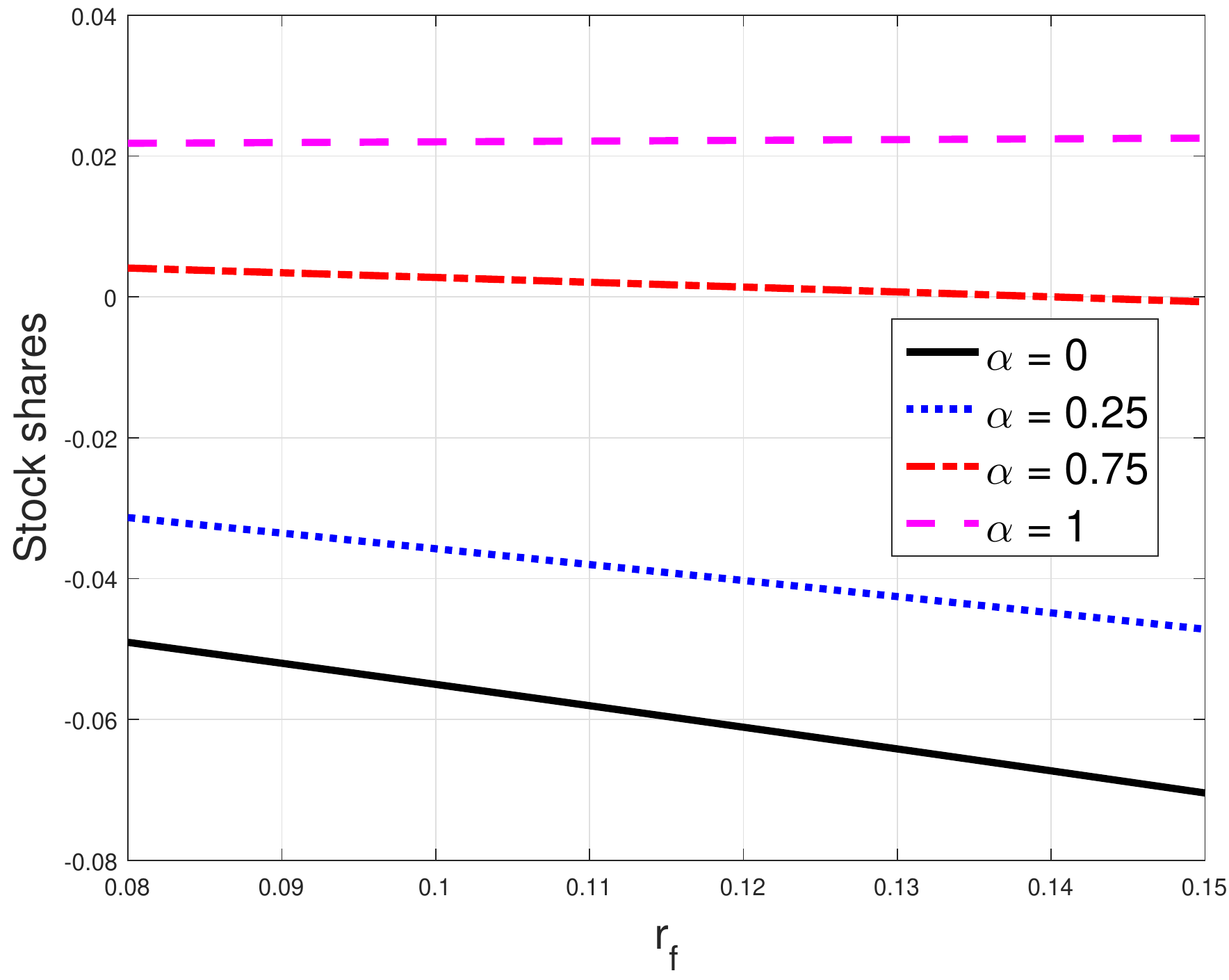}
      \includegraphics[width=6.6cm]{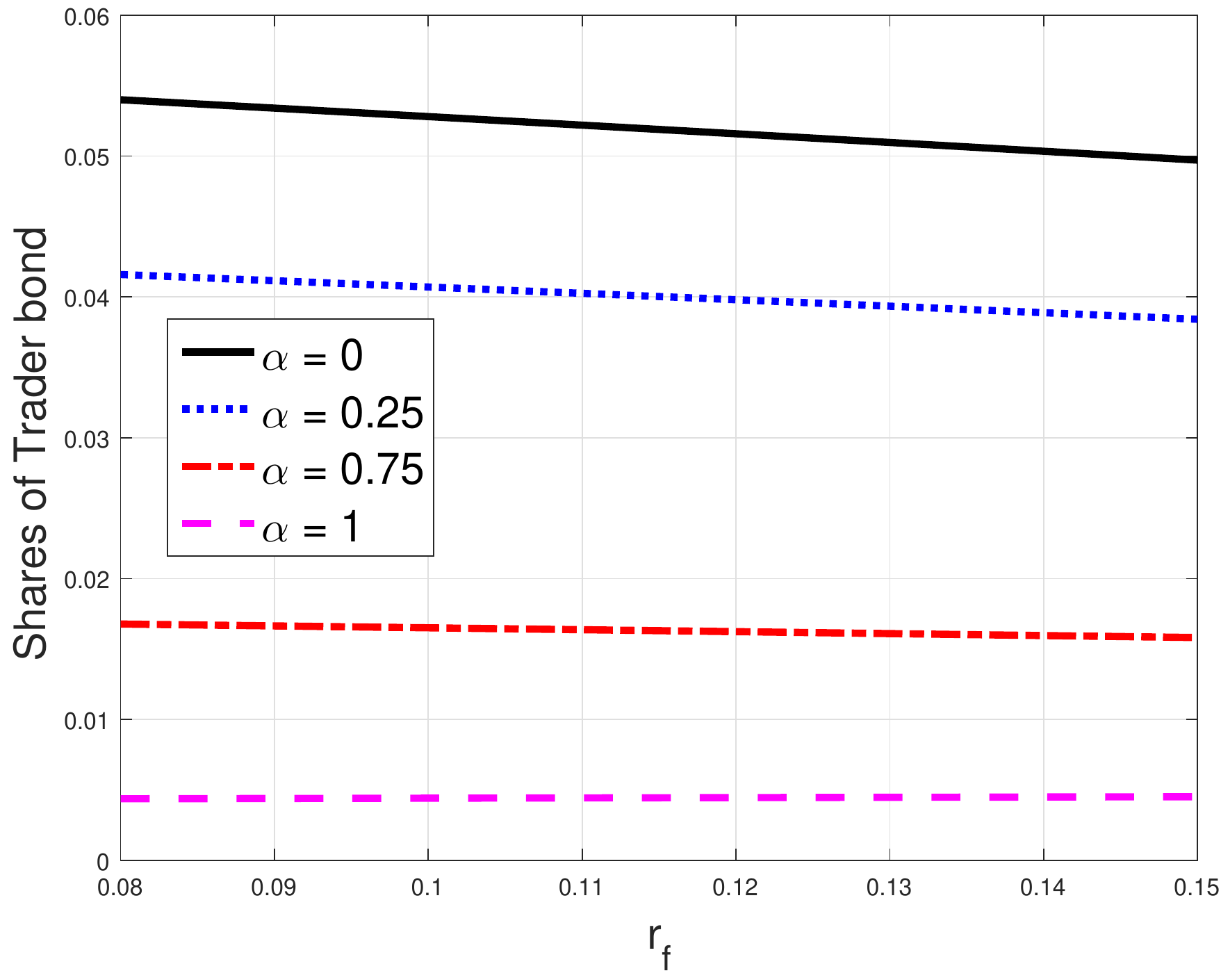}
      \includegraphics[width=6.6cm]{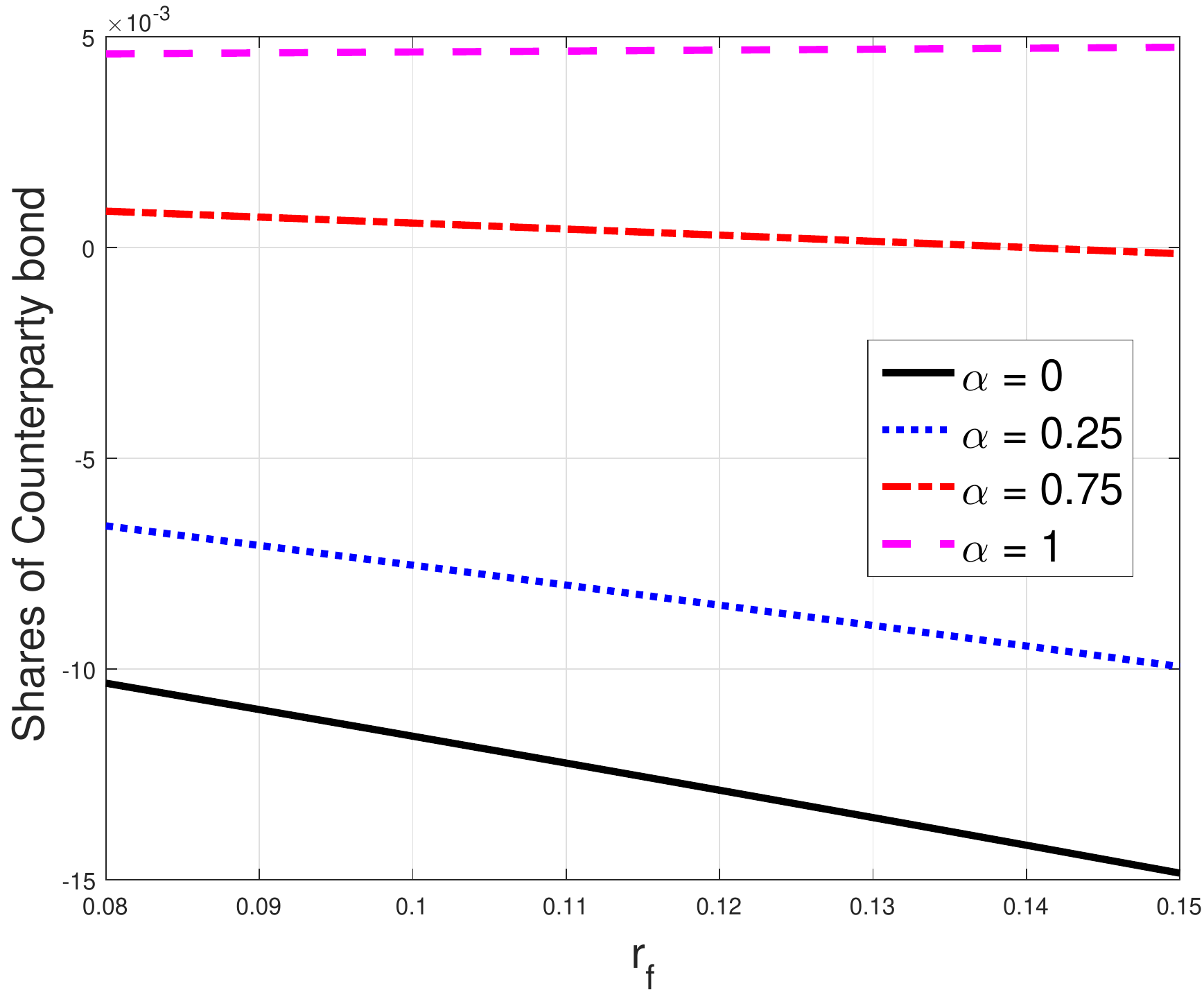}
    \caption{Top left: $\XVA$ as a function of $r_f$ for different $\alpha$. Top right: Number of stock shares in the replication strategy. Bottom left: Number of trader bond shares in the replication strategy. Bottom right: Number of counterparty bond shares in the replication strategy. We set $r_r = 0.05$, $r_c = 0.01$, $\sigma = 0.2$, $L_C = 0.5$, $L_I = 0.5$, $\mu_I = 0.16$ and $\mu_C = 0.21$. The claim is an at-the-money European call option with maturity $T=1$.}
  \label{fig:defXVA}
  \end{figure}

  \begin{figure}[ht!]
    \centering
      \includegraphics[width=6.6cm]{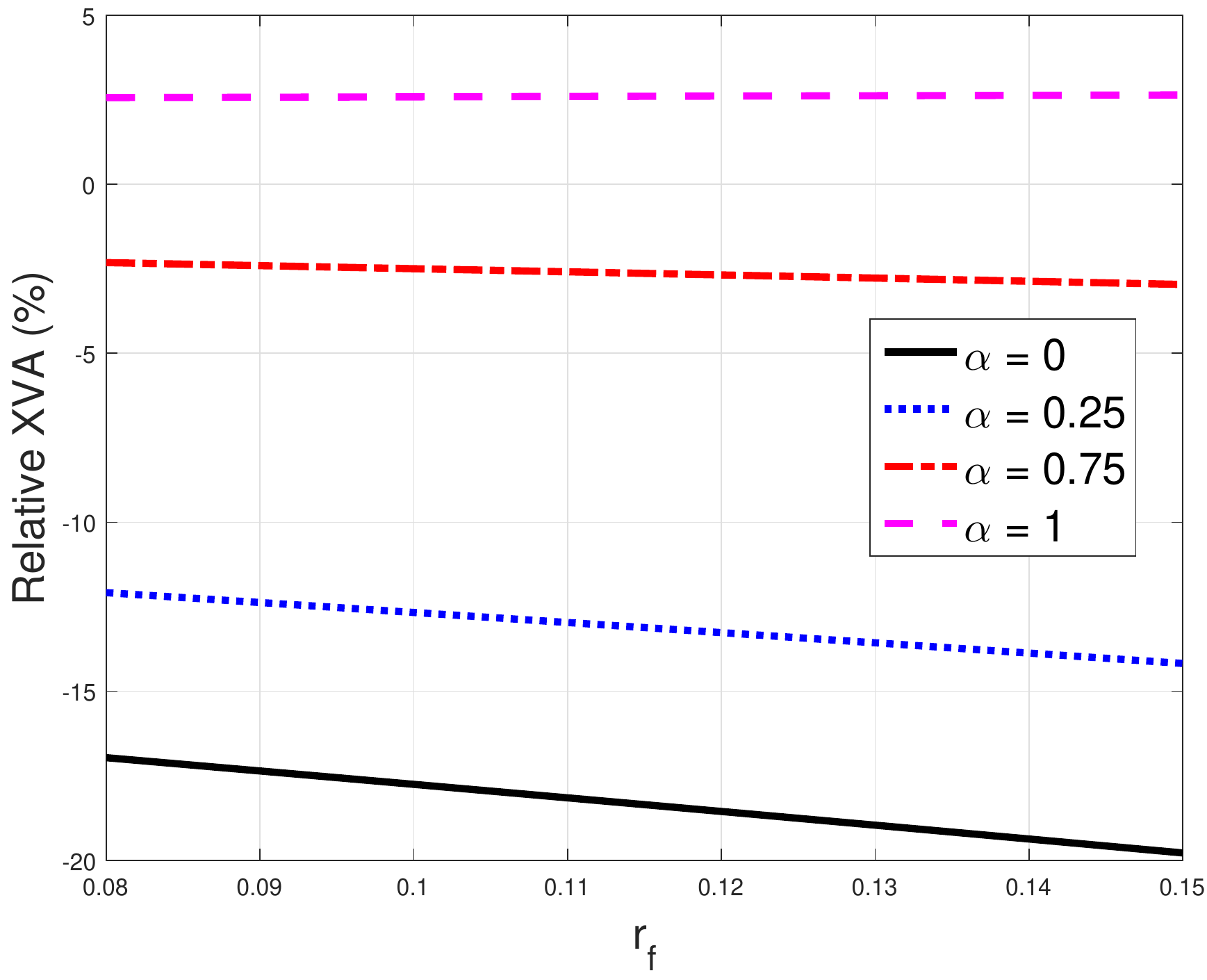}
      \includegraphics[width=6.6cm]{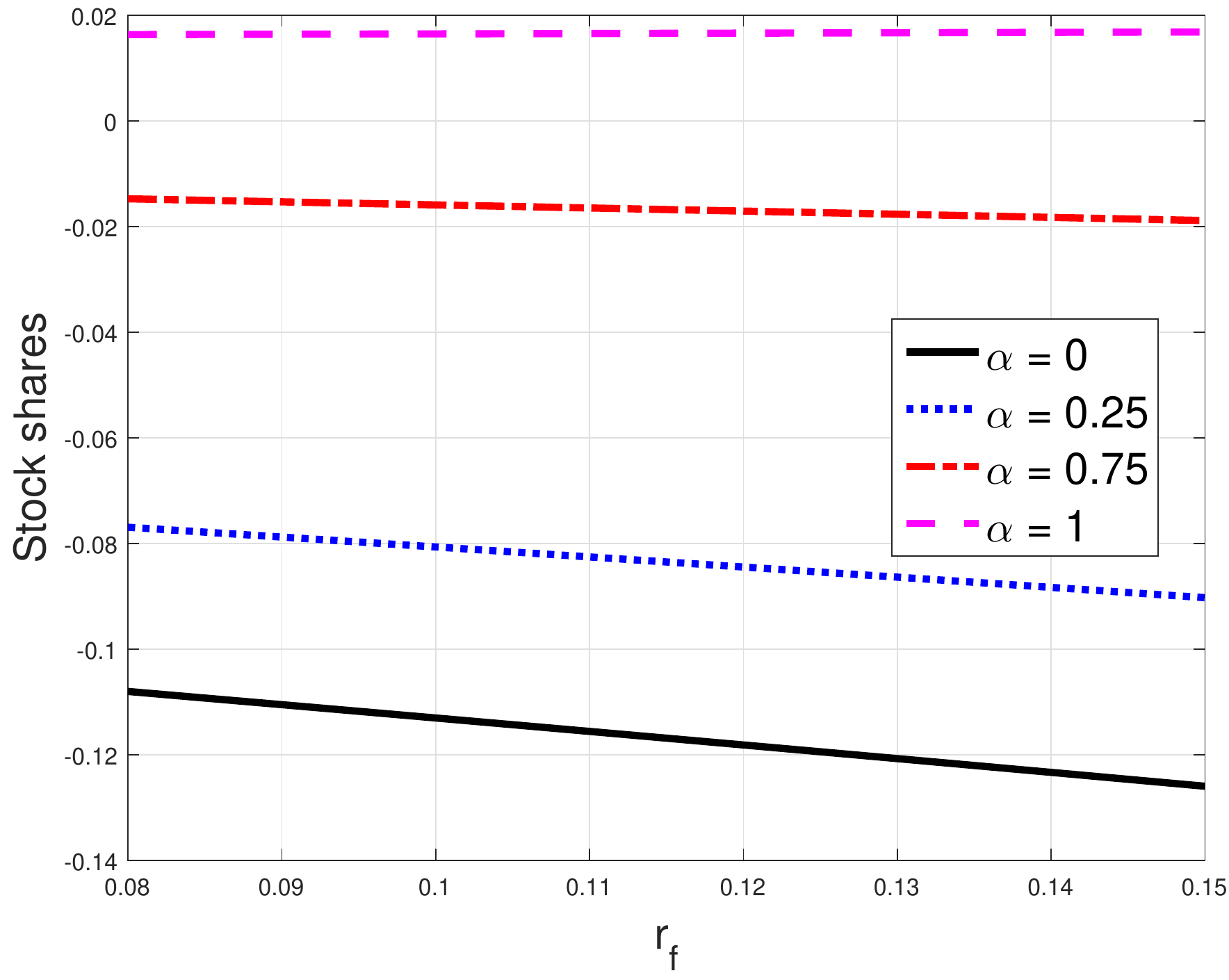}
      \includegraphics[width=6.6cm]{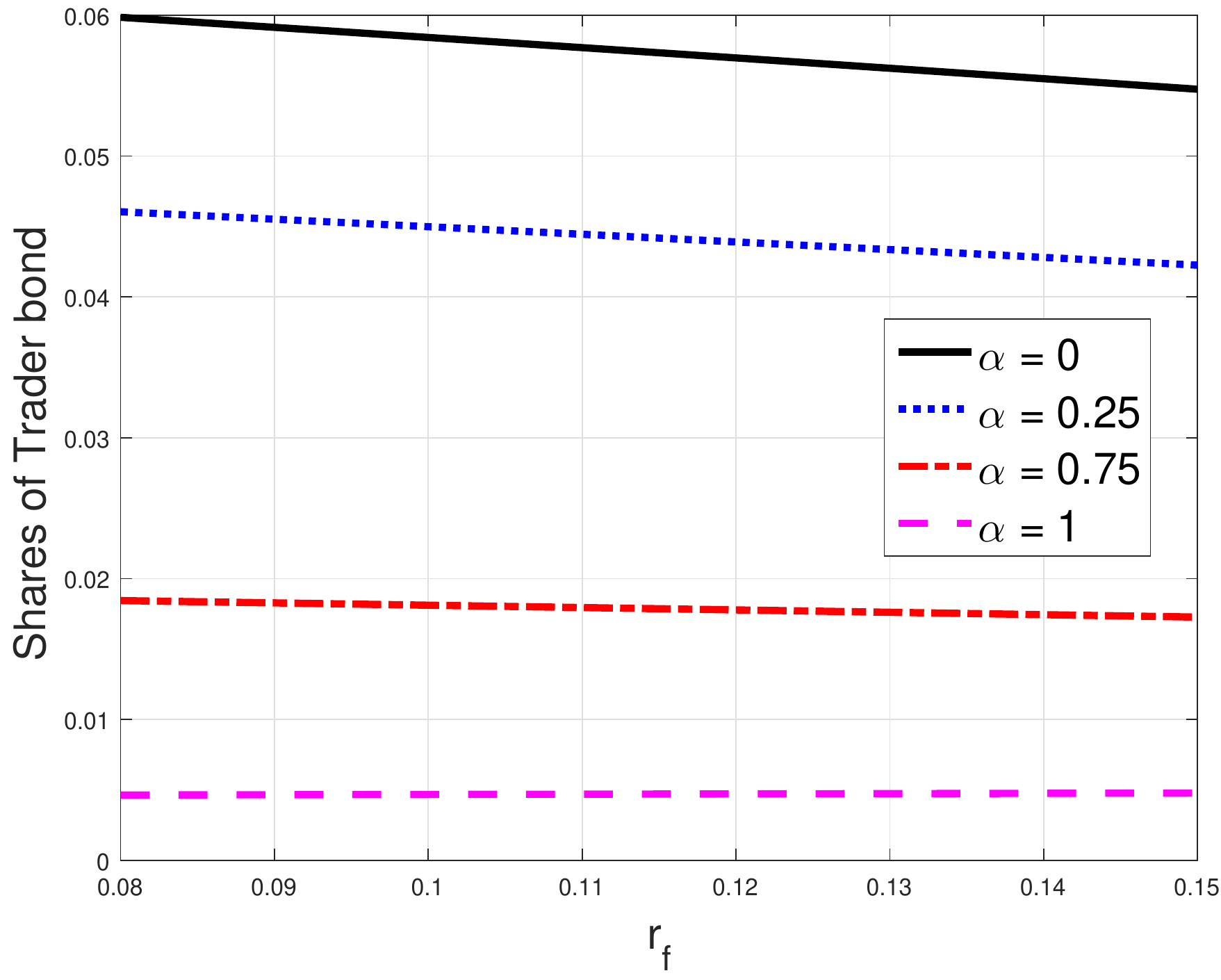}
      \includegraphics[width=6.6cm]{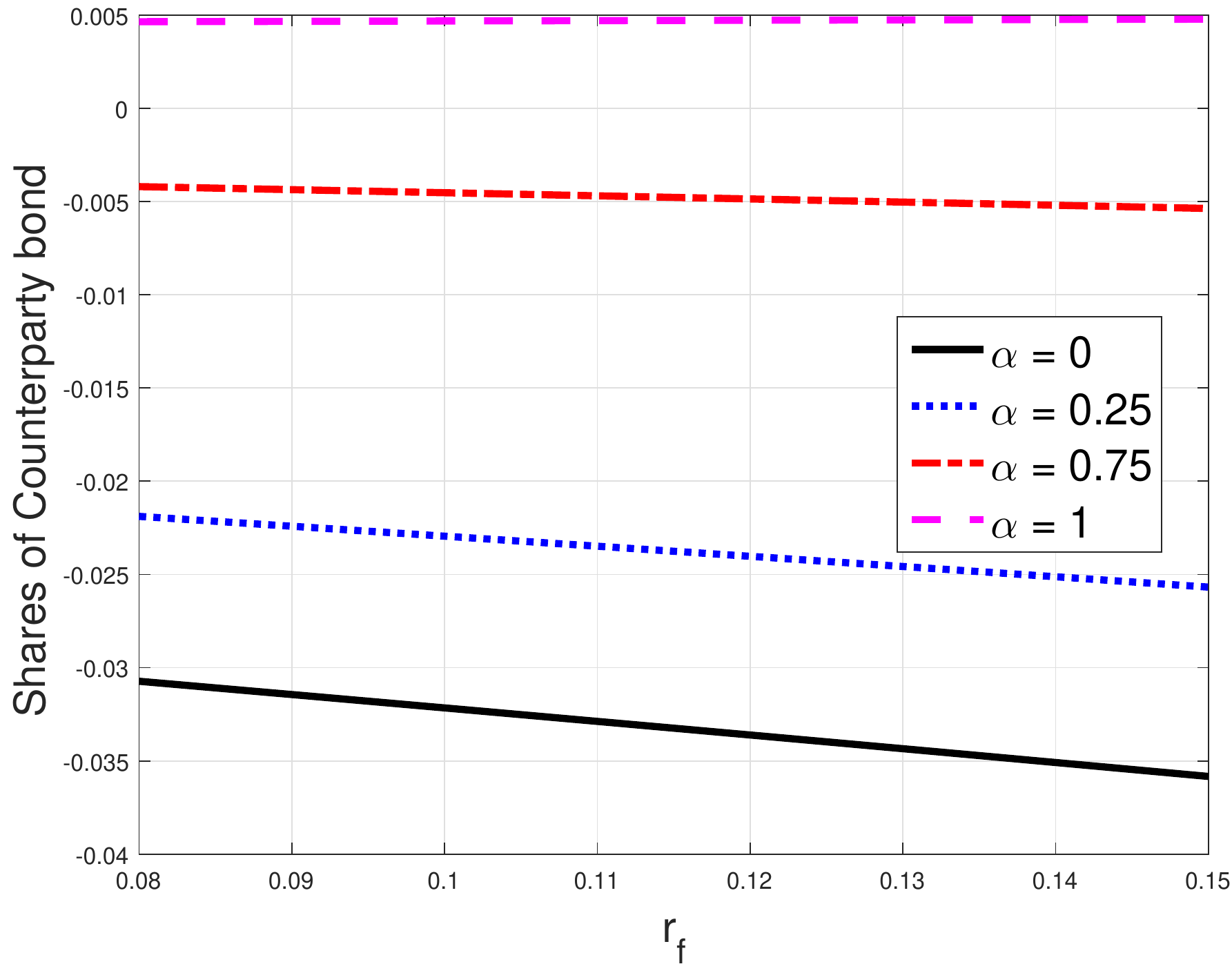}
    \caption{Top left: $\XVA$ as a function of $r_f$ for different $\alpha$. Top right: Number of stock shares in the replication strategy. Bottom left: Number of trader bond shares in the replication strategy. Bottom right: Number of counterparty bond shares in the replication strategy. We set $r_r = 0.05$, $r_c = 0.01$, $\sigma = 0.2$, $L_C = 0.5$, $L_I = 0.5$, $\mu_I = 0.51$ and $\mu_C = 0.51$. The claim is an at-the-money European call option with maturity $T=1$.}
  \label{fig:defXVArisk}
  \end{figure}

\section{Numerical analysis} \label{sec:numanalysis}

We conduct a comparative statics analysis to analyze the dependence of $\XVA$ and portfolio replicating strategies on funding rates, bond returns, and collateralization levels in the general nonlinear setting of section~\ref{sec:BSDEform}. We consider the relative XVA, i.e., express the adjustment as a percentage of the price $\hat{V}_t$ of the claim, given by $\XVA_t^{\buysell} / \hat{V}_t$. The claim is chosen to be a European-style call option on the stock security, i.e., $\Phi(x) = (x-K)^+$. We consider one at-the-money option with $S_0=K=1$ maturing at $T=1$. In order to focus on the impact of funding costs (which in practice is the most relevant) and separate it from additional contributions to the $\XVA$ coming from asymmetries in collateral and repo rates, we set $\rcp = \rcm = 0.01$ and $\rrp = \rrm = 0.05$, unless specified otherwise.

We use the following benchmark parameters: $\sigma = 0.2$, $\rfp = 0.05$, $\rfm = 0.08$, $r_D = 0.01$,  $\mu_I = 0.21$, $\mu_C = 0.16$, $L_I = L_C = 0.5$, and $\alpha = 0.9$. We compute the numerical solution of the PDE using a finite difference Crank-Nicholson scheme. The main findings of our analysis are discussed in the sequel.

\paragraph{Higher funding rate increases the width of the no-arbitrage band.}

As the derivative contract specifies both the price of the option and the level of collateralization of the deal, the no-arbitrage region appears as a (two-dimensional) band in $\XVA$ and $\alpha$ rather than as a (one-dimensional) interval in $\XVA$ only. Figure \ref{fig:alpharfm} displays the no-arbitrage band, whose width is increasing in the funding rate $\rfm$. As $\alpha$ gets higher, the band noticeably shrinks reaching its minimum around $\alpha = 80\%$ before widening again. Notice that buyer{'}s and seller{'}s $\XVA$ do not have a symmetric behavior. This can be better understood by analyzing the dependence of the band on the collateralization level $\alpha$ in Figure \ref{fig:alpharfm}. If $\alpha$ is not too high ($\alpha < 0.5$), the widening of the no-arbitrage band with respect to the funding rate $\rfm$ is due to decreasing buyer{'}s XVA. On the other hand, if $\alpha$ is high the buyer{'}s $\XVA$ is insensitive to changes in $\rfm$, whereas the seller{'}s $\XVA$ increases with $\rfm$, contributing to widen the no-arbitrage band. This is further supported by the numerical values reported in Table \ref{tab:Tablealrfm}. When $\alpha < 0.5$, the position in the funding account for the seller{'}s $\XVA$ is long and the same regardless of the funding rate $r_f^-$. On the other hand, the size of the long position for the buyer{'}s $\XVA$ increases in $\rfm$. In presence of full collateralization, i.e., $\alpha=1$, the situation reverses. The position in the funding account for the buyer{'}s $\XVA$ is short and stays constant with respect to $\rfm$. Vice versa, for the seller{'}s $\XVA$ the size of the short position increases in $\rfm$.

If $\alpha$ is high, the trader will have to post more collateral and consequently reduce the cash available to finance his replicating strategy. He will then have to borrow more from the funding desk, resulting in higher funding costs. This drives up both the seller{'}s $\XVA$ and the number of shares of stocks and bonds needed for the replication strategy.

  \begin{figure}[ht!]
    \centering
      \includegraphics[width=6.6cm]{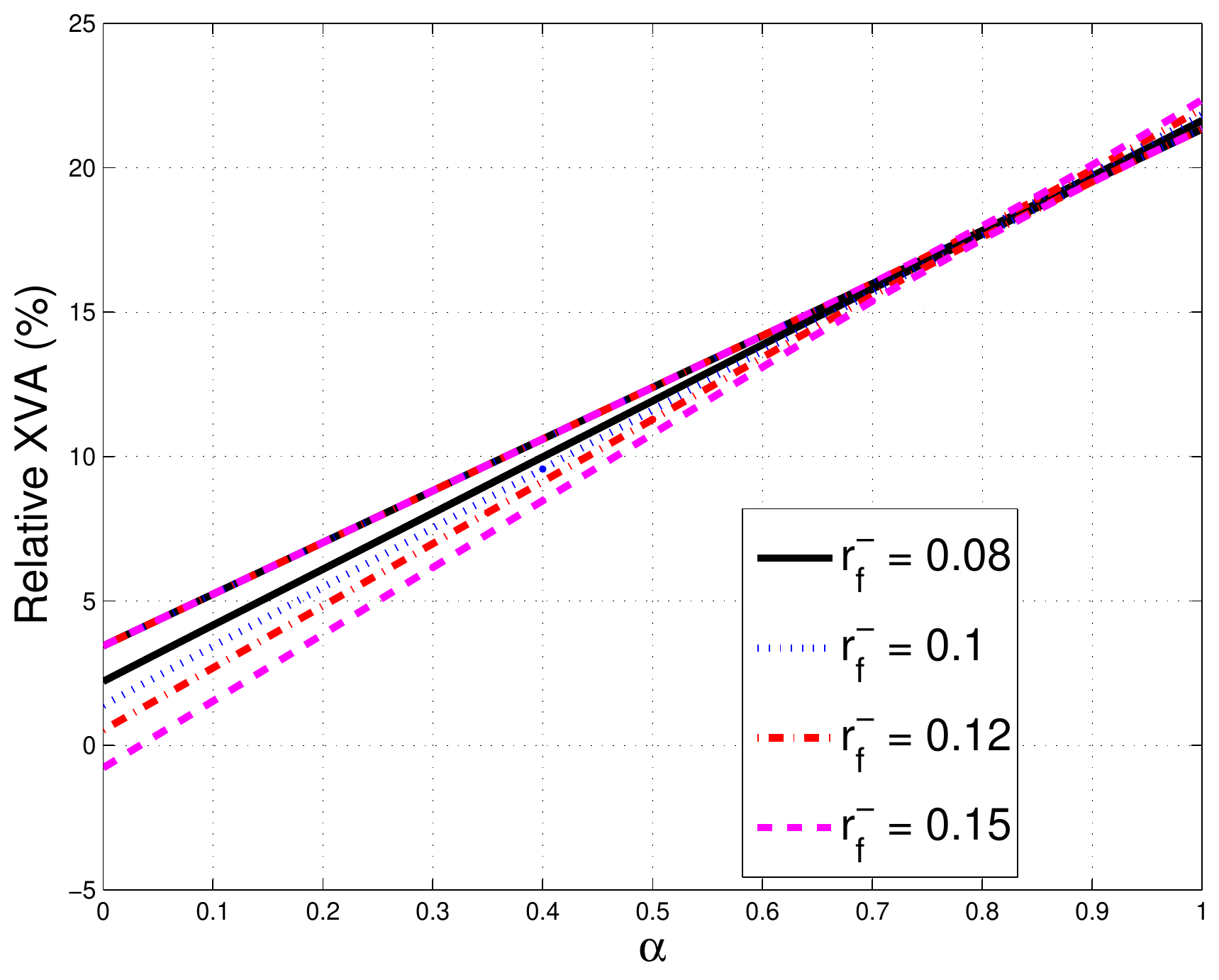}
      \includegraphics[width=6.6cm]{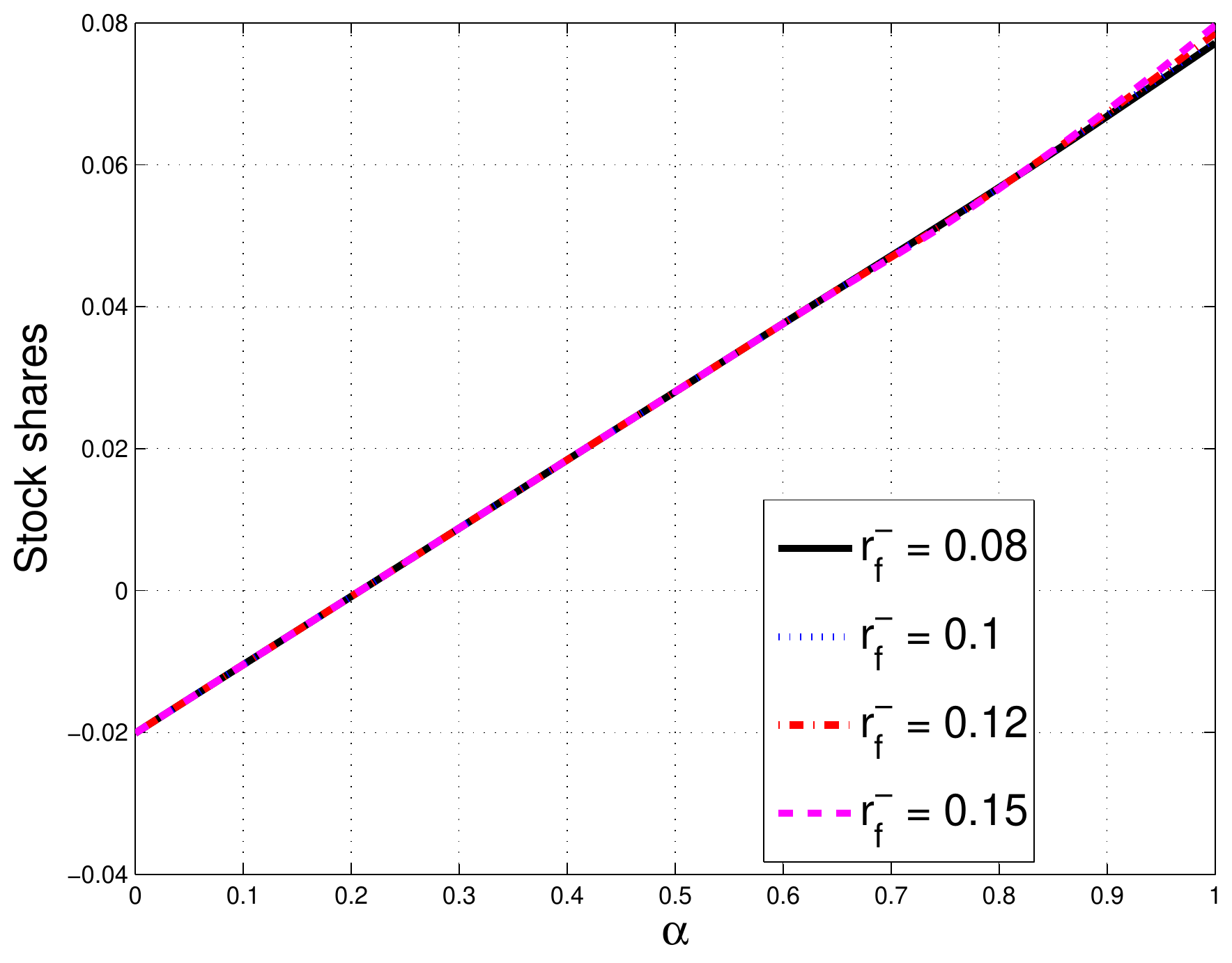}
      \includegraphics[width=6.6cm]{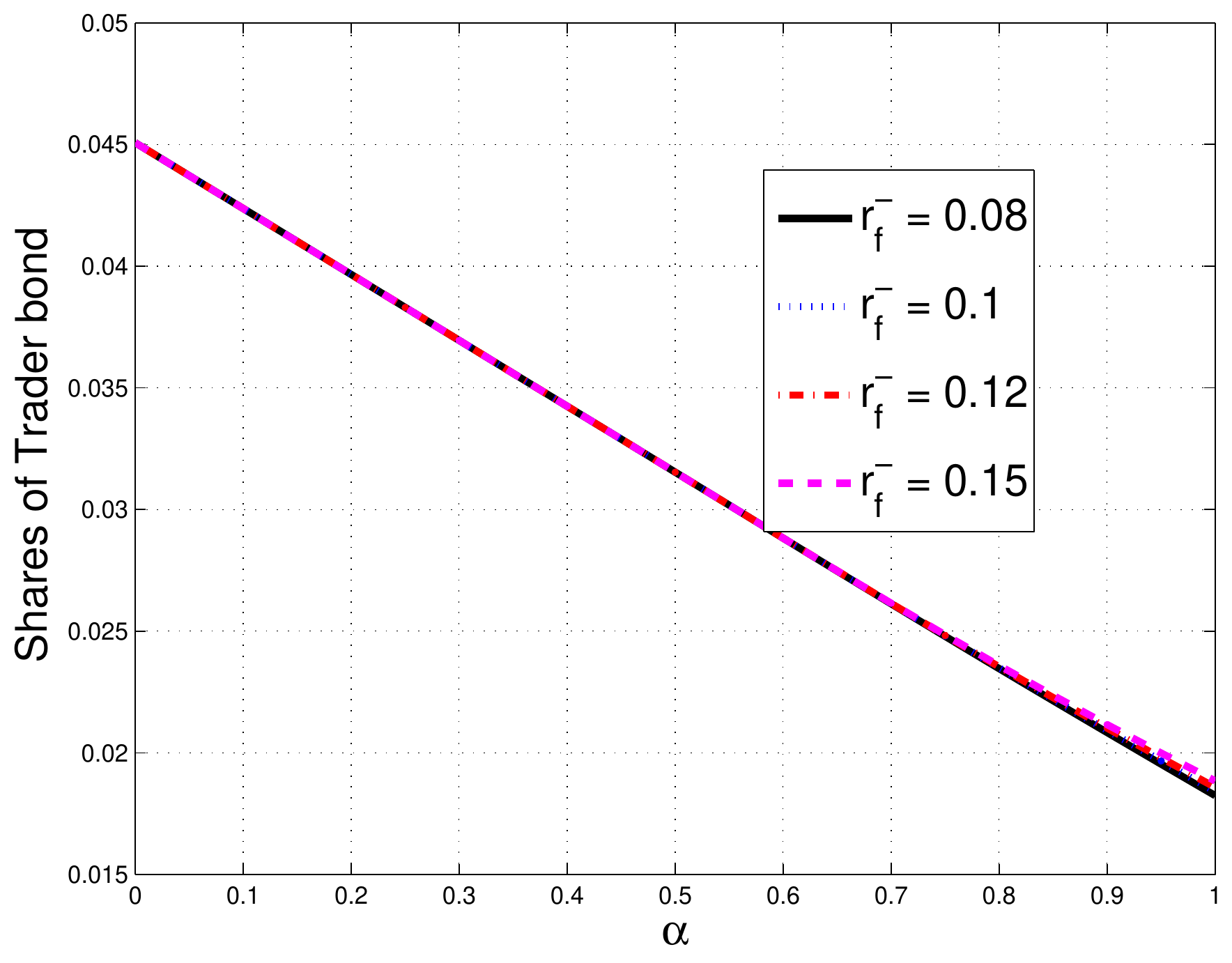}
      \includegraphics[width=6.6cm]{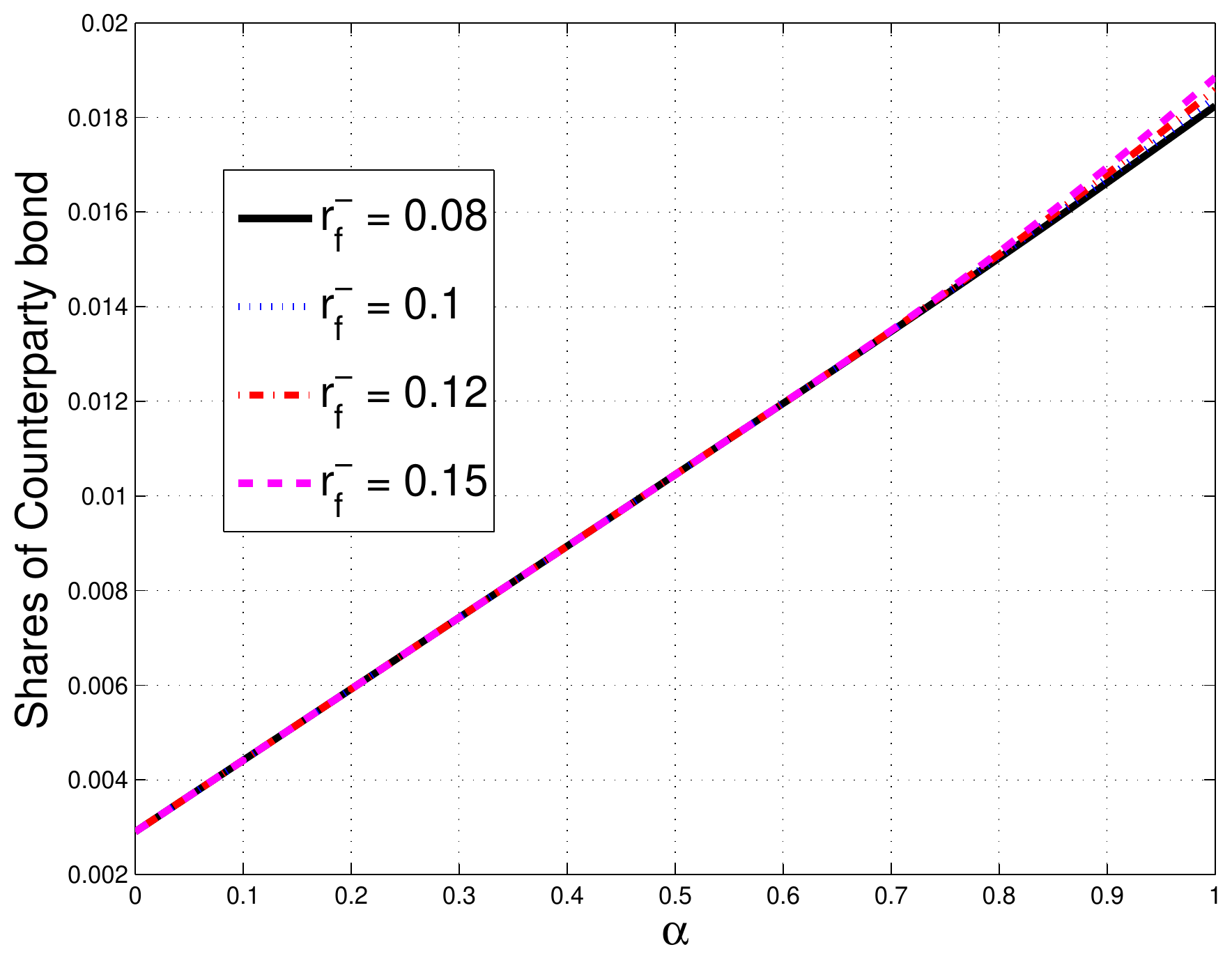}
    \caption{Top left: Buyer{'}s and seller{'}s $\XVA$ as a function of $\alpha$ for different values of $\rfm$. The seller's lies above the buyer{'}s $\XVA$ and the same line style is used for
both. Top right: Number of stock shares in the replication strategy. Bottom left: Number of trader{'}s bond shares in the replication strategy. Bottom right: Number of counterparty{'}s
bond shares in the replication strategy. The strategies refer to the portfolio replicating the seller{'}s XVA.}
  \label{fig:alpharfm}
  \end{figure}

  \begin{table}[hpt]
    \centering
      \begin{tabular}{|c|c|c|c|}
	\hline
	$\alpha$ & $\rfm$ & {\text Seller{'}s XVA: funding account} (\$) & {\text Buyer{'}s XVA: funding account} (\$)\\
	\hline
	\hline
	0 & 0.08 & 0.0039 & 0.0403 \\
	\hline
	0 & 0.15 & 0.0039 &  0.0428 \\
	\hline
	0.25 & 0.08 & 0.0249 & 0.0257 \\
	\hline
	0.25 & 0.15 & 0.0249 & 0.0274 \\
	\hline
	0.75 & 0.08 & -0.0037 & -0.0036 \\
	\hline
	0.75 & 0.15 & -0.0038 & -0.0033 \\
	\hline
	1 & 0.08 & -0.0182 & -0.018 \\
	\hline
	1 & 0.15 & -0.0193 & -0.018 \\
	\hline
      \end{tabular}
    \caption{The columns give the dollar position in the funding account corresponding to the replicating strategies of seller{'}s $\XVA$ and buyer{'}s XVA.}
  \label{tab:Tablealrfm}
  \end{table}
Table \ref{tab:Tablerf} also indicates that the funding position corresponding to the seller{'}s  $\XVA$ is negative, but low, hence explaining why the buyer{'}s $\XVA$ is only mildly sensitive
to changes in the funding rate $\rfm$.

\paragraph{The no-arbitrage interval widens as the difference between repo rates increases.}
Figure \ref{fig:rrmrrp} analyzes how the width of the no-arbitrage interval varies, when the difference between borrowing and lending repo rates increases. For a fixed repo lending rate, the seller's $\XVA$ increases in the repo borrowing rate. This is because the hedger needs to replicate a long position in the claim and hence incurs higher funding costs when he
purchases stocks (cash driven repo activity, see also Figure \ref{fig:cashdriven}). In contrast, the buyer's $\XVA$ is not sensitive to the repo borrowing rate. In this case, the hedger needs to replicate a short position in the claim and hence implements a security driven repo activity which only depends on the repo lending rate $\rrp$ (see also Figure \ref{fig:secdriven}). If the repo lending rate gets higher, the hedger receives larger proceeds from the repo market and hence he is willing to purchase the claim at a higher price as he gets more income from his short selling strategy, resulting in an increase of buyer's XVA. This is also reflected in the right panel of Figure \ref{fig:rrmrrp}, suggesting that the trader who wants to hedge his long position shorts a larger number of shares as $\rrp$ gets higher in order to benefit from the higher rate received from the repo market.

\begin{figure}[ht!]
    \centering
      \includegraphics[width=6.6cm]{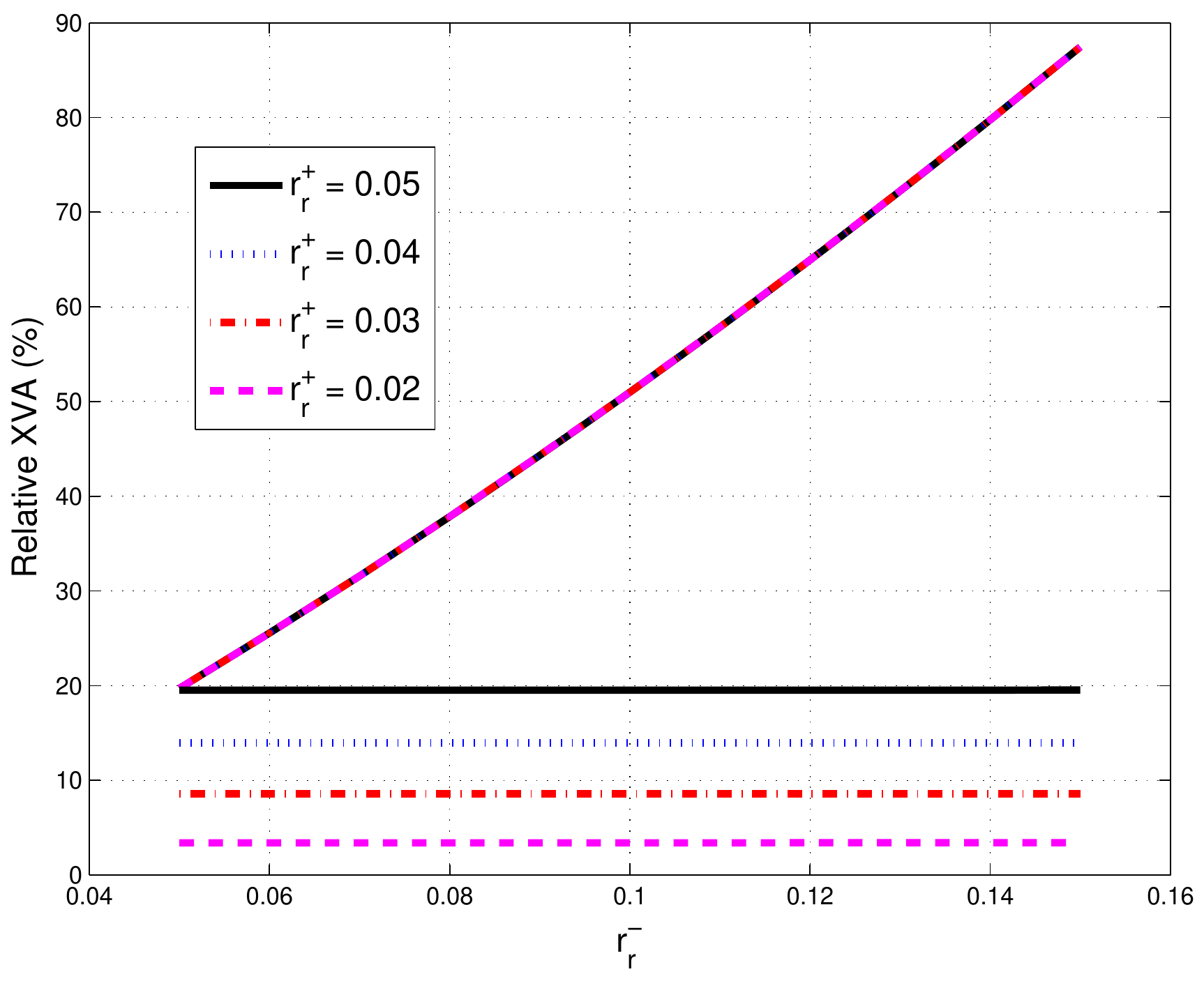}
      \includegraphics[width=6.6cm]{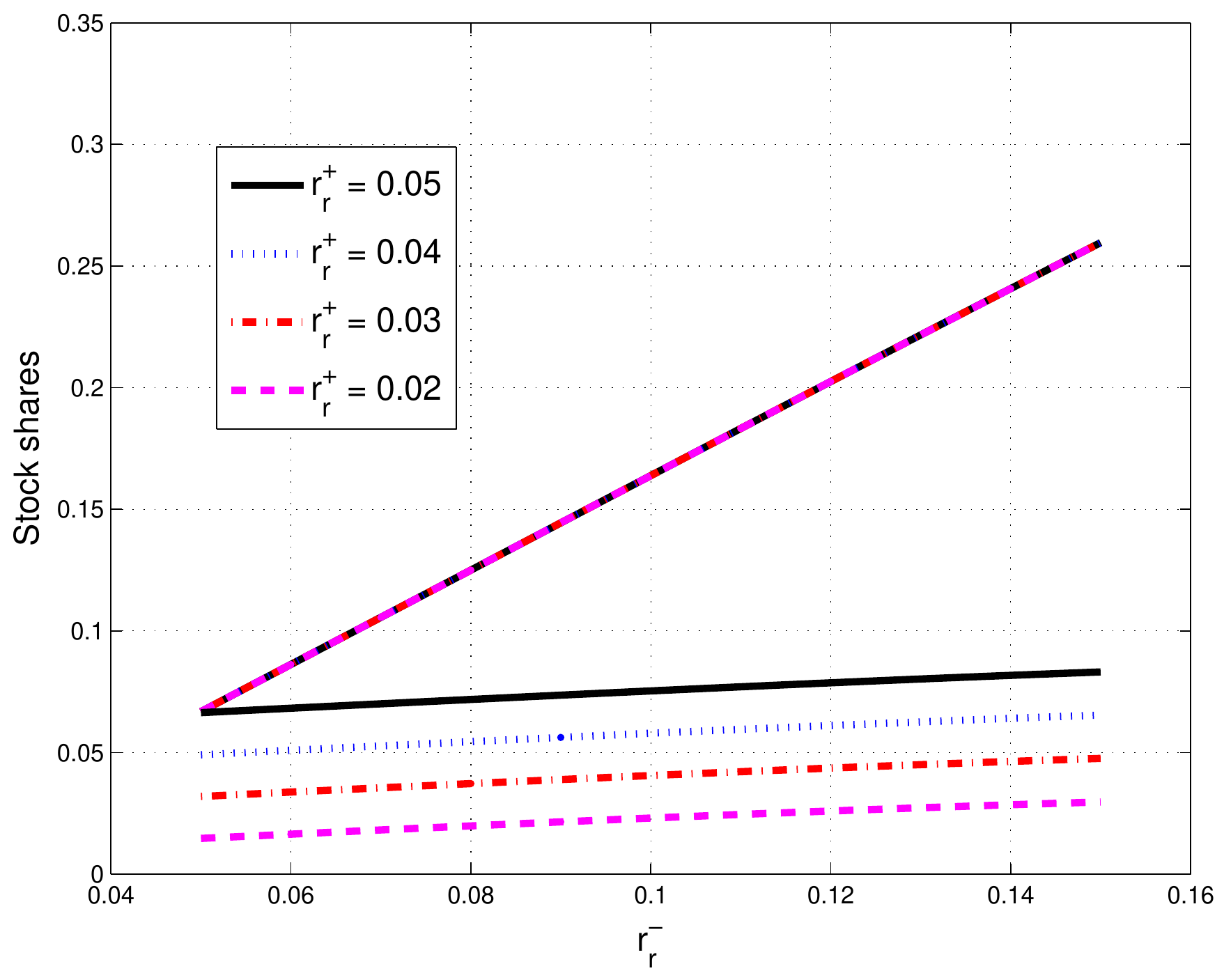}
    \caption{Left: Buyer{'}s and seller{'}s $\XVA$  as a function of $\rrm$ for different values of $r_r^{+}$. Right: Number of stock shares in the strategy replicating the seller{'}s $\XVA$ (top) and the buyer{'}s $\XVA$ (bottom).}
  \label{fig:rrmrrp}
  \end{figure}

\paragraph{Higher collateralization increases portfolio holdings.}

  \begin{figure}[ht!]
    \centering
      \includegraphics[width=6.6cm]{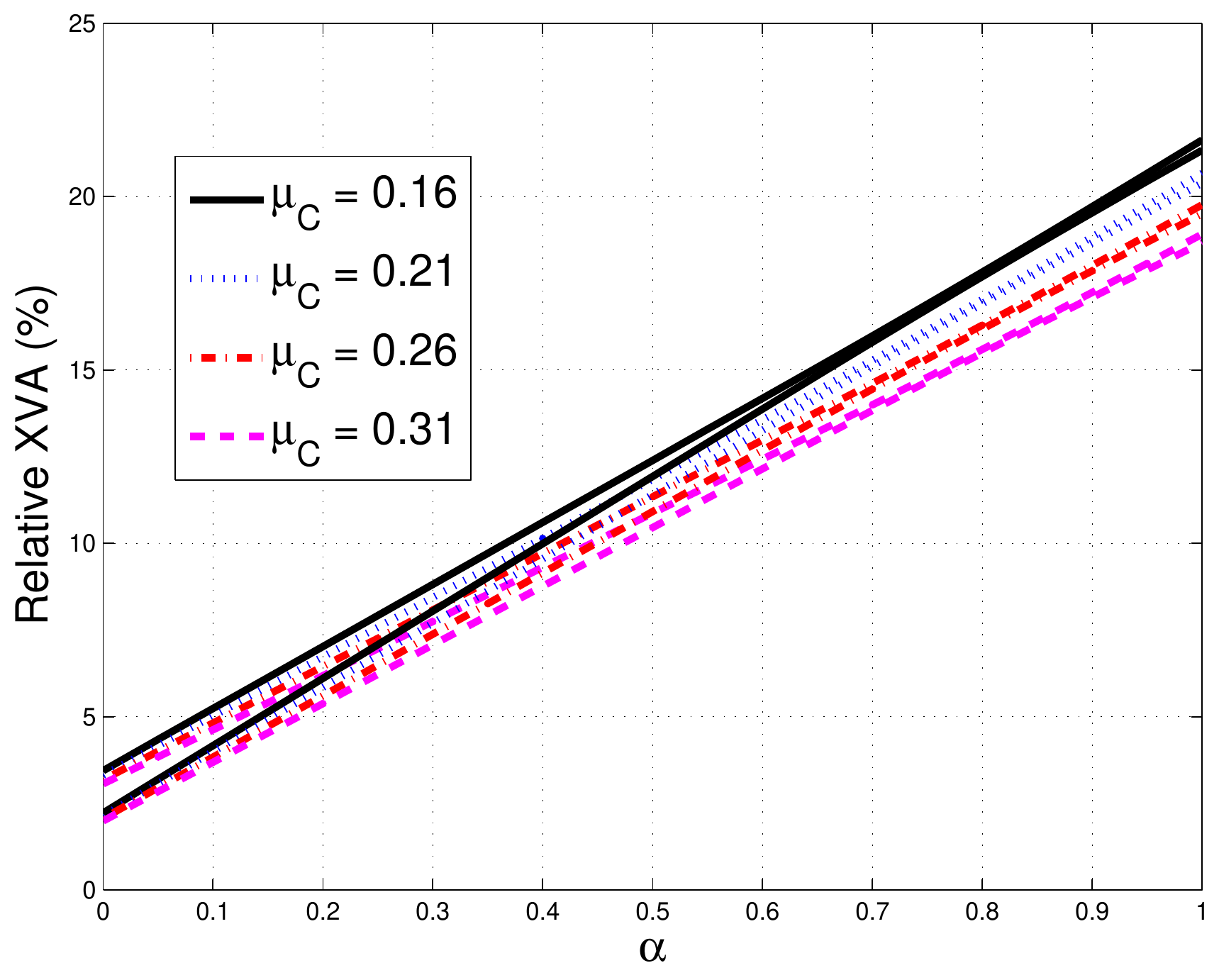}
      \includegraphics[width=6.6cm]{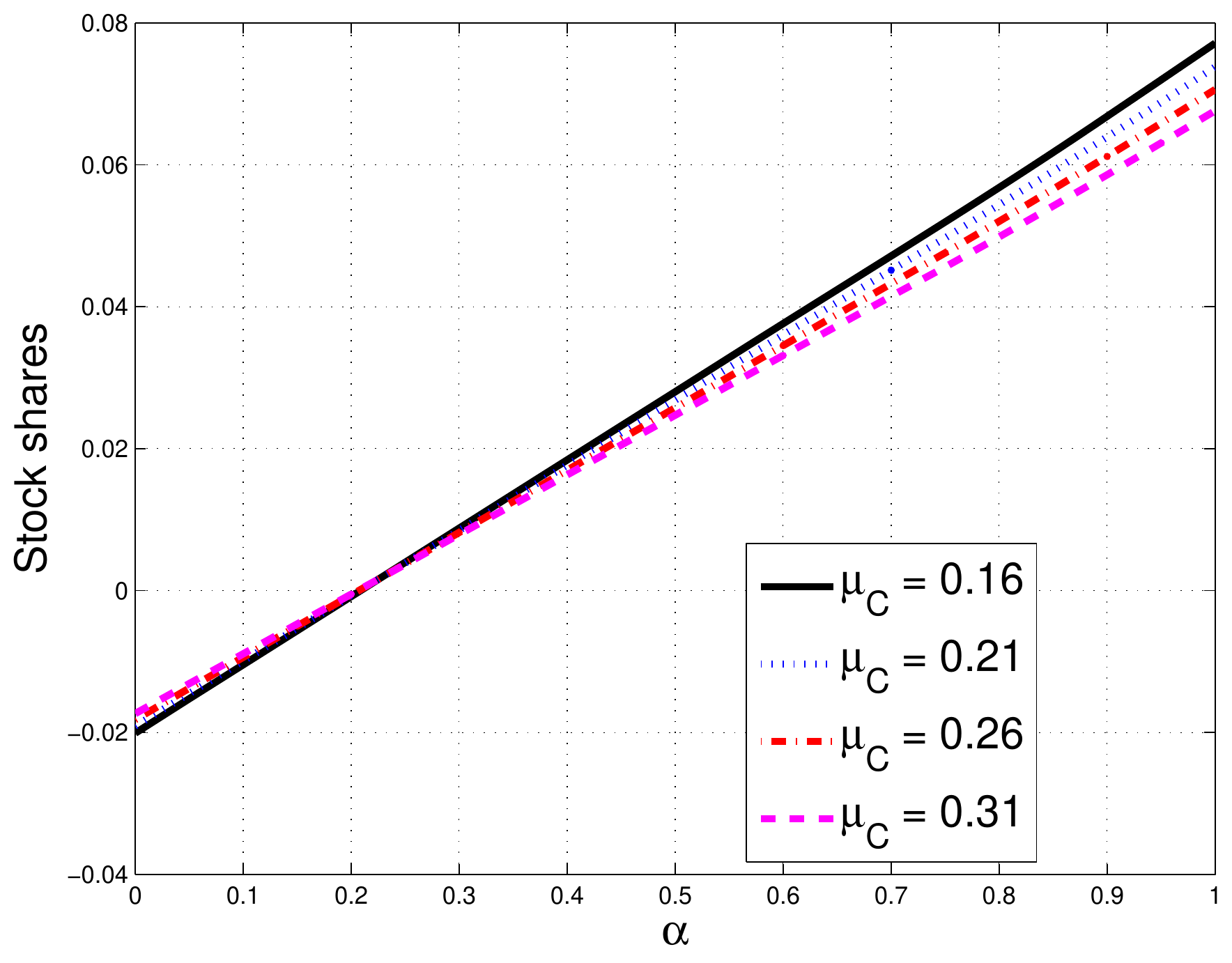}
      \includegraphics[width=6.6cm]{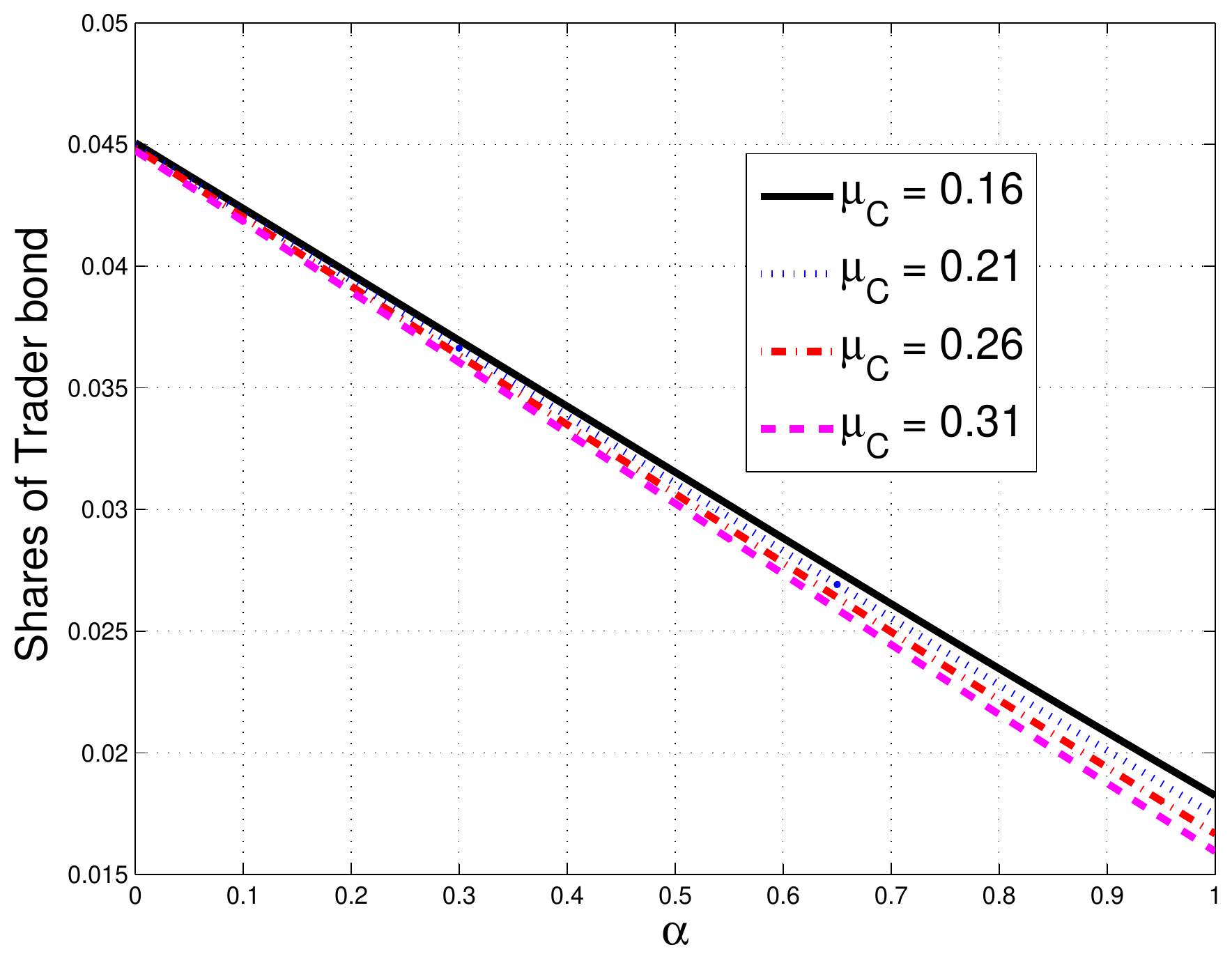}
      \includegraphics[width=6.6cm]{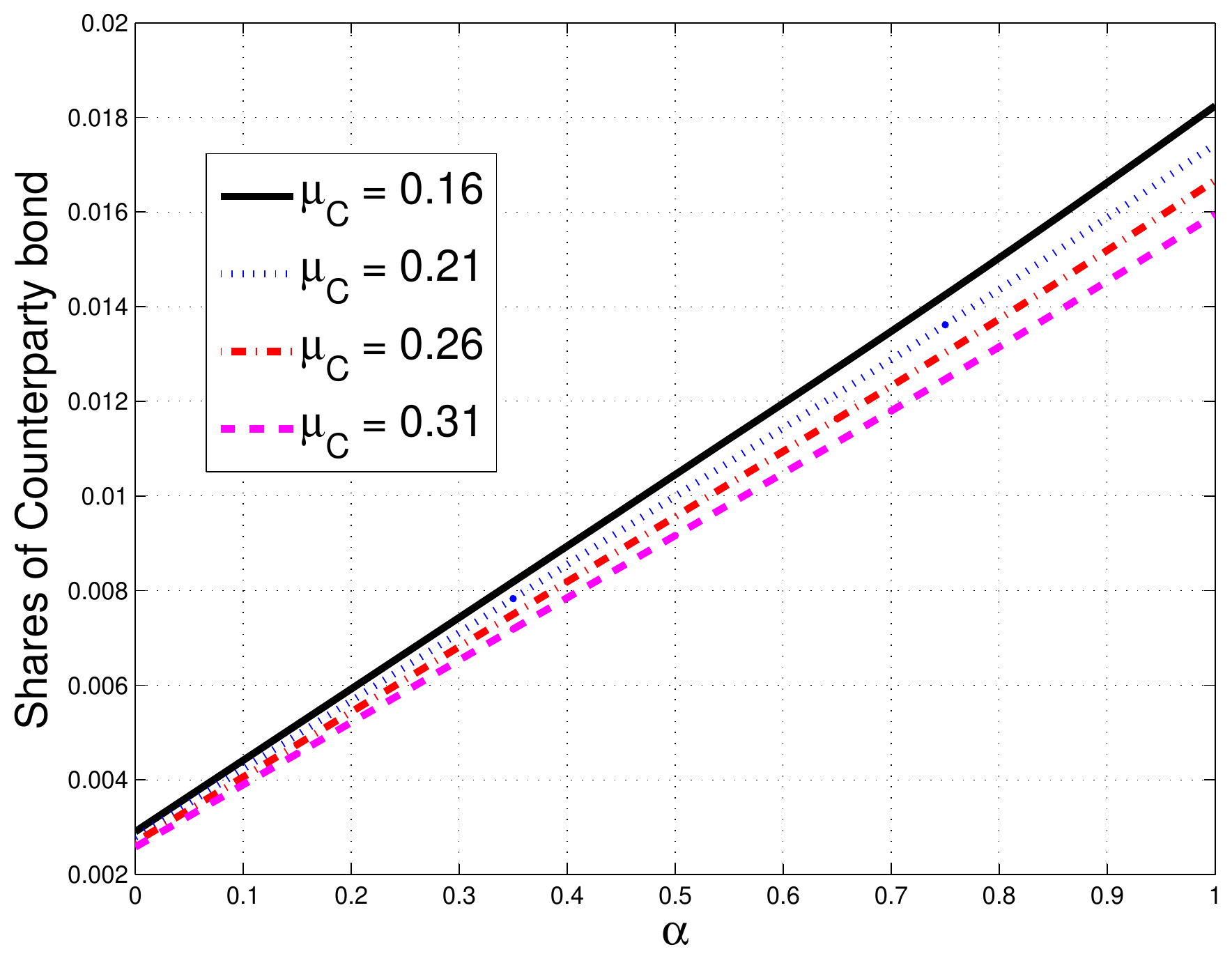}
    \caption{Top left: Buyer{'}s and seller{'}s $\XVA$  as a function of $\alpha$ for different values of $\mu_C$. Top right: Number of stock shares in the replication strategy. Bottom left: Number of trader{'}s bond shares in the replication strategy. Bottom right: Number of counterparty{'}s bond shares in the replication strategy. The strategies refer
        to the portfolio replicating the seller{'}s XVA.}
  \label{fig:alphahc}
  \end{figure}

As the collateralization level $\alpha$ increases, the seller{'}s $\XVA$ increases. This happens because the funding costs incurred from replicating the collateralized position become higher. Because the trader needs to construct a portfolio replicating a larger position, he must take more risk. He achieves this by increasing the number of shares of stock and bonds underwritten by the counterparty. Moreover, he reduces the purchases of his own bonds given that he needs to replicate a smaller residual $\DVA$ as the position becomes more collateralized (the size of the downward jump to the closeout value at the trader's default time is smaller). This behavior is evident in the plot of Figure \ref{fig:alphahc}.

\paragraph{The width of the no-arbitrage band is insensitive to bond returns. }
Figure~\ref{fig:alphahc} shows that both seller{'}s and buyers{'}s $\XVA$ decrease, if the return rate on the counterparty bond increases. When $\alpha$ is low, the two quantities drop by nearly the same amount and the width of the no-arbitrage band is unaffected. As $\alpha$ gets larger, the seller{'}s $\XVA$ decreases faster relative to the buyer{'}s $\XVA$ and the two quantities almost coincide when $\alpha=1$.

  \begin{table}[hpt]
    \centering
      \begin{tabular}{|c|c|c|}
	\hline
	$\rfm$ & {\text Seller{'}s XVA: funding} (\$) & {\text Buyer{'}s XVA: funding account} (\$)\\
	\hline
	\hline
	0.08 & -0.0124 & -0.0123 \\
	\hline
	0.1 &  -0.0125 & -0.0122 \\
	\hline
	0.15 & -0.0127 & -0.0122 \\
	\hline
	0.2 & -0.013 & -0.0122  \\
	\hline
      \end{tabular}
    \caption{The columns give the dollar position in the funding account corresponding to the replicating strategies of seller{'}s $\XVA$ and buyer{'}s XVA. We set $\mu_C = 0.16$.}
  \label{tab:Tablerf}
  \end{table}

Consistently with Figure \ref{fig:alphahc}, Figure \ref{fig:hcalpha} shows that the seller{'}s $\XVA$ decreases when the return $\mu_C$ on the counterparty bond increases. This happens
because, keeping the historical default probability constant, the trader would then earn a higher premium from his long position in counterparty{'}s bonds
(see also bottom panels of Figure \ref{fig:hcalpha}). Such a gain dominates over the funding costs incurred when replicating a larger closeout position
(Eq.~\eqref{eq:theta} indicates that the closeout payment increases to the risk-free payoff $\hat{V}$ as $h^{\Qxx}_C$ increases, and $\mu_C = h^{\Qxx}_C + r_D$). Altogether, this means that the funding costs of the investor would be reduced as $\mu_C$ increases.

  \begin{figure}[ht!]
    \centering
      \includegraphics[width=6.6cm]{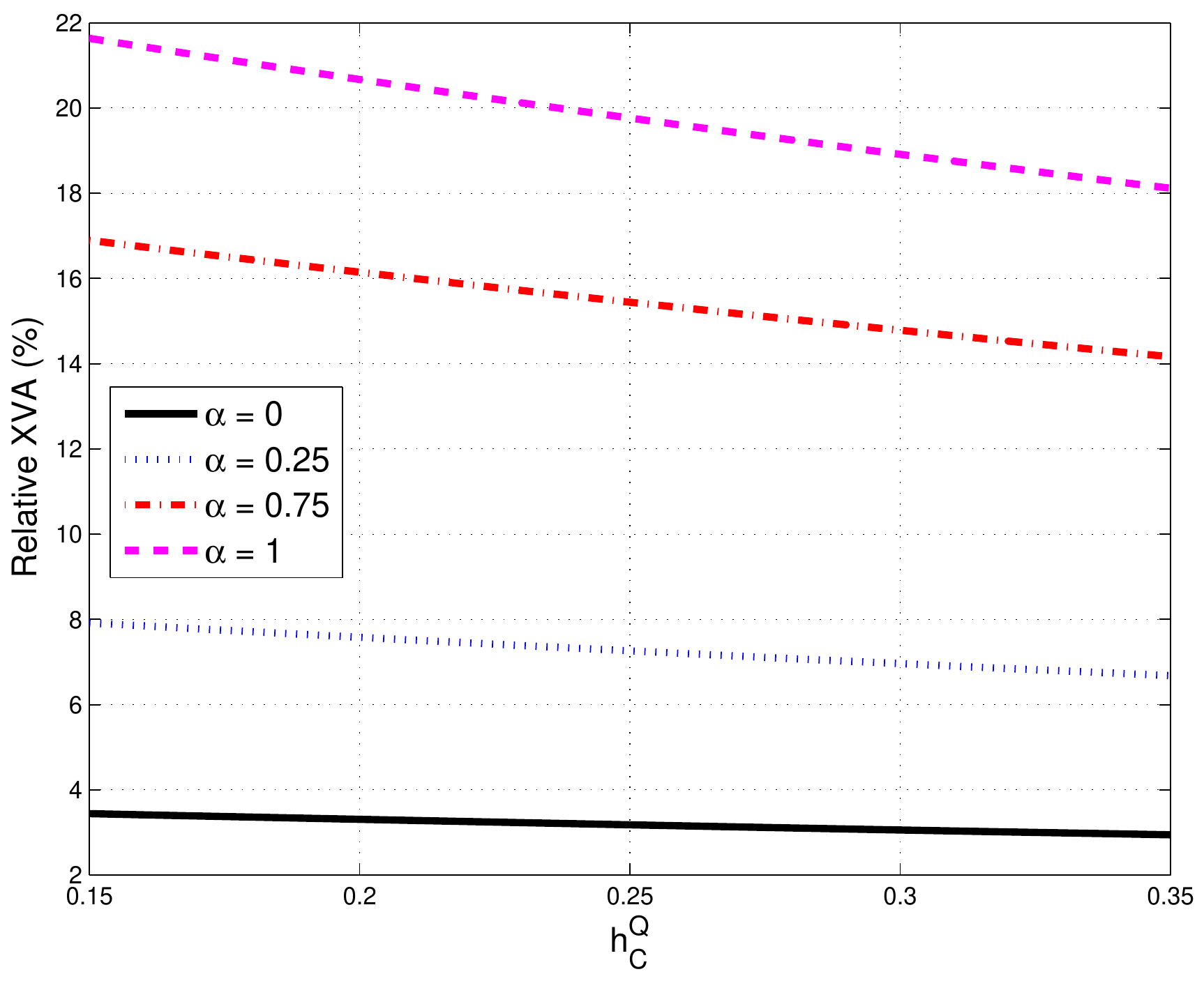}
      \includegraphics[width=6.6cm]{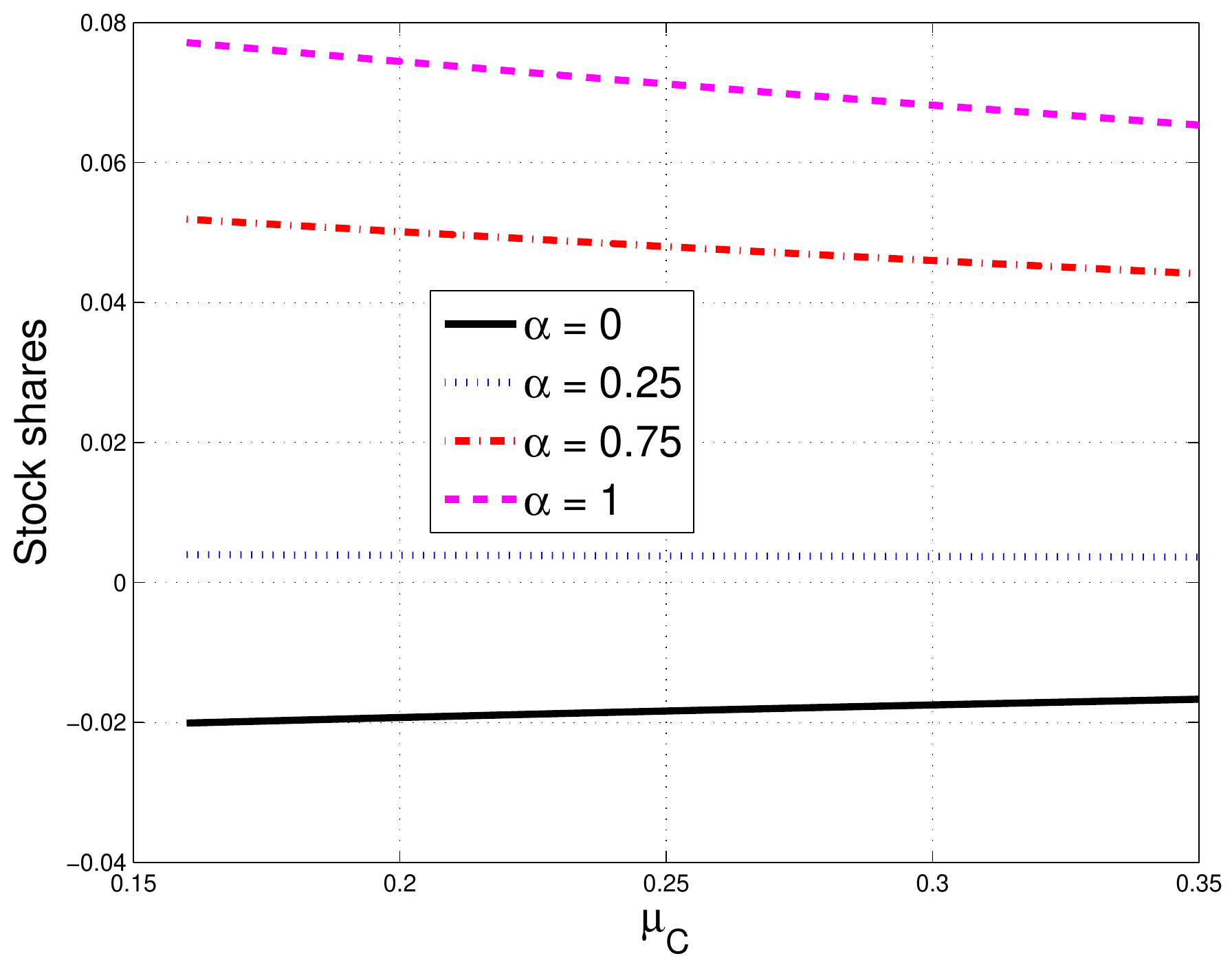}
      \includegraphics[width=6.6cm]{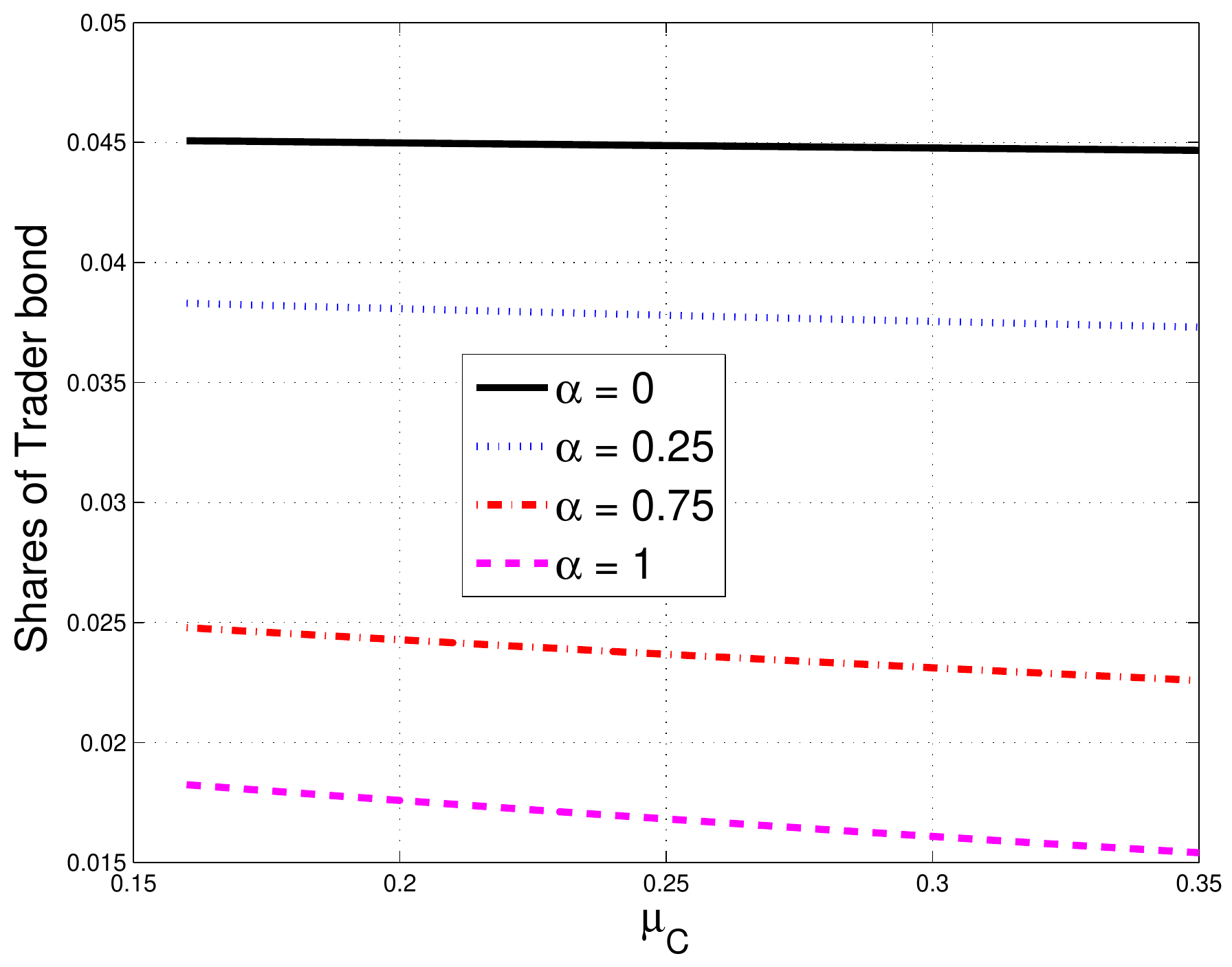}
      \includegraphics[width=6.6cm]{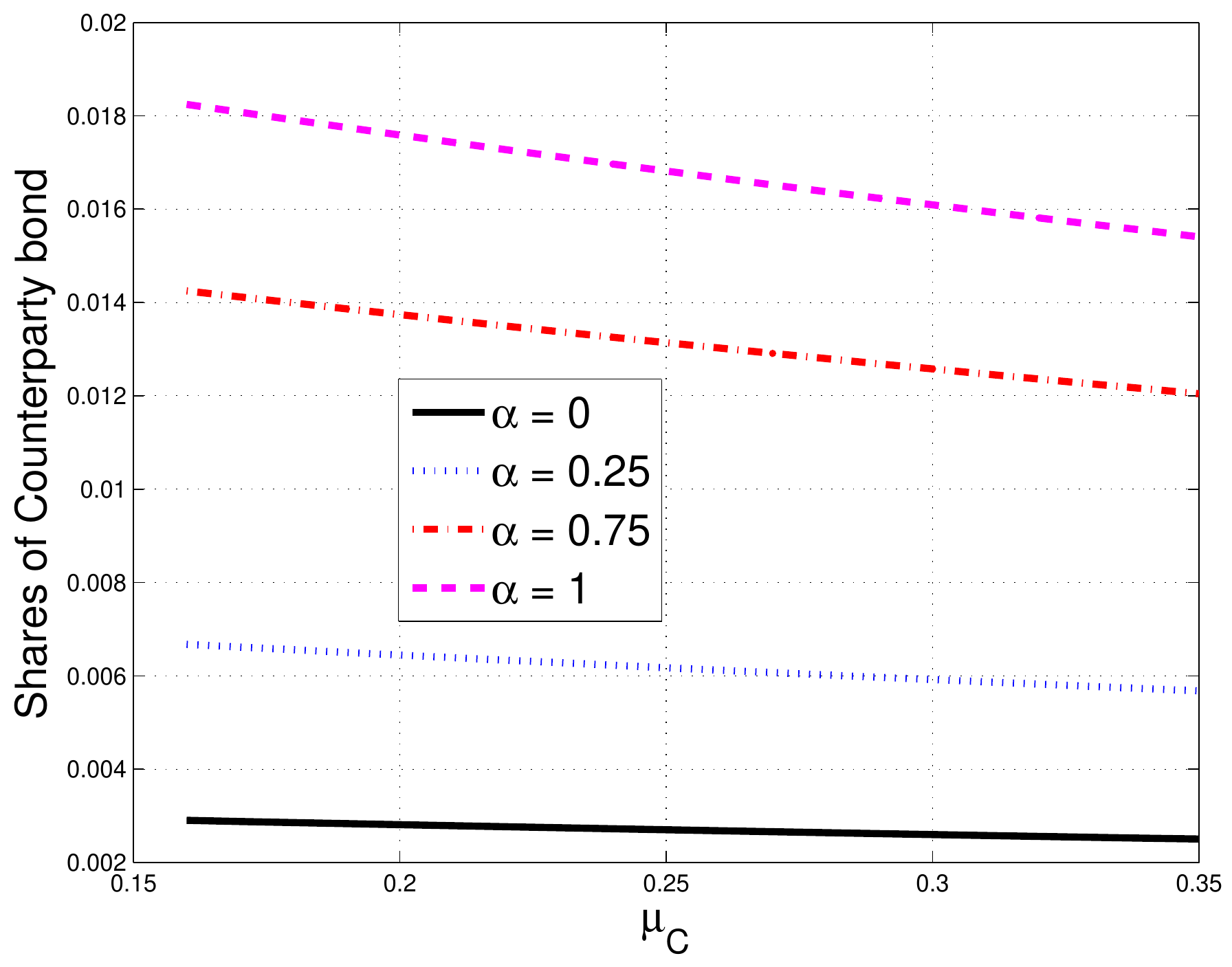}
    \caption{Top left: Seller{'}s $\XVA$ as a function of $\mu_C$ for different values of $\alpha$. Top right: Number of stock shares in the replication strategy. Bottom left: Number of trader{'}s bond shares in the replication strategy. Bottom right: Number of counterpart bond shares in the replication strategy. The replicating portfolio refers to the seller{'}s XVA.}
  \label{fig:hcalpha}
  \end{figure}

\section{Conclusions} \label{sec:conclusions}

We have developed an arbitrage-free valuation framework for the price of a European claim transacted between two risky counterparties. Our analysis takes into account funding spreads to the treasury, the repo market, collateral servicing costs, and counterparty credit risk. We have derived the no-arbitrage band associated with the valuations of buyer and seller{'}s XVA, and shown that it collapses to a unique $\XVA$ in the absence of rate asymmetries. Such a setting corresponds to a generalization of \cite{Piterbarg}{'}s model for which we are able to derive an explicit expression for the XVA.

Using the PDE representation of the BSDE, we have conducted a thorough numerical study analyzing the sensitivity of $\XVA$ and of the corresponding replication strategy to collateralization levels, default risk and spreads between borrowing and lending rates.

\section*{Acknowledgment}

The authors are grateful to two anonymous referees for valuable comments and suggestions which significantly contributed to improve the paper.

\appendix

\section{Proofs of Propositions}

\textbf{Proof of Proposition \ref{thm:PDE-sol}.}
\begin{proof}
 First, notice that by the identity $\XVA_t^{\buysell}\ind_{\{t<\tau\}} = {\check{U}_t^{\buysell}}\ind_{\{t<\tau\}} $ proven in Theorem \ref{thm:reduction} and using the nonlinear Feynman-Kac theorem (cf. \cite[Theorem 3.2]{KarHamMat}), it follows that {$\XVA_t^{\buysell}\ind_{\{t<\tau\}} = \xva^\pm(t,S_t)$}, where {$\xva^{\pm}$} are the unique viscosity solutions of
  \begin{align*}
    -\xva_t^{\pm} - r_D s \xva_s^{\pm} - \frac12\sigma^2 s^2 \xva_{ss}^{\pm} -g^\pm\bigl(t, \xva^{\pm}, \sigma s (\xva_s^{\pm}+\hat v_s\bigr) ;\hat v )& =0,\\
    \xva^{\pm}(T, s) &= 0.
  \end{align*}
In the above expression ${\hat{v} =} \hat{v}(t,s)$ denotes Black-Scholes price at time $t$ of the claim with payoff $\Phi(S_T)$ when $S_t = s$. Applying the change of variables $x = \log{(s)}$ and setting $u^{\pm}(t,x) = {\xva^{\pm}}(t,e^{x} ), {\hat w} (t,x) = \hat v(t, e^{x}), ~x\in\R$, implies the first part of the proposition.

Assume now first that $\Phi$ is continuously differentiable and $\Phi$ and $\Phi'$ are of polynomial growth, i.e.,  $\abs{\Phi(s)} \le C(1+s^n),\abs{\Phi'(s)} \le C(1+s^n)$, for all $s\in\Rplus$ for some $n>0$. Using the transformations
\begin{equation}
\uu^{\pm}(t,x) = \frac{u^{\pm}(t,x)}{1+e^{2nx}}, \quad \bar{ {\hat w} }(t,x) = \frac{ {\hat w}(t,x)}{1+e^{2nx}} , \quad \bar\Phi(x) = \frac{\Phi(e^x)}{1+e^{2nx}},\label{eq:var-transform}
\end{equation}
we note that $\uu^{\pm}$ satisfy the Cauchy problem
\begin{align}\label{eq:PDE_bndd}
&\phantom{=}-\uu_t^{\pm} -\frac{1}{2}\sigma^2\uu_{xx}^{\pm} - \biggl(\frac{\bigl((4n-1)e^{2nx}-1\bigr)\sigma^2}{2(1+e^{2nx})} + r_D\biggr)\uu_x^{\pm} - \biggl(2r_D +(2n-1)\sigma^2\biggr)\frac{ne^{2nx}}{1+e^{2nx}} \uu^{\pm} \nonumber \\
& = g^\pm \biggl(t,\uu^{\pm}, \sigma \Bigl( \uu_x^{\pm}+ \frac{2ne^{2nx}}{1+e^{2nx}}\uu^{\pm}  +  {\bar{\hat{w}}}_x + \frac{2ne^{2nx}}{1+e^{2nx}} {\bar{ \hat w}} \Bigr); {\bar{ \hat w}}\biggr), \nonumber \\
\uu^{\pm}(T,x) & = 0,
\end{align}
together with
\begin{align}
&-{\bar{\hat{w}}}_t -\frac{1}{2}\sigma^2 {\bar{\hat{w}}}_{xx} - \biggl(\frac{\bigl((4n-1)e^{2nx}-1\bigr)\sigma^2}{2(1+e^{2nx})} + r_D\biggr){\bar{\hat{w}}}_x - \biggl(2r_D +(2n-1)\sigma^2\biggr)\frac{ne^{2nx}}{1+e^{2nx}} {\bar{ \hat w}} =-r_D {\bar{ \hat w}}, \nonumber \\
&{\bar{ \hat w}}(T,x) = \bar\Phi(x). \label{eq:PDE_bndd1}
\end{align}
The above transformation guarantees that both $\bar \Phi$ and $\bar\Phi'$ are bounded. Then the existence of a smooth (and bounded) solution to \eqref{eq:PDE_bndd},  \eqref{eq:PDE_bndd1} follows from Theorem 20.2.1 in \cite{Cannon}. In the case that $\bar\Phi$ is only piecewise smooth, the original proof can be modified following a similar procedure to \cite{JK}. Hence, using the change of variables \eqref{eq:var-transform}, we conclude that there exists a classical solution to the Cauchy problem \eqref{eq:PDE_representation}.
\end{proof}

\noindent \textbf{Proof of Proposition \ref{prop:Pitnodef}.}

\begin{proof}
In the absence of defaults, the BSDE~\eqref{eq:XVABSDE} is given by
\begin{align} \label{BDESredp}
-d\XVA_t& = -r_f \XVA_t dt + (r_f - r_c ) \alpha \hat{V}_t dt + (r_r - r_f ) \hat{V}_t dt - \check{Z}_t\, dW_t^{\Qxx} \nonumber  \\
\XVA_T &= 0,
\end{align}
where we have omitted the superscript $\buysell$ given that seller's and buyer's $\XVA$ coincide due to the symmetry of rates. Moreover, we have used the collateral specification given in Eq.~\eqref{eq:rulecoll}, and the assumption that $r_D = r_r$. The above BSDE admits the following integral representation
  \begin{align}\label{PitebargBSDE}
   e^{-r_f t} \XVA_t &= -\int_t^T (B_s^{r_f})^{-1} \check{Z}_s \, dW_s^{\Qxx} + (r_f - r_c) \alpha \int_t^T (B_s^{r_f})^{-1} \hat{V}_s ds + (r_r - r_f) \int_t^T (B_s^{r_f})^{-1} \hat{V}_s ds.
  \end{align}
Using the Clark-Ocone formula, we can find $\check{Z}_t$ by means of Malliavin Calculus (see \cite{Nualartbook} for an introduction to Malliavin derivatives). We have that
  \[
    e^{-r_f t} \check{Z}_t = \mathbb{E}^{\Qxx} \biggl[\Bigl. D_t \Bigl(\int_t^T \bigl(B_s^{r_f}\bigr)^{-1}  \bigl(\alpha \bigl(r_f - r_c\bigr) + (r_r - r_f) \bigr) \hat{V}_s \,ds \Bigr) \Bigr\vert \mathcal{F}_t\biggr],
  \]
where $D_t$ denotes the Malliavin derivative which may be computed as
  \begin{align}\label{eq:malliavin}
    & \phantom{==} D_t \biggl(\int_t^T \bigl(B_s^{r_f}\bigr)^{-1}  \bigl(\alpha \bigl(r_f - r_c\bigr) + (r_r - r_f) \bigr) \hat{V}_s \,ds \biggr) \nonumber\\
    &=  \int_t^T \bigl(B_s^{r_f}\bigr)^{-1} \bigl( \alpha (r_f - r_c) + (r_r - r_f) \bigr) D_t \hat{V}_s \,ds \nonumber\\
    & = \int_t^T \bigl(B_s^{r_f}\bigr)^{-1} \bigl( \alpha (r_f - r_c) + (r_r - r_f) \bigr) \frac{\partial}{\partial S}\hat{V}(s,S_s)\sigma S_s \, ds.
  \end{align}
Above, we have used the chain rule of Malliavin calculus and the well known fact that $D_t S_s = \sigma S_s$ for $s>t$. We expect that the $\check{Z}$ term of the BSDE would correspond to an ``adjusted delta hedging'' strategy, with the delta hedging strategy recovered if all rates are identical. Indeed, using the definition of $\hat{\Delta}$ given in Eq.~\eqref{eq:delta}, the Malliavin derivative in Eq.~\eqref{eq:malliavin} may be written in terms of $\Delta$ as
  \[
    D_t\bigl(...\bigr) =  \bigl(\alpha(r_f - r_c) + (r_r - r_f) \bigr) \int_t^T \frac{1}{B_s^{r_f}}  \hat{\Delta}_s \sigma S_s \, ds.
  \]
Therefore, we get
  \begin{align}
    \nonumber e^{-r_f t} \check{Z}_t & = \bigl( \alpha (r_f - r_c) + (r_r - r_f) \bigr) \int_t^T \frac{1}{B_s^{r_f}}  \sigma \mathbb{E}^{\Qxx}\bigl[\bigl.\hat{\Delta}_s S_s\bigr\vert \mathcal{F}_t\bigr] \, ds \\
    \nonumber & =  \bigl( \alpha (r_f - r_c) + (r_r - r_f) \bigr) \int_t^T \frac{1}{B_s^{r_f}} \sigma B_s^{r_r} \mathbb{E}^{\Qxx}\Bigl[\Bigl.\hat{\Delta}_s \frac{S_s}{B_s^{r_r}}\Bigr\vert \mathcal{F}_t\Bigr] \, ds\\
    \nonumber & =  \bigl( \alpha (r_f - r_c) + (r_r - r_f) \bigr) \int_t^T \frac{B_s^{r_r}}{B_s^{r_f}} \sigma \hat{\Delta}_t \frac{S_t}{B_t^{r_r}} \, ds\\
    & =   \bigl( \alpha (r_f-r_c) +  (r_r - r_f ) \bigr) \frac{\sigma S_t}{r_r-r_f} \frac{1}{B_t^{r_r}}
    \biggl(\frac{B_T^{r_r}}{B_T^{r_f}} - \frac{B_t^{r_f}}{B_t^{r_f}}\biggr) \hat{\Delta}_t, \label{eq:finalz}
  \end{align}
  where we have used the martingale property $\mathbb{E}^{\Qxx}\bigl[\Bigl.\hat{\Delta}_s \frac{S_s}{B_s^{r_r}}\bigr\vert \mathcal{F}_t\bigr] = \frac{S_t}{B_t^{r_r}} \hat{\Delta}_t$. Indeed, from Eq.~\eqref{eq:delta} and using the fact that $S_T = \frac{B_T^{r_r}}{B_t^{r_r}}S_t e^{ -\frac{\sigma^2}{2}(T-t) {+ \sigma} (W_T^{\Qxx} -W_t^{\Qxx} )}$ (which follows from \eqref{eq:S-Q}), we conclude that
\begin{align*}
\frac{S_t}{B_t^{r_r}}\hat{\Delta}_t &=  \frac{S_t}{B_t^{r_r}} \frac{\partial}{\partial S} \mathbb{E}^{\Qxx}\biggl[\frac{B_t^{r_r}}{B_T^{r_r}}  \Phi(S_T) \bigg\vert \mathcal{F}_t\biggr]= \frac{S_t}{B_t^{r_r}}  \mathbb{E}^{\Qxx}\biggl[  \Phi'(S_T)  e^{ -\frac{\sigma^2}{2}(T-t)  + \sigma (W_T^{\Qxx} -W_t^{\Qxx} )} \bigg\vert \mathcal{F}_t\biggr]\\
&= \mathbb{E}^{\Qxx}\biggl[  \Phi'(S_T)  \frac{S_T}{B_T^{r_r}} \bigg\vert \mathcal{F}_t\biggr] = \mathbb{E}^{\Qxx}\biggl[  \hat{\Delta}_T   \frac{S_T}{B_T^{r_r}} \bigg\vert \mathcal{F}_t\biggr],
\end{align*}
where we have interchanged derivative and expectation by differentiating under the integral sign. This is well justified given that we are computing the expectation of a smooth function of a Gaussian random variable.

We note that {$\check{Z}$} is square integrable and therefore the stochastic integral in \eqref{PitebargBSDE} is a true martingale. Using this fact in the integral representation~\eqref{PitebargBSDE} and taking conditional expectations, we can provide an explicit solution for the BSDE as follows:
\begin{align*}
\XVA_t& = \bigl(\alpha (r_f - r_c) +  (r_r - r_f) \bigr) \int_t^T \frac{B_t^{r_f}}{B_u^{r_f}} \mathbb{E}^{\Qxx} \bigl[\hat{V}_u \big| \mathcal{F}_t \bigr] du \\
&= \bigl(\alpha (r_f - r_c) +  (r_r - r_f) \bigr) B_t^{r_f} \int_t^T e^{-(r_f-r_r) u} \mathbb{E}^{\Qxx} \bigl[\bigl({B_u^{r_r}}\bigr)^{-1} \hat{V}_u \big| \mathcal{F}_t\bigr] du \\
  & = \biggl( {1 - \alpha \frac{r_f - r_c}{r_f - r_r}}  \biggr)  \frac{B_t^{r_f}}{B_t^{r_r}} \Biggl(\frac{B_T^{r_r}}{B_T^{r_f}} -\frac{B_t^{r_r}}{B_t^{r_f}} \Biggr) \hat{V}_t,
\end{align*}
where in the last step we have used the martingale property of the discounted payoff. This correspond with Eq.~\eqref{eq:FVAPiterbarg} after straightforward adjustments. Finally, Eq.~\eqref{eq:unitedDeltaPiterbarg} follows from \eqref{eq:finalz} together with the first identity in \eqref{eq:Zetas}.
\end{proof}

\noindent \textbf{Proof of Proposition \ref{prop:Pitdef}.}

\begin{proof}
The proof follows a similar route to that of Proposition~\ref{prop:Pitnodef}. In the presence of defaults and when the rates are symmetric, the reduced BSDE for XVA~\eqref{eq:reduced} becomes
\begin{align}\label{eq:reduced1}
-d{\check{U}_t}& = \biggl((r_f - r_c) \alpha \hat{V}_t + (r_r - r_f) \hat{V}_t + \sum_{j\in\{I,C\}}( {\mu_j - r_f} )\tilde{\theta} _{j}(\hat{V}_t) \biggr) \, dt  - \eta {\check{U}_t} \, dt -\check{Z}^{\pm}_t\, dW_t^{\Qxx}, \nonumber \\
{\check{U}_T} &= 0.
\end{align}
The above BSDE admits the following integral representation:
\begin{align}\label{eq:integr-sol1}
e^{-\eta t} {\check{U}_t}  &=  - \int_{t}^T e^{-\eta s} \check{Z}_s dW_s^{\Qxx} + \int_t^T  (r_f - r_c) \alpha e^{-\eta s} \hat{V}_s ds + \int_t^T (r_r - r_f) e^{-\eta s} \hat{V}_s ds
\nonumber \\
& \phantom{=} + \sum_{j\in\{I,C\}}\bigl( \mu_j - r_f \bigr) \int_t^T e^{-\eta s} \tilde{\theta} _{j}(\hat{V}_s) ds  - {\eta} \int_t^T e^{-\eta s} ds.
\end{align}
Using the Clark-Ocone formula, we can find $\check{Z}_t$ by means of Malliavin Calculus. We have that
\begin{align*}
e^{-\eta t} \check{Z}_t &= \mathbb{E}^{\Qxx} \biggl[\Bigl.  \int_t^T e^{-\eta s}  \bigl(\alpha (r_f - r_c) + (r_r - r_f) \bigr) D_t \hat{V}_s \,ds  \Bigr\vert \mathcal{F}_t\biggr] \\
& \phantom{=} + \mathbb{E}^{\Qxx} \biggl[\sum_{j\in\{I,C\}}( {\mu_j - r_f} ) \int_t^T e^{-\eta s} D_t \tilde{\theta} _{j}(\hat{V}_s) {\, ds} \biggr].
\end{align*}
It holds that $D_t \hat{V}_s = \hat{\Delta}_s \sigma S_s$. Using Proposition 1.2.4 in \cite{Nualartbook} and Eq.~\eqref{eq:hats}, we obtain
\[
D_t  \tilde{\theta} _{C}(\hat{V}_s) = L_C (1-\alpha) D_t \bigl(\hat{V}_s\bigr)^-  =  L_C (1-\alpha) {\ind_{\{\hat{V}_s < 0\}}} \frac{\partial}{\partial S} \hat{V}(s,S_s) \sigma S_s,
\]
and
\[
D_t \tilde{\theta} _{I}(\hat{V}_s)  = -L_I (1-\alpha) D_t \bigl(\hat{V}_s\bigr)^+  =  -L_I (1-\alpha) \ind_{\{\hat{V}_s > 0\}} \frac{\partial}{\partial S} \hat{V}(s,S_s) \sigma S_s.
\]
From that, we obtain the following equality
\begin{align}\label{eq:chech-Z}
    e^{-\eta t} \check{Z}_t &= \bigl(\alpha (r_f - r_c) + (r_r-r_f) \bigr) \sigma \int_t^T e^{-\eta s} \mathbb{E}^{\Qxx} \bigl[\hat{\Delta}_s S_s  \big| \mathcal{F}_t \bigr] \, ds \nonumber\\
    & \phantom{=} + ( {\mu_C - r_f} ) L_C (1-\alpha) \int_t^T e^{-\eta s} \mathbb{E}^{\Qxx} \bigl[\hat{\Delta}_s {\ind_{\{\hat{V}_s < 0\}}} \sigma S_s  \big| \mathcal{F}_t \bigr] {\, ds} \nonumber\\
    & \phantom{=} - ( {\mu_I - r_f}  ) L_I (1-\alpha) \int_t^T e^{-\eta s}  \mathbb{E}^{\Qxx} \bigl[\hat{\Delta}_s  {\ind_{\{\hat{V}_s > 0\}}} \sigma S_s  \big| \mathcal{F}_t \bigr] {ds}\nonumber\\
    &= \bigl(\alpha (r_f - r_c) + (r_r-r_f) \bigr) \sigma \int_t^T e^{-\eta s} B_s^{r_r}  \mathbb{E}^{\Qxx} \biggl[\hat{\Delta}_s \frac{S_s}{B_s^{r_r}} \bigg| \mathcal{F}_t \biggr] ds \nonumber\\
    & \phantom{=} + ( {\mu_C - r_f} ) L_C (1-\alpha) \sigma \int_t^T e^{-\eta s} B_s^{r_r} \mathbb{E}^{\Qxx} \biggl[\hat{\Delta}_s {\ind_{\{\hat{V}_s < 0\}}} \frac{S_s}{B_s^{r_r}} \bigg| \mathcal{F}_t \biggr] {\, ds} \nonumber\\
    & \phantom{=} - ( {\mu_I - r_f}  ) L_I (1-\alpha) \sigma \int_t^T e^{-\eta s} B_s^{r_r} \mathbb{E}^{\Qxx} \bigg[\hat{\Delta}_s {\ind_{\{\hat{V}_s > 0\}}} \sigma \frac{S_s}{B_s^{r_r}} \bigg| \mathcal{F}_t \biggr] {\, ds} \nonumber\\
    &= \bigl(\alpha (r_f - r_c) + (r_r-r_f) \bigr) \sigma \frac{S_t}{B_t^{r_r}} \frac{1}{r_r - \eta} \biggl(\frac{B_T^{r_r}}{e^{\eta T}} - \frac{B_t^{r_r}}{e^{\eta t}} \biggr) \hat{\Delta}_t \nonumber\\
    & \phantom{=} + ( {\mu_C - r_f}  ) L_C (1-\alpha) \sigma \frac{S_t}{B_t^{r_r}} \frac{1}{r_r - \eta}
    \biggl(\frac{B_T^{r_r}}{e^{\eta T}} - \frac{B_t^{r_r}}{e^{\eta t}} \biggr) \hat{\Delta}_t \nonumber\\
    & \phantom{=} - ( {\mu_I - r_f} ) L_I (1-\alpha) \sigma \frac{S_t}{B_t^{r_r}} \frac{1}{r_r - \eta}
    \biggl(\frac{B_T^{r_r}}{e^{\eta T}} - \frac{B_t^{r_r}}{e^{\eta t}} \biggr) \hat{\Delta}_t.
\end{align}
The last step is justified by the fact that $\eta > r_r$. We note that {$\check{Z}$} is square integrable and therefore the stochastic integral in \eqref{PitebargBSDE} is a true martingale. Using this fact in the integral
representation~\eqref{eq:reduced1}, and taking the conditional expectation, it follows that
\begin{align}\label{eq:finalzdef}
e^{-\eta t} {\check{U}_t} &=  \int_t^T  (r_f - r_c) \alpha e^{-\eta s} \mathbb{E} \big[\hat{V}_s  \big\vert \mathcal{F}_t \big] ds + \int_t^T (r_r - r_f) e^{-\eta s} \mathbb{E} \bigl[\hat{V}_s \big\vert \mathcal{F}_t \bigr] ds \nonumber \\
& \phantom{=} +\int_t^T e^{-\eta s} \Bigl( ( {\mu_C - r_f} )   L_C \mathbb{E}^{\Qxx} \bigl[((1-\alpha)\hat{V}_s )^{-}  \big\vert  \mathcal{F}_t \big] + (  {\mu_I - r_f} ) L_I \mathbb{E}^{\Qxx} \bigl[\bigl((1-\alpha) \hat{V}_s \bigr)^{+} \big\vert  \mathcal{F}_t \big] \Bigr) ds \nonumber\\
&=  \bigl((r_f - r_c) \alpha + (r_r - r_f)\bigr) \int_t^T  e^{-\eta s} B_s^{r_r} \mathbb{E}^{\Qxx} \bigl[  \bigl(B_s^{r_r}\bigr)^{-1} \hat{V}_s \big\vert \mathcal F_t \bigr]  ds \nonumber \\
& \phantom{=}+ ({\mu_C - r_f} ) L_C \int_t^T e^{-\eta s}  B_s^{r_r}\mathbb{E}^{\Qxx} \bigl[ \bigl(B_s^{r_r}\bigr)^{-1} ((1-\alpha) \hat{V}_s )^{-} \big\vert \mathcal F_t \bigr]  ds \nonumber\\
& \phantom{=}-  ({\mu_I - r_f}  )  L_I \int_t^T e^{-\eta s} B_s^{r_r}\mathbb{E}^{\Qxx} \bigl[ \bigl(B_s^{r_r}\bigr)^{-1} ((1-\alpha) \hat{V}_s )^{+} \big\vert \mathcal F_t \bigr] ds \nonumber\\
&=  \bigl((r_f - r_c) \alpha + (r_r - r_f)\bigr) \int_t^T  e^{-\eta s} B_s^{r_r} \bigl(B_t^{r_r}\bigr)^{-1}  \hat{V}_t \nonumber \\
& \phantom{=}+({\mu_C - r_f}  ) L_C \int_t^T e^{-\eta s}  B_s^{r_r} \bigl(B_t^{r_r}\bigr)^{-1}  \bigl((1-\alpha) \hat{V}_t \bigr)^{-} \nonumber\\
& \phantom{=} - ({\mu_I - r_f} )  L_I \int_t^T e^{-\eta s} B_s^{r_r}\bigl(B_t^{r_r}\bigr)^{-1}  \bigl((1-\alpha) \hat{V}_t \bigr)^{+},
\end{align}
where we have used the martingale properties of the discounted payoffs. It thus follows that
\begin{align}
{\check{U}_t} &= (r_r -r_f) \frac{1- e^{-(\eta-r_r) (T-t)} }{ \eta-r_r} \hat{V}_t + \alpha (r_f-r_c) \frac{1- e^{-(\eta-r_r) (T-t)}}{\eta-r_r} \hat{V}_t \nonumber\\
&\phantom{=}+ (  {\mu_C - r_f}  ) L_C  \frac{1- e^{-(\eta-r_r) (T-t)} }{ \eta-r_r} (1-\alpha) \bigl(\hat{V}_t \bigr)^{-} \nonumber\\
& \phantom{=} - ( {\mu_I - r_f}  )   L_I \frac{1- e^{-(\eta-r_r) (T-t)} }{ \eta-r_r}  (1-\alpha) \bigl(\hat{V}_t \bigr)^{+},
\end{align}
which yields by~\eqref{eq:reduced_identity2} and multiplying with the indicator $\ind_{\{\tau>t\}}$ Eq.~\eqref{eq:FVAPiterbargdef}. We next compute the hedging strategy $\tilde{\xi}$ in the stock using the martingale representation theorem. Consider the stock replicating strategy $\tilde{\xi}$. Then the investment in stock has the dynamics
	\[
    {\tilde{\xi}_t} dS_t = \tilde{\xi}_t \mu S_t dt + {\tilde{\xi}_t} \sigma S_t dW_t.
	\]
	By the martingale representation theorem in the $\mathbb{H}$-filtration one can split every semimartingale uniquely into an absolutely continuous part, a Brownian martingale and two jump martingales. It follows that $\tilde{\xi}_t \sigma S_t dW_t = - \tilde{Z}_t dW_t$. By the uniqueness of the martingale representation, it follows that $\tilde{\xi}_t \sigma S_t = - \tilde{Z}_t$ and thus the claimed result. An analogous argument applies for the bond strategies. Finally, Eq.~\eqref{eq:defaul-I} follows directly from the expression for $\XVA_t$ given in Eq.~\eqref{eq:reduced_identity2}.
\end{proof}

\end{document}